\documentclass[12pt]{article}
\usepackage[english]{babel}
\usepackage[utf8]{inputenc}
\usepackage[T1]{fontenc}
\usepackage{listings}
\usepackage{amsbsy}
\usepackage{amsfonts}
\usepackage{amsmath}
\usepackage{amssymb}
\usepackage{amsthm}
\usepackage{graphicx} 
\usepackage[pdftex]{color}
\usepackage{dsfont}
\usepackage{enumerate}
\usepackage{float}
\usepackage{calc, color, setspace}
\usepackage[top=0.48in, bottom=0.68in, left=0.60in, right=0.60in]{geometry}

\usepackage{indentfirst}
\usepackage{longtable}
\usepackage{multirow}
\usepackage{natbib}
\usepackage{xr-hyper}
\usepackage{hyperref,bookmark}
\usepackage{mathtools}
\DeclarePairedDelimiter\floor{\lfloor}{\rfloor}

\usepackage{listings}
\usepackage{fancyvrb}
\usepackage{tcolorbox}

\newtheorem{theorem}{Theorem}[section]
\newtheorem{definition}{Definition}[section]
\newtheorem{lemma}[theorem]{Lemma}
\newtheorem{proposition}[theorem]{Proposition}

\newtheorem{remark}[theorem]{Remark}

\hypersetup{
  colorlinks=true,
  linkcolor=red,
  citecolor=blue,
  filecolor=red,
  urlcolor=teal,
}
\newcommand{\blind}{1}

\begin{document}

\def\spacingset#1{\renewcommand{\baselinestretch}%
{#1}\small\normalsize} \spacingset{1.4}

%\spacingset{1.45} 

\if1\blind
{
  \title{\bf Nearly Unstable Integer-Valued ARCH Process and Unit Root Testing
  }
  \author{Wagner Barreto-Souza$^\star$\footnote{E-mail: \textcolor{teal}{\texttt{wagner.barretosouza@kaust.edu.sa}} (Corresponding Author)}\,\,\, and\, Ngai Hang Chan$^\sharp$\footnote{E-mail: \textcolor{teal}{\texttt{nhchan@sta.cuhk.edu.hk}}}\hspace{.2cm}\\
    {\it \normalsize $^\star$Statistics Program, King Abdullah University of Science and Technology, Thuwal, Saudi Arabia}\\
    {\it \normalsize $^\sharp$Department of Statistics, The Chinese University of Hong Kong, Hong Kong}}
  \maketitle
} \fi
\if0\blind
{
  \bigskip
  \bigskip
  \bigskip
  \begin{center}
    {\LARGE\bf Nearly Unstable INARCH Process}
\end{center}
  \medskip
} \fi

\bigskip
\addtocontents{toc}{\protect\setcounter{tocdepth}{1}}

\begin{abstract}
This paper introduces a Nearly Unstable INteger-valued AutoRegressive Conditional Heteroskedasticity (NU-INARCH) process for dealing with count time series data. It is proved that a proper normalization of the NU-INARCH process endowed with a Skorohod topology weakly converges to a Cox-Ingersoll-Ross diffusion. The asymptotic distribution of the conditional least squares estimator of the correlation parameter is established as a functional of certain stochastic integrals. Numerical experiments based on Monte Carlo simulations are provided to verify the behavior of the asymptotic distribution under finite samples. These simulations reveal that the nearly unstable approach provides satisfactory and better results than those based on the stationarity assumption even when the true process is not that close to non-stationarity. A unit root test is proposed and its Type-I error and power are examined via Monte Carlo simulations. As an illustration, the proposed methodology is applied to the daily number of deaths due to COVID-19 in the United Kingdom.
\end{abstract}

{\it \textbf{Keywords}:} Count time series; Cox-Ingersoll-Ross diffusion process; Inference; Limit theorems; Stochastic integral.

\section{Introduction}\label{intro}
  
  First-order nearly unstable continuous autoregressive processes have been well explored in the literature, see for example \cite{chanwei1987}, \cite{phi1987}, \cite{chanetal2019}, and the references therein. In these works, it is assumed that the model approaches the non-stationarity region as the sample size increases. More specifically, a nearly unstable continuous process $\{Y_t^{(n)}\}_{t\in\mathbb N}$ is defined by
  \begin{eqnarray*}
  	Y_t^{(n)}=\rho_nY_{t-1}^{(n)}+\eta_t,\quad t\in\mathbb N,
  \end{eqnarray*}	 
  where $\{\eta_t\}_{t\in\mathbb N}$ is a white noise and $\rho_n=1-b/n$, for $b>0$.
  
  In the past few years, nearly unstable discrete processes have emerged
  based on the INteger-valued AutoRegressive (INAR) approach \citep{mck1985,alal1987}. The first attempt on this subject was due to \cite{ispetal2003}. More specifically, a nearly unstable INAR(1) process $\{X_t\}_{t\in\mathbb N}$ is defined by
  \begin{eqnarray*}
  X^{(n)}_t=\alpha_n\circ X^{(n)}_{t-1}+\epsilon^{(n)}_t, \quad t\in\mathbb N,
  \end{eqnarray*}	
  where $\circ$ is the thinning operator proposed by \cite{stevan1979}, given by $\alpha_n\circ X^{(n)}_{t-1}=\sum_{j=0}^{X^{(n)}_{t-1}} B^{(n)}_{j\,t}$ with $\{B^{(n)}_{j\,t}\}_{j,t\in\mathbb N}\stackrel{iid}{\sim}\mbox{Bernoulli}(\alpha_n)$, for $\alpha_n\in(0,1)$, and $\{\epsilon^{(n)}_t\}_{t\in\mathbb N}$ is a sequence of independent and identically distributed (iid) random variables with $\epsilon^{(n)}_t$ being independent of the counting series $\{B^{(n)}_{j\,k}\}_{j\in\mathbb N}$ for all $k\leq t$, for $t\in\mathbb N$. These authors assumed that $\alpha_n$ approaches 1 (non-stationarity) when $n\rightarrow\infty$ as given in \cite{chanwei1987} in the continuous context. By assuming $\mu_\epsilon\equiv E(\epsilon_t)$ is known, the conditional least squares (CLS) estimator of $\alpha_n$ was explored by \cite{ispetal2003}. They showed that, under nearly non-stationarity and assuming finite second moment for $\epsilon_t$, the CLS estimator weakly converges to a normal distribution at the rate $n^{3/2}$. Other related works dealing with nearly unstable INAR (Galton-Watson/branching) processes are due to \cite{weiwin1990}, \cite{win1991}, \cite{ispetal2005}, \cite{rah2007}, \cite{rah2008}, \cite{droetal2009}, \cite{rah2009}, \cite{barcetal2011}, \cite{ispetal2014}, \cite{barcetal2014}, \cite {guozha2014}, and \cite{barcetal2016}. Practical situations demonstrating evidence of a nearly unstable INAR model are discussed for instance by \cite{hel2001}. 
  
  Another popular way for dealing with count time series data is the INteger-valued Genenalized AutoRegressive Conditional Heterokedastic (INGARCH) models by \cite{ferlandetal2006}, \cite{foketal2009}, \cite{fokfri2010}, \cite{zhu2011},  \cite{foktjo2011}, \cite{zhu2012}, \cite{chrfok2015}, \cite{gonetal2015}, \cite{davliu2016}, \cite{silbar2019}, \cite{weietal2020}, which constitute in some sense an integer-valued counterpart of the classical GARCH models by \cite{bollerslev1986}. The INGARCH methodology is the focus of this paper. Like the existing literature on nearly unstable continuous and INAR processes that assumes first-order autoregressive dependence, in this paper we consider the first-order autoregressive version of the INGARCH approach, which is known as INARCH(1) (INteger-valued AutoRegressive Conditional Heteroskedasticity).

Our chief goal in this paper is to introduce a Nearly Unstable INARCH (denoted by NU-INARCH) process for dealing with count time series data. To the best of our knowledge, this is the first time that a nearly unstable count time series model is being proposed based on the INARCH approach; all existing nearly unstable discrete processes in the literature consider the INAR approach. We establish the weak convergence of the NU-INARCH process (when properly normalized) endowed with a Skorohod topology. With this result at hand, we derive the asymptotic distribution of the conditional least squares estimator of the correlation parameter as a functional of certain stochastic integrals. An equally important contribution of this paper is to develop a unit root test (URT) for the INARCH model, where the asymptotic distribution of the statistics under the null hypothesis is provided. Note that although URTs are well explored in the continuous case, only sporadic results are available for the discrete case. A few works dealing with this relevant problem, based on the INAR approach, are due to \cite{hel2001}  and \cite{droetal2009}.

The paper is organized as follows. In Section \ref{sec:model_fluct}, the NU-INARCH model is introduced and a fluctuation theorem is established, which involves the Cox-Ingersoll-Ross diffusion process. The asymptotic distribution of the CLS estimator for the correlation parameter is derived in Section \ref{sec:cls} under the nearly unstable and stationarity assumptions. Section \ref{sec:sim} provides simulated results about the asymptotic distribution of the CLS estimator under both nearly unstable and stationary approaches and also compares them in terms of confidence interval coverages. A unit root test for the INARCH process is proposed in Section \ref{sec:urt} and its performance is evaluated via Monte Carlo simulations. An empirical application about the daily number of deaths due to COVID-19 in the United Kingdom, which exhibits a nearly unstable/non-stationary behavior, is provided in Section \ref{sec:application}. Concluding remarks and future research are addressed in Section \ref{conclusions}.
  
  \section{Model and the Fluctuation Theorem}\label{sec:model_fluct}

In this section, we define the nearly unstable INARCH process and obtain its weak convergence (under a proper normalization) in the space of the non-negative c{\`a}dl{\`a}g functions endowed with the Skorokhod topology.  
  
  \begin{definition}\label{def:n_unst}
  	We say that a sequence $\{X^{(n)}_t\}_{t\in\mathbb N}$ is a first-order nearly unstable integer-valued ARCH process (in short NU-INARCH) if 
  	\begin{eqnarray}
  	X^{(n)}_t|\mathcal F^{(n)}_{t-1}&\sim&\mbox{Poisson}(\lambda^{(n)}_t),\label{eq1}\\
  	\lambda^{(n)}_t\equiv E(X^{(n)}_t|\mathcal F^{(n)}_{t-1})&=&\beta+\alpha_n X^{(n)}_{t-1},\quad t\geq1,\label{eq2}
  	\end{eqnarray}	 
  	for $n\in\mathbb N$, where $\mathcal F^{(n)}_{t-1}=\sigma\{X^{(n)}_{t-1},\ldots,X^{(n)}_0\}$, $\beta>0$, and $\alpha_n=1-\dfrac{\gamma_n}{n}$,           
	with $\displaystyle\lim_{n\rightarrow\infty}\gamma_n=\gamma>0$, and $X_0^{(n)}=\kappa\in\mathbb N$ (constant starting value). 
  \end{definition}
  
  \begin{remark}
  	For the nearly unstable INARCH model defined above, we have that $\mbox{corr}(X^{(n)}_t,X^{(n)}_{t-k})=\alpha_n^k$, for $k\in\mathbb N$. The parameterization of $\alpha_n$ in (\ref{eq2}) was first proposed by \cite{chanwei1987} and subsequently used in \cite{ispetal2003}.
  \end{remark}	

In the next proposition, we provide the mean, variance, and autocorrelation function of the NU-INARCH process. These results will be important to establish the proper normalization in order to obtain a non-trivial limit for the counting process.

\begin{proposition}\label{prop:moments}
 Let $\{X^{(n)}_t\}_{t\in\mathbb N}$ be a nearly unstable INARCH process. Then, its marginal mean and variance, and autocorrelation function are given respectively by
 \begin{eqnarray*}
&& E(X^{(n)}_t)=\beta\dfrac{1-\alpha_n^t}{1-\alpha_n},\\
&& \mbox{Var}(X^{(n)}_t)=\dfrac{\beta}{1-\alpha_n}\left\{\dfrac{1-\alpha_n^{2t}}{1-\alpha_n^2}-\alpha_n^t\dfrac{1-\alpha_n^{t}}{1-\alpha_n}\right\},\\
&& \mbox{cov}(X^{(n)}_{t+k},X^{(n)}_t)=\alpha_n^k\mbox{Var}(X^{(n)}_t), \quad t,k\in\mathbb N_0\equiv\{0,1,2,\ldots\}.
 \end{eqnarray*}
\end{proposition}  

\begin{proof}
	We have that $E(X^{(n)}_t)=E\left(E(X^{(n)}_t|\mathcal F^{(n)}_{t-1})\right)=\beta+\alpha_n E(X^{(n)}_{t-1})$. By using recursion $t$ times, we obtain the result for the marginal mean. For the variance, it follows that 
	\begin{eqnarray*}
		&&	\mbox{Var}(X^{(n)}_t)=E\left(\mbox{Var}(X^{(n)}_t|\mathcal F^{(n)}_{t-1})\right)+\mbox{Var}\left(E(X^{(n)}_t|\mathcal F^{(n)}_{t-1})\right)=\beta+\alpha_n E(X^{(n)}_{t-1})+\alpha_n^2\mbox{Var}(X^{(n)}_{t-1})=\\
		&&	\beta\dfrac{1-\alpha_n^{t}}{1-\alpha_n}+\alpha_n^2\mbox{Var}(X^{(n)}_{t-1})=\dfrac{\beta}{1-\alpha_n}\left\{\sum_{k=0}^{t-1}\alpha_n^{2k}-\alpha_n^t\sum_{k=0}^{t-1}\alpha_n^{k}\right\}=\dfrac{\beta}{1-\alpha_n}\left\{\dfrac{1-\alpha_n^{2t}}{1-\alpha_n^2}-\alpha_n^{t}\dfrac{1-\alpha_n^{t}}{1-\alpha_n}\right\}.
	\end{eqnarray*}

	Finally, for $k,t\in\mathbb N_0$, the autocorrelation function becomes
	\begin{eqnarray*} 
		\mbox{cov}(X^{(n)}_{t+k},X^{(n)}_t)&=&E(\mbox{cov}(X^{(n)}_{t+k},X^{(n)}_t)|\mathcal F^{(n)}_t)+\mbox{cov}(E(X^{(n)}_{t+k}|\mathcal F^{(n)}_t),E(X^{(n)}_t|\mathcal F^{(n)}_t))\\
		&=&\mbox{cov}(E(X^{(n)}_{t+k}|\mathcal F^{(n)}_t),X^{(n)}_t))=\alpha_n\mbox{cov}(E(X^{(n)}_{t+k-1}|\mathcal F^{(n)}_t),X^{(n)}_t))\\
		&=&\alpha_n\mbox{cov}(X^{(n)}_{t+k-1},X^{(n)}_t)=\alpha_n^k\mbox{Var}(X^{(n)}_t),
	\end{eqnarray*}	 
	where we have used in the third equality the fact that $E(X^{(n)}_{t+k}|\mathcal F^{(n)}_t)=E\left(E(X^{(n)}_{t+k}|\mathcal F^{(n)}_{t+k-1})|\mathcal F^{(n)}_t\right)=\beta+\alpha_nE(X^{(n)}_{t+k-1}|\mathcal F^{(n)}_t)$ since $\mathcal F^{(n)}_t\subseteq\mathcal F^{(n)}_{t+k-1}$ for $k\geq1$.
\end{proof}	 

From Proposition \ref{prop:moments}, we have that $E(X^{(n)}_{\floor{nt}})\approx\beta\gamma^{-1}n(1-e^{-\gamma t})=\mathcal O(n)$ and $\mbox{Var}(X^{(n)}_{\floor{nt}})\approx\beta\gamma^{-2}n^2(1-e^{-\gamma t})^2/2=\mathcal O(n^2)$. We then define the normalized process $\mathcal X^{(n)}(t)\equiv X^{(n)}_{\floor{nt}}/n$ and obtain that $\mathcal X^{(n)}(t)=\mathcal O_p(1)$, for $t\geq0$. In the following theorem, we establish the weak convergence of the process $\{\mathcal X^{(n)}(t);\,\, t\geq0\}$ as $n\rightarrow\infty$. We introduce some notation before presenting such a result. Denote by $D^+[0,\infty)$ the space of the non-negative c{\`a}dl{\`a}g (right continuous with left limits) functions on $[0,\infty)$ and $C_c^{\infty}[0,\infty)$ the space of infinitely differentiable functions on $[0,\infty)$ having compact supports.

\begin{theorem}\label{mainthm}
The stochastic process $\{\mathcal X^{(n)}(t);\,\, t\geq0\}$ weakly converges in $D^+[0,\infty)$ endowed with the Skorokhod topology to a diffusion process $\{\mathcal X(t);\,\, t\geq0\}$ given by the solution of the stochastic differential equation
\begin{eqnarray}\label{diffusion}
d\mathcal X(t)=(\beta-\gamma\mathcal X(t))dt+\sqrt{\mathcal X(t)}dB(t),\quad t>0,
\end{eqnarray}
and $\mathcal X(0)=0$, as $n\rightarrow\infty$, where $\{B(t);\,\,t\geq0\}$ is a standard Brownian motion.
\end{theorem}

\begin{remark}
	The process $\{\mathcal X(t);\,\,t\geq0\}$ appearing in Theorem \ref{mainthm}, Equation (\ref{diffusion}), is known in the literature as the Cox-Ingersoll-Ross (CIR) process \citep{coxetal1985}.
\end{remark}

\begin{proof}
We have that $X^{(n)}_{t}|X^{(n)}_{t-1}=n x\sim\mbox{Poisson}(\beta+\alpha_n nx)$, with $x\in E_n\equiv\{j/n: j=0,1,2,\ldots \}$; we here denote $Z_x^{(n)}\sim \mbox{Poisson}(\beta+\alpha_n nx)$ and $\widetilde Z_x^{(n)}\equiv Z_x^{(n)}/n$. In particular, $X^{(n)}_0=\kappa/n\rightarrow0$ almost surely. Note that $\widetilde Z_x^{(n)}$ is a Markov chain assuming values in $E_n$. For $h\in C_c^{\infty}[0,\infty)$, define $T_n h(x)\equiv E\left(h(\widetilde Z_x^{(n)})\right)$. From Theorem 6.5 in Chapter 1 and Corollary 8.9 in Chapter 4 of \cite{ethkur1986}, to obtain the desired result, it is enough to show that
\begin{eqnarray}\label{cond_conv_diff}
\lim_{n\rightarrow\infty}\sup_{x\in E_n}|\epsilon_n(x)|=0, \quad h\in C_c^{\infty}[0,\infty),
\end{eqnarray}
with $\epsilon_n(x)=n\left(T_n h(x)-h(x)\right)-(\beta-\gamma x)h'(x)-\dfrac{1}{2}xh''(x)$, where $h'(\cdot)$ and $h''(\cdot)$ denote the first and second derivatives of $h(\cdot)$, respectively.

For $\widetilde Z_x^{(n)}\neq x$, we have that
\begin{eqnarray}\label{aux1}
\int_0^1h''(x+v(\widetilde Z_x^{(n)}-x))dv=\dfrac{h'(\widetilde Z_x^{(n)})-h'(x)}{\widetilde Z_x^{(n)}-x}
\end{eqnarray}
and 
\begin{eqnarray}\label{aux2}
\int_0^1vh''(x+v(\widetilde Z_x^{(n)}-x))dv=\dfrac{h'(\widetilde Z_x^{(n)})}{\widetilde Z_x^{(n)}-x}-\dfrac{h(\widetilde Z_x^{(n)})-h(x)}{(\widetilde Z_x^{(n)}-x)^2}.
\end{eqnarray}

By combining (\ref{aux1}) and (\ref{aux2}), we obtain that
\begin{eqnarray}\label{aux3}
n\left(h(\widetilde Z_x^{(n)})-h(x)\right)=\int_0^1n(\widetilde Z_x^{(n)}-x)^2(1-v)h''(x+v(\widetilde Z_x^{(n)}-x))dv+n(\widetilde Z_x^{(n)}-x)h'(x).
\end{eqnarray}

Note that Equation (\ref{aux3}) also holds for $\widetilde Z_x^{(n)}=x$. Further, we can write
\begin{eqnarray}\label{aux4}
-\dfrac{1}{2}E\left(n(\widetilde Z_x^{(n)}-x)^2\right)h''(x)=E\left(-\int_0^1n(\widetilde Z_x^{(n)}-x)^2(1-v)h''(x)dv\right).
\end{eqnarray}

We now use the Equations (\ref{aux3}) and (\ref{aux4}) to express $\epsilon_n(x)$ as follows:
\begin{eqnarray}\label{aux5}
\epsilon_n(x)&=&E\left(\int_0^1n(\widetilde Z_x^{(n)}-x)^2(1-v)\left(h''(x+v(\widetilde Z_x^{(n)}-x))-h''(x)\right)dv\right)+\nonumber\\
&&h'(x)\left\{E\left(n\widetilde Z_x^{(n)}\right)-(\beta-\gamma x)\right\}+\dfrac{1}{2}h''(x)\left\{E\left(n(\widetilde Z_x^{(n)}-x)^2\right)-x\right\}\nonumber\\
&\coloneqq&\epsilon^{(1)}_n(x)+\epsilon^{(2)}_n(x)+\epsilon^{(3)}_n(x).
\end{eqnarray}

We will show that $\lim_{n\rightarrow\infty}\sup_{x\in E_n}|\epsilon^{(j)}_n(x)|=0$, for $j=1,2,3$. This result, Equation (\ref{aux5}), and the triangular inequality imply that (\ref{cond_conv_diff}) holds and therefore conclude the proof of the theorem. 

To show the case $j=1$, we argue as in the proof of Theorem 3.1 in Chapter 9 of \cite{ethkur1986}. Then, the result follows by showing that $\lim_{n\rightarrow\infty}|\epsilon^{(1)}_n(x_n)|=0$ for any convergent sequence $\{x_n\}_{n\in\mathbb N}$, where $x_n\rightarrow\infty$ is allowed.
Without loss of generality, suppose that the support of $h(\cdot)$ is contained in the interval $[0,c]$, for constant $c>0$. For $v\in(0,1)$ and $x\in E_n^*\equiv E_n-\{0\}$, it folllows that $x+v(\widetilde Z_x^{(n)}-x)>x(1-v)$ and therefore the integral involved in $\epsilon^{(1)}_n(x)$ equals 0 under the region $x(1-v)>c$ ($h''(z)=0$ for $z>c$), that is $v<1-c/x$. Define $\omega_*(x)=\min\{0,1-c/x\}$ for $x>0$,  $\omega_*(0)=0$, $\omega^*(x)=\max\{1,c/x\}$ for $x>0$,  and $\omega^*(0)=1$. Hence, it follows that 
\begin{eqnarray}\label{ineq_thm1}
|\epsilon^{(1)}_n(x_n)|&=& \bigg|E\left(\int_{\omega_*(x)}^1n(\widetilde Z_x^{(n)}-x)^2(1-v)\left(h''(x+v(\widetilde Z_x^{(n)}-x))-h''(x)\right)dv\right)\bigg|\nonumber\\
&\leq&E\left(\int_{\omega_*(x)}^1n(\widetilde Z_x^{(n)}-x)^2(1-v)2\|h''\|dv\right)=nE\left((\widetilde Z_x^{(n)}-x)^2\right)\|h''\|\omega^*(x)^2.
\end{eqnarray}

Further, we have that $E\left((\widetilde Z_x^{(n)}-x)^2\right)=n^{-2}(\beta+\beta^2+\gamma_n^2x^2)+2\beta n^{-1}(\alpha_n-1)x+n^{-1}\gamma_n^2 x^2$. Consider $x_n\rightarrow 0$, then $nE\left((\widetilde Z_x^{(n)}-x_n)^2\right)\rightarrow 0$ and $\omega^*(x_n)\rightarrow 1$. These results give us that the right-hand side of (\ref{ineq_thm1}) goes to 0 as $n\rightarrow\infty$. We obtain the same conclusion when $x_n\rightarrow\infty$ since $nE\left((\widetilde Z_x^{(n)}-x_n)^2\right)=\mathcal O(x_n)$ and $\omega^*(x_n)^2=\mathcal O(x_n^{-2})$, and hence $\lim_{n\rightarrow\infty}nE\left((\widetilde Z_x^{(n)}-x)^2\right)\|h''\|\omega^*(x)^2=\lim_{n\rightarrow\infty}\mathcal O(x_n^{-1})=0$. Suppose now that $x_n\rightarrow x\in(0,\infty)$. We can establish the weak convergence of $\sqrt{n}(\widetilde Z_x^{(n)}-x)$ via its characteristic function as follows:
\begin{eqnarray*}
E\left(\exp\{it\sqrt{n}(\widetilde Z_x^{(n)}-x_n)\}\right)&=&\exp\left\{-it\sqrt{n}x_n+(\beta+\alpha_nnx_n)(e^{it/\sqrt{n}}-1)\right\}\\
&=&\exp\left\{itx_n\dfrac{\gamma_n}{\sqrt{n}}-\alpha_n x_n\dfrac{t^2}{2}+\mathcal O(n^{-3/2})\right\}\\
&\longrightarrow& \exp\{-xt^2/2\},\quad t\in\mathbb R,
\end{eqnarray*}
as $n\rightarrow\infty$. Therefore, $\sqrt{n}(\widetilde Z_x^{(n)}-x_n)\stackrel{d}{\longrightarrow}N(0,x)$. Hence, the integrand in $\epsilon^{(1)}_n(x_n)$ is bounded above by an integrable random variable. Further, this integrand converges in probability to 0 since $\widetilde Z_x^{(n)}-x_n\stackrel{p}{\longrightarrow}0$. We then apply the Dominated Convergence Theorem to conclude that $\lim_{n\rightarrow\infty}|\epsilon^{(1)}_n(x_n)|=0$.

For the case $j=2$, it follows that 
\begin{eqnarray*}
\sup_{x\in E_n}|\epsilon^{(2)}_n(x)|&=&\sup_{x\in E_n}x|h'(x)||\gamma_n-\gamma|\leq \sup_{x\in E_n}x|h'(x)|I\{0\leq x\leq c\}|\gamma_n-\gamma|\\
&\leq&c\|h'\||\gamma_n-\gamma|\rightarrow 0,
\end{eqnarray*}
as $n\rightarrow\infty$. In a similar fashion, for $j=3$, it can be shown that $\lim_{n\rightarrow\infty}\sup_{x\in E_n}|\epsilon^{(3)}_n(x)|=0$, which concludes the proof.
\end{proof}

\section{Conditional Least Squares}\label{sec:cls}

In this section, we provide the asymptotic distribution of the conditional least squares estimator of $\alpha_n$ for the nearly unstable INARCH process. The parameter $\beta$ is assumed to be known. This can be seen as a nuisance parameter since our main interest relies on the parameter $\alpha_n$ that controls the dependence in the model. In the empirical illustration, we discuss how to deal with the unknown $\beta$ case.
  
  The CLS estimator of $\alpha$ is obtained by minimizing the $Q$-function given by $Q(\alpha)=\sum_{t=2}^n(X_t-E(X_t|\mathcal F_{t-1}))^2=\sum_{t=2}^n(X_t-\beta-\alpha X_{t-1})^2$. Hence, we obtain explicitly the CLS estimator of $\alpha$, say $\widehat\alpha_n$, which is given by
    \begin{eqnarray}\label{cls}
  	\widehat\alpha_n=\dfrac{\displaystyle\sum_{t=2}^nX_{t-1}(X_t-\beta)}{\displaystyle\sum_{t=2}^nX_{t-1}^2}.
  \end{eqnarray}
  
We begin by deriving the asymptotic distribution of $\widehat\alpha_n$ under the stationary assumption, where we denote the count time series by $\{X_t\}_{t\in\mathbb N}$ (no need for the superscript $(n)$). This case will be contrasted to the nearly unstable INARCH process through simulation in the following section. 
  
\begin{theorem}\label{weak_limit_stat}
	Assume that $X_1,\ldots,X_n$ is a trajectory from a stationary Poisson INARCH(1) model, that is $\alpha_n=\alpha<1$. Then, the CLS estimator $\widehat\alpha_n$ given in (\ref{cls}) satisfies
	\begin{eqnarray*}
		\sqrt{n}(\widehat\alpha_n-\alpha)\stackrel{d}{\longrightarrow}N(0,\widetilde\sigma^2),
	\end{eqnarray*}	
as $n\rightarrow\infty$, where 
\begin{eqnarray*}
\widetilde\sigma^2=\dfrac{(1-\alpha)(1-\alpha^2)}{(1+\beta(1+\alpha))^2}\left\{1+\beta(1-\alpha)+\dfrac{\alpha(2+\beta^{-1})}{1-\alpha}-\dfrac{\alpha^2(1-\alpha)\beta^{-1}}{1-\alpha^3}+\dfrac{1+\beta(1+\alpha)}{1-\alpha}\right\}.
\end{eqnarray*}	
\end{theorem}  
  
\begin{proof}
From \cite{foketal2009}, we have that $\{X_t\}$ is strictly stationary and ergodic since $\alpha<1$. Hence, we can use Theorem 3.2 from \cite{tjo1986} to establish the asymptotic normality of the CLS estimator $\widehat\alpha_n$. The other conditions necessary to obtain this weak convergence can be straightforwardly checked in our case and therefore are omitted. Applying this theorem, we get that the asymptotic variance, say $\widetilde\sigma^2$, assumes the form $\widetilde\sigma^2=R/U^2$, with $U=E\left(\left(\dfrac{\partial E(X_t|\mathcal F_{t-1})}{\partial\alpha}\right)^2\right)=E(X_{t-1}^2)$
and $R=E\left(\left(\dfrac{\partial E(X_t|\mathcal F_{t-1})}{\partial\alpha}\right)^2 \mbox{Var}(X_t|\mathcal F_{t-1})\right)=\beta E(X_{t-1}^2)+\alpha E(X_{t-1}^3)$. Explicit expression for the marginal moments of a Poisson INARCH(1) model are given in \cite{wei2010}. Using these results and the notation considered there with $f_k\equiv \dfrac{\beta}{\prod_{i=1}^k(1-\alpha^i)}$, for $k\in\mathbb N$, we obtain $U=f_2(1+\beta(1+\alpha))$ and $R=\dfrac{\alpha f_2(1+\beta)}{1-\alpha}+\alpha f_1f_2-\alpha^2(1-\alpha)f_3+\alpha f_1f_2(1+\beta(1+\alpha))+\beta f_2(1+\beta(1+\alpha))$. Direct algebric manipulations conclude the proof.
\end{proof}

From now on assume that $\{X_t^{(n)}\}_{t\in\mathbb N}$ is a nearly unstable INARCH process as given in Definition \ref{def:n_unst}. 
Define $W_t^{(n)}=X^{(n)}_t-E(X^{(n)}_t|\mathcal F^{(n)}_{t-1})$, $X^{(n)}(s)=X^{(n)}_{\floor*{ns}}$, and $W^{(n)}(s)=\sum_{k=1}^{\floor*{ns}}W^{(n)}_k$, for $t\in\mathbb N_0$ and $s\geq0$, where $\floor*{x}$ denotes the integer-part of $x\in\mathbb R$. Like in the nearly unstable INAR process by \cite{ispetal2003}, we can express $\widehat\alpha_n-\alpha_n$ as
\begin{eqnarray}\label{useful_rep}
	\widehat\alpha_n-\alpha_n=\dfrac{\displaystyle\sum_{t=2}^nX^{(n)}_{t-1}W^{(n)}_t}{\displaystyle\sum_{t=2}^n(X^{(n)}_{t-1})^2}=\dfrac{\displaystyle\int_0^1X^{(n)}(s)dW^{(n)}(s)}{n\displaystyle\int_0^1(X^{(n)}(s))^2ds}.
\end{eqnarray}	  
  
 In the following lemma, we provide the asymptotic behavior of the autocovariance function of the process $\{W^{(n)}(s);\,\,s\geq0\}$; note that $E(W^{(n)}(s))=0$. This will be important to identify the proper normalization of $\widehat\alpha_n-\alpha_n$ in (\ref{useful_rep}) yielding a non-trivial weak limit.

\begin{lemma}\label{lem:cov_M}
	For $s,v\geq0$, we have that $\mbox{cov}(W^{(n)}(s),W^{(n)}(v))\approx n^2 C_W(s\wedge v)$, where $C_W(u)=\beta\gamma^{-2}(\gamma u+e^{-\gamma u}-1)$ for $u\geq0$, $s\wedge v=\min(s,v)$, and $a_n\approx b_n$ denoting that $\lim_{n\rightarrow\infty}a_n/b_n=1$ for real sequences $\{a_n\}$ and $\{b_n\}$.
\end{lemma}    
  
\begin{proof}  It is straightforward that $E(W^{(n)}(s))=0$ and $\mbox{cov}(W^{(n)}_k,W^{(n)}_l)=0$ for $k\neq l$. Further, $\mbox{Var}(W^{(n)}_k)=\mbox{Var}(X^{(n)}_k)+\alpha_n^2\mbox{Var}(X^{(n)}_{k-1})-2\alpha_n\mbox{cov}(X^{(n)}_k,X^{(n)}_{k-1})=\mbox{Var}(X^{(n)}_k)-\alpha_n^2\mbox{Var}(X^{(n)}_{k-1})$, where the last equality follows from the expression of the covariance given in Proposition \ref{prop:moments}. After using the expression of the variance given in that lemma, we obtain that
	$\mbox{Var}(W^{(n)}_k)=\beta\dfrac{1-\alpha_n^k}{1-\alpha_n}$.
	
	From the above results and Proposition \ref{prop:moments}, we obtain that
	\begin{eqnarray*}
		\mbox{cov}(W^{(n)}(s),W^{(n)}(v))&=&\sum_{k=1}^{{\floor*{ns}}\wedge{\floor*{nv}}}\mbox{Var}(W^{(n)}_k)=\dfrac{\beta}{1-\alpha_n}\left\{{\floor*{ns}}\wedge{\floor*{nv}}-\alpha_n\dfrac{1-\alpha_n^{{\floor*{ns}}\wedge{\floor*{nv}}}}{1-\alpha_n}\right\}\\
		&\approx& n^2 \beta\gamma^{-2}(\gamma u+e^{-\gamma u}-1)=n^2C_W(s\wedge v).
	\end{eqnarray*}	 
\end{proof}  
  
Lemma \ref{lem:cov_M} and Theorem \ref{mainthm} give us that $	\widehat\alpha_n-\alpha_n=\mathcal O_p(n^{-1})$. We now are able to establish the asymptotic distribution of the CLS estimator $\widehat\alpha_n$ under the nearly unstable INARCH process as follows.  
  
\begin{theorem}\label{weak_limit_nonstat}
Let $\{\mathcal X(t);\,\,t\geq0\}$ be the diffusion process given in (\ref{diffusion}). Then, the CLS estimator $\widehat\alpha_n$ satisfy the following weak convergence
\begin{eqnarray}\label{cls_weak_conv}
n(\widehat\alpha_n-\alpha_n)\stackrel{d}{\longrightarrow}\dfrac{\displaystyle\int_0^1\mathcal X(t)d\mathcal W(t)}{\displaystyle\int_0^1\mathcal X(t)^2dt}=\dfrac{\displaystyle\int_0^1\mathcal X(t)^{3/2}dB(t)}{\displaystyle\int_0^1\mathcal X(t)^2dt},
\end{eqnarray}
as $n\rightarrow\infty$, where $d\mathcal W(t)=\sqrt{\mathcal X(t)}dB(t)$, for $t>0$, with $\mathcal W(0)=0$.
\end{theorem}

\begin{proof}
Define $\mathcal W^{(n)}(s)=W^{(n)}(s)/n$, for $s>0$. We have that $$n(\widehat\alpha_n-\alpha_n)=\dfrac{\displaystyle\int_0^1X^{(n)}(s)dW^{(n)}(s)}{\displaystyle\int_0^1(X^{(n)}(s))^2ds}=\dfrac{\displaystyle\int_0^1\mathcal X^{(n)}(s)d\mathcal W^{(n)}(s)}{\displaystyle\int_0^1(\mathcal X^{(n)}(s))^2ds},$$ where both numerator and denominator have the same order of magnitude $\mathcal O_p(n^{2})$.
	
For $s>0$, it follows that 	
\begin{eqnarray*}
W^{(n)}(s)=\sum_{k=1}^{\floor*{ns}}\left(X^{(n)}_k-\beta-\alpha_nX^{(n)}_{k-1}\right)=X^{(n)}_{\floor*{ns}}+\dfrac{\gamma_n}{n}\sum_{k=1}^{\floor*{ns}}X^{(n)}_{k-1}-\beta\floor*{ns}-\kappa,
\end{eqnarray*}	
and then $\mathcal W^{(n)}(s)$ can be expressed by
\begin{eqnarray*}
\mathcal W^{(n)}(s)=\mathcal X^{(n)}\left(\frac{\floor*{ns}}{n}\right)+\gamma_n\displaystyle\int_0^{\frac{\floor*{ns}}{n}}\mathcal X^{(n)}(u)du-\beta\frac{\floor*{ns}}{n}-\frac{\kappa}{n}.
\end{eqnarray*}	

Define the functions $\Phi_n$ ($n=1,2,\ldots$) and $\Phi$ mapping $D^+[0,\infty)$ into $D(\mathbb R_+,\mathbb R^2)$ as $\Phi_n(x)(s)=\bigg(x(s),\\x\left(\frac{\floor*{ns}}{n}\right)+\gamma_n\displaystyle\int_0^{\frac{\floor*{ns}}{n}}x(u)du-\beta\frac{\floor*{ns}}{n}-\frac{\kappa}{n}\bigg)$ and $\Phi(x)(s)=\bigg(x(s),x(s)+\gamma\displaystyle\int_0^{s}x(u)du-\beta s\bigg)$. Hence, it follows that $(\mathcal X^{(n)}(s),\mathcal W^{(n)}(s))=\Phi_n(\mathcal X^{(n)})(s)$. Using the fact that the CIR process has almost sure continuous trajectories and similar arguments given in the proof of Proposition 4.1 of \cite{ispetal2003}, we obtain that $\Phi_n(\mathcal X^{(n)})$ weakly converges to $\Phi(\mathcal X)$ as $n\rightarrow\infty$. 

In particular, we have that $\mathcal W^{(n)}(s)$ weakly converges to $\mathcal W(s)=\mathcal X(s)+\gamma\displaystyle\int_0^{s}\mathcal X(u)du-\beta s$. From the definition of $\mathcal X$, we have that $\gamma\displaystyle\int_0^{s}\mathcal X(u)du=-\mathcal X(s)+\beta s+\displaystyle\int_0^s\sqrt{\mathcal X (u)}dB(u)$ and, therefore, $\mathcal W(s)=\displaystyle\int_0^s\sqrt{\mathcal X (u)}dB(u)$. In other words, $d\mathcal W(t)=\sqrt{\mathcal X(t)}dB(t)$. The above results and the continuous mapping theorem give us that $\displaystyle\int_0^1\mathcal X^{(n)}(s)d\mathcal W^{(n)}(s)\stackrel{d}{\longrightarrow}\displaystyle\int_0^1\mathcal X(s)d\mathcal W(s)=\displaystyle\int_0^1\mathcal X(s)^{3/2}dB(s)$.

The above arguments are straightforwardly extended to establish the joint weak convergence 
\begin{eqnarray*}
\left(\mathcal X^{(n)}(s),\mathcal W^{(n)}(s),\int_0^1(\mathcal X^{(n)}(u))^2du\right)\Rightarrow\left(\mathcal X(s),\mathcal W(s),\int_0^1(\mathcal X(u))^2du\right)
\end{eqnarray*}
in $D(\mathbb R_+,\mathbb R^3)$ as $n\rightarrow\infty$. Then, the desired result given in (\ref{cls_weak_conv}) is obtained by applying the continuous mapping theorem.
\end{proof}

\section{Simulated Experiments}\label{sec:sim}
  
In this section, we present simulated results illustrating the behavior of the asymptotic distributions of the normalized CLS estimator under the nearly unstable and stable cases. All the numerical results of this paper were obtained by using the statistical software \texttt{R} \citep{R2021}. We conduct Monte Carlo simulations with 10000 replications, where we generate Poisson INARCH(1) trajectories with $\beta=1$, $\alpha=0.98,0.99,0.999$, and initially a sample size of $n=500$. Note that the chosen values for $\alpha$ here indicate nearly unstable count processes. For each replication, we compute the CLS estimate of $\alpha$ using (\ref{cls}) and then its standardized estimate as $n(\widehat\alpha_n-\alpha)$ and $\sqrt{n}(\widehat\alpha_n-\alpha)$ according to the nearly unstable (Theorem \ref{weak_limit_nonstat}) and stable/stationary (Theorem \ref{weak_limit_stat}) cases, respectively.
  
A generator from the asymptotic distribution given on the right-hand side of (\ref{cls_weak_conv}) was implemented, where the stochastic integrals are approximately evaluated via type-Riemann integrals. Hence, for instance, we can obtain its quantiles and also plot the associated density function by generating samples and then applying a non-parametric density estimator (here the Gaussian kernel is considered), which are important for what follows. We present the histograms and qq-plots of the standardized CLS estimates along with their associated asymptotic density/quantiles under the stable and nearly unstable cases in Figures \ref{fig:STAT_alpha_n500} and \ref{fig:NONSTAT_alpha_n500}, respectively. From Figure \ref{fig:STAT_alpha_n500}, it is evident that the normal approximation is not adequate and it is worsening when $\alpha$ gets closer to 1, which is expected since these results are based on stationarity. On the other hand, the histograms and qq-plots regarding the nearly unstable approximation given in Figure \ref{fig:NONSTAT_alpha_n500} show an excellent agreement between the empirical standardized estimates and the theoretical asymptotic distribution for all scenarios.

\begin{figure}
\begin{center}
\includegraphics[width=5.4cm]{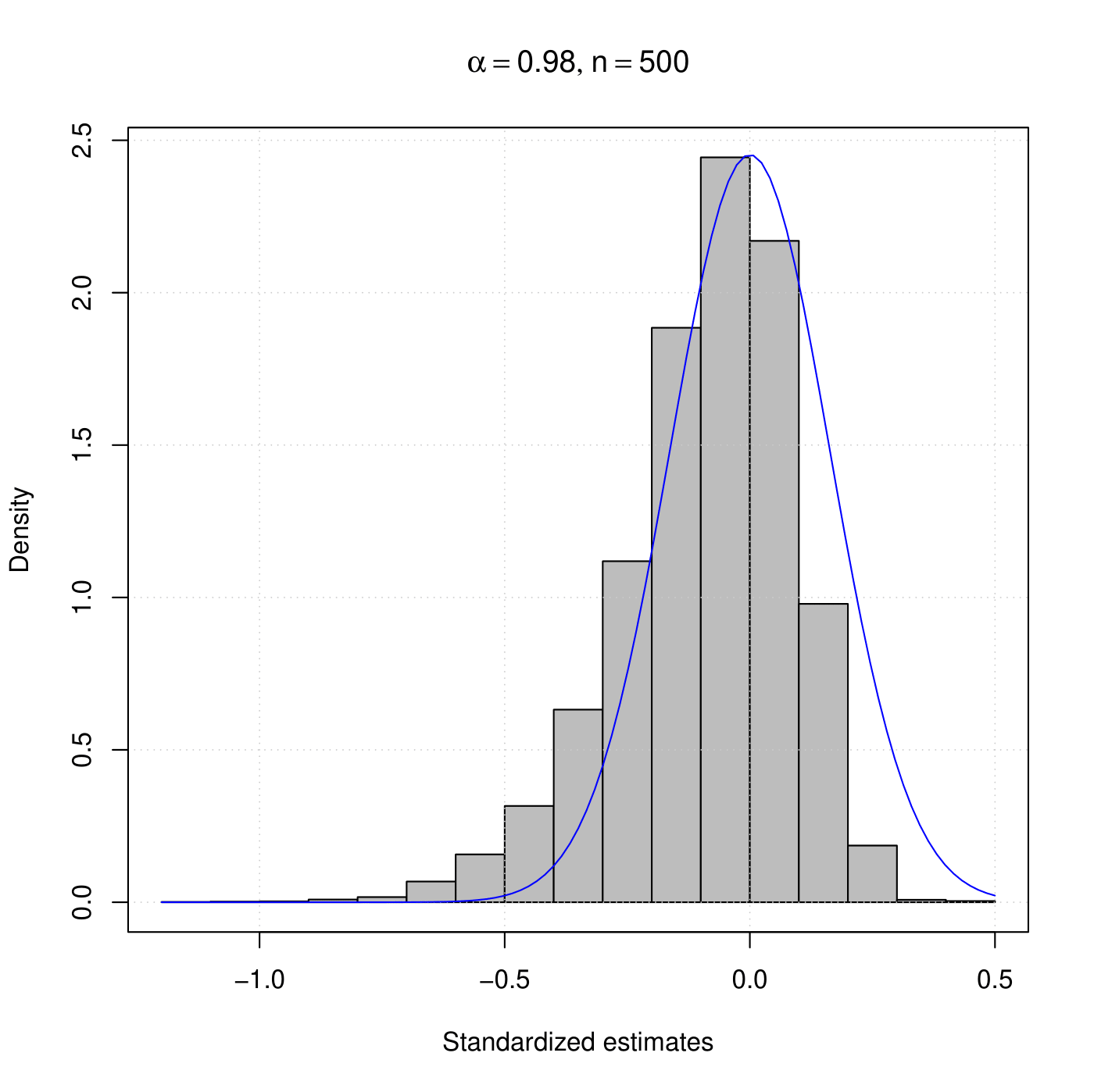}\includegraphics[width=5.4cm]{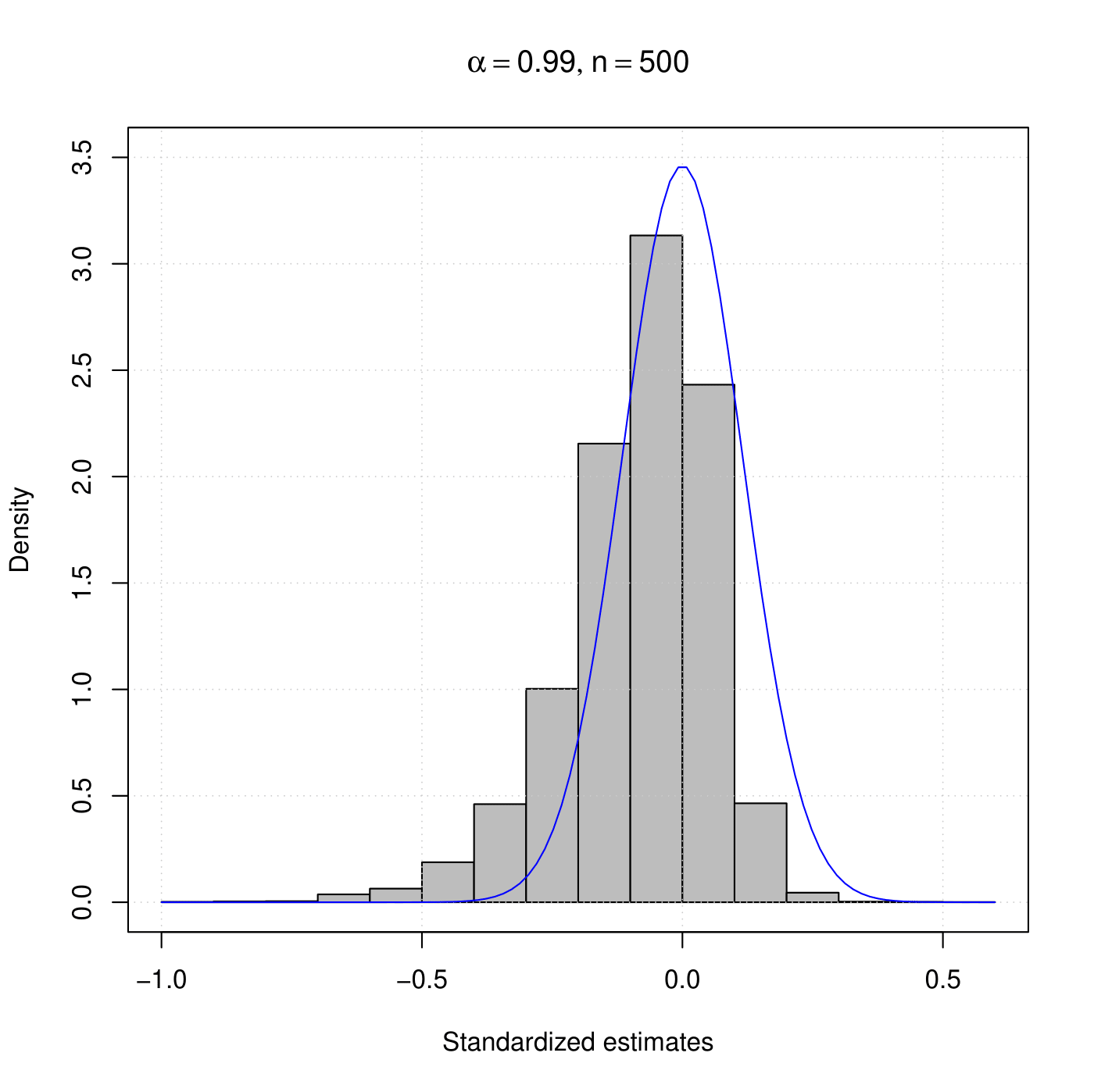}\includegraphics[width=5.4cm]{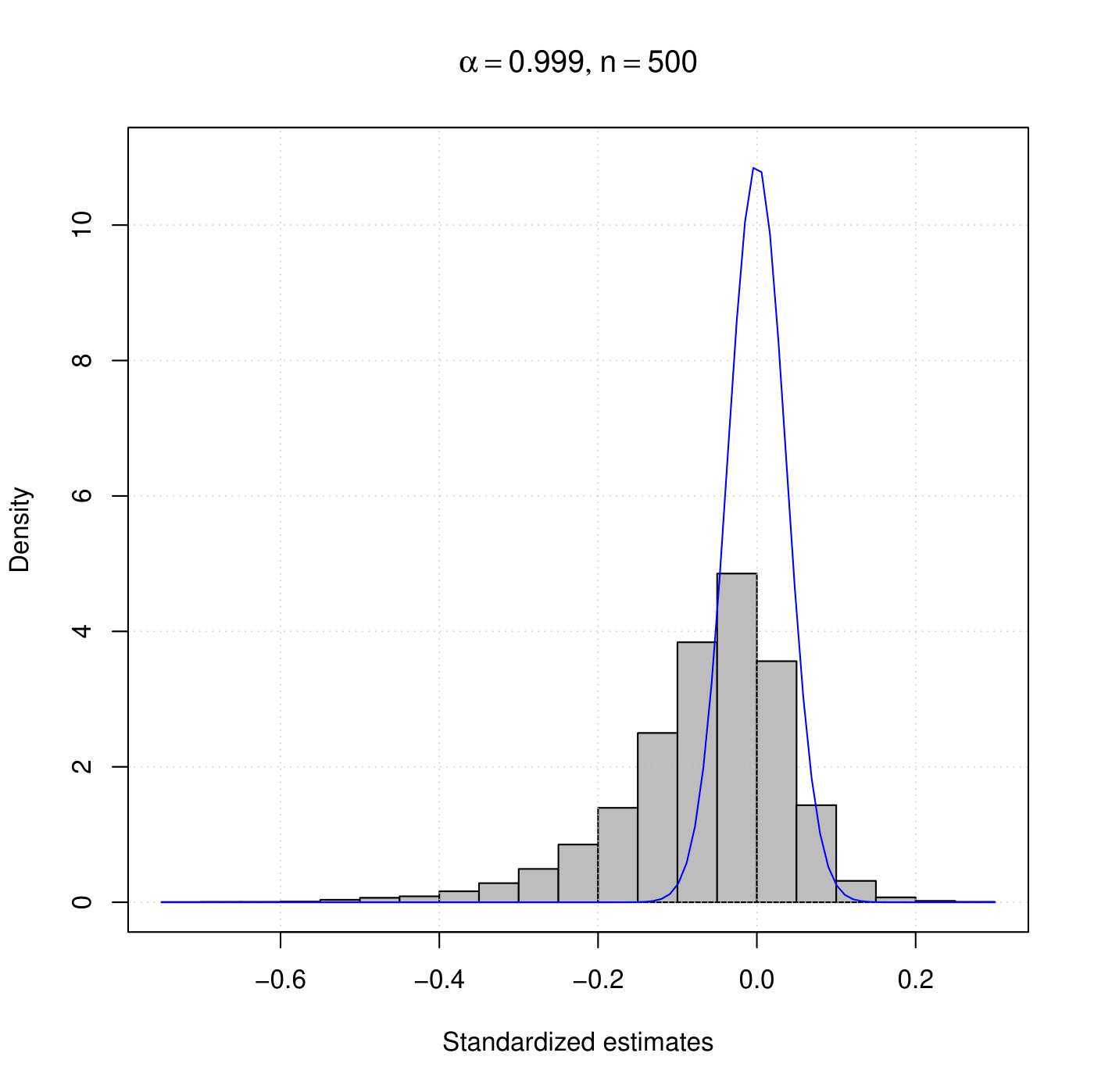}
\includegraphics[width=5.4cm]{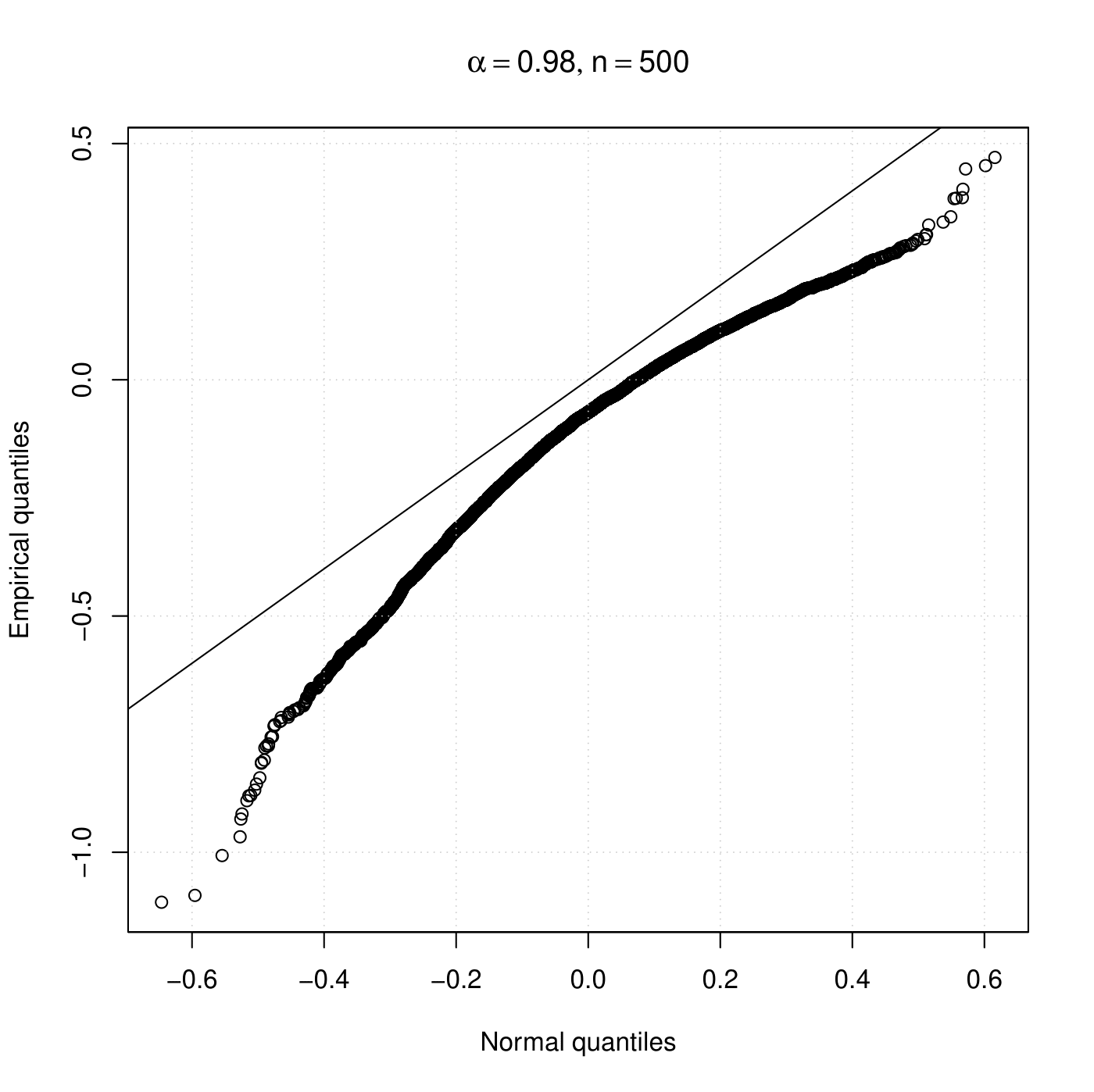}\includegraphics[width=5.4cm]{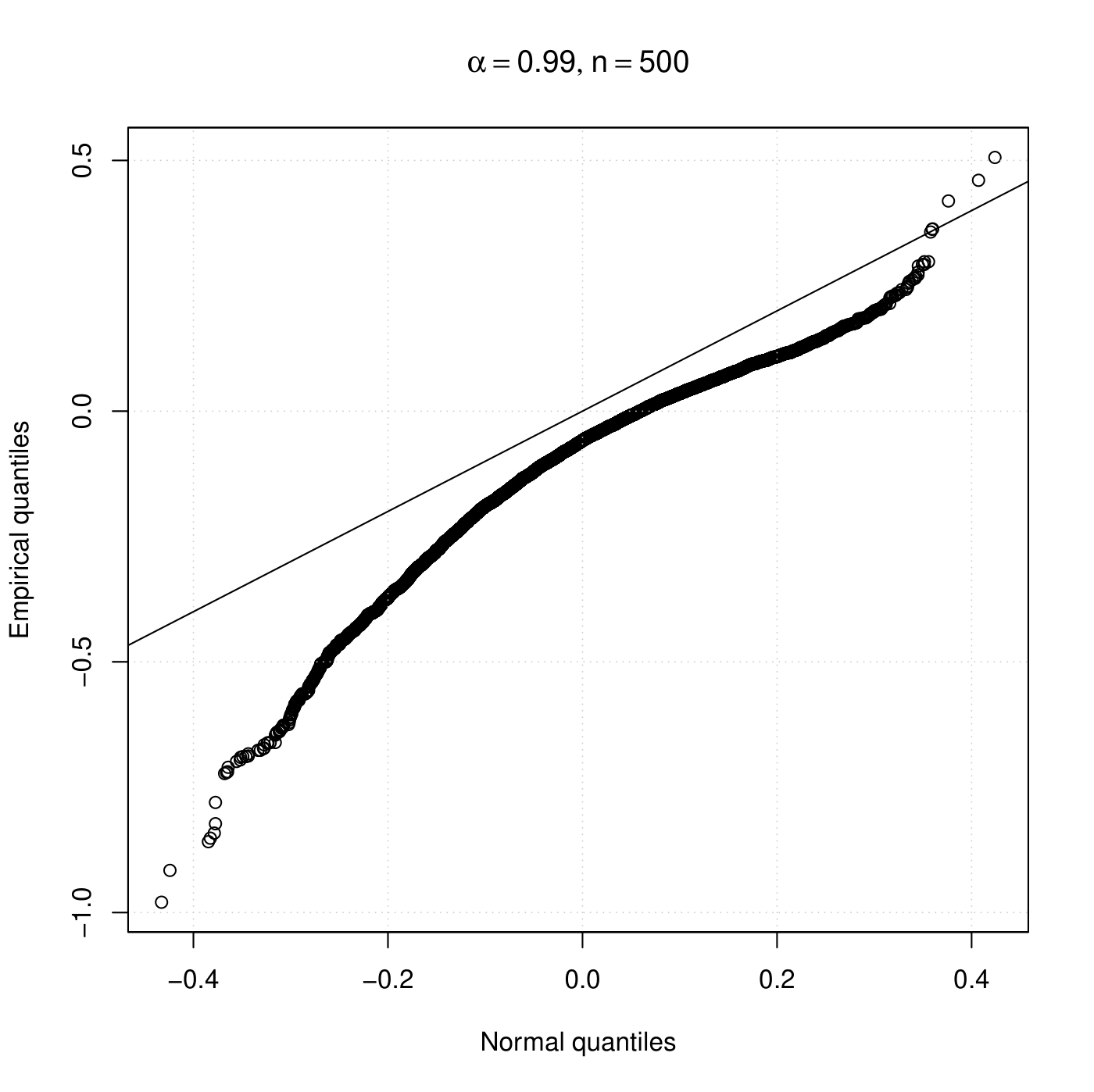}\includegraphics[width=5.4cm]{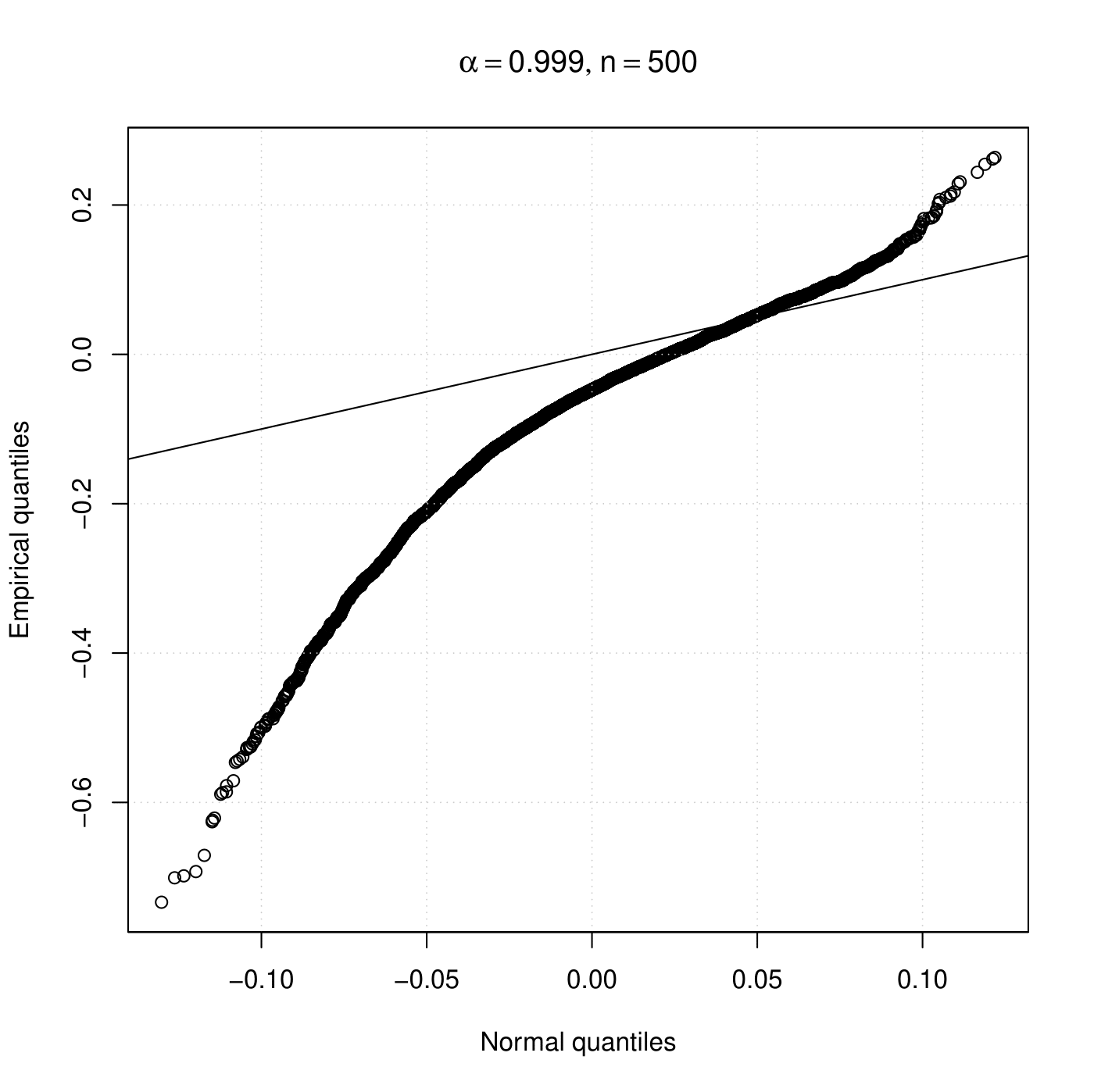}  	
\caption{Histograms and qq-plots of the standardized estimates $\sqrt{n}(\widehat\alpha_n-\alpha)$ for $\alpha=0.98$, $\alpha=0.99$, and $\alpha=0.999$, along with their associated limiting normal density/quantiles given in Theorem \ref{weak_limit_stat} (under stationarity). The sample size is $n=500$.}\label{fig:STAT_alpha_n500}
\includegraphics[width=5.4cm]{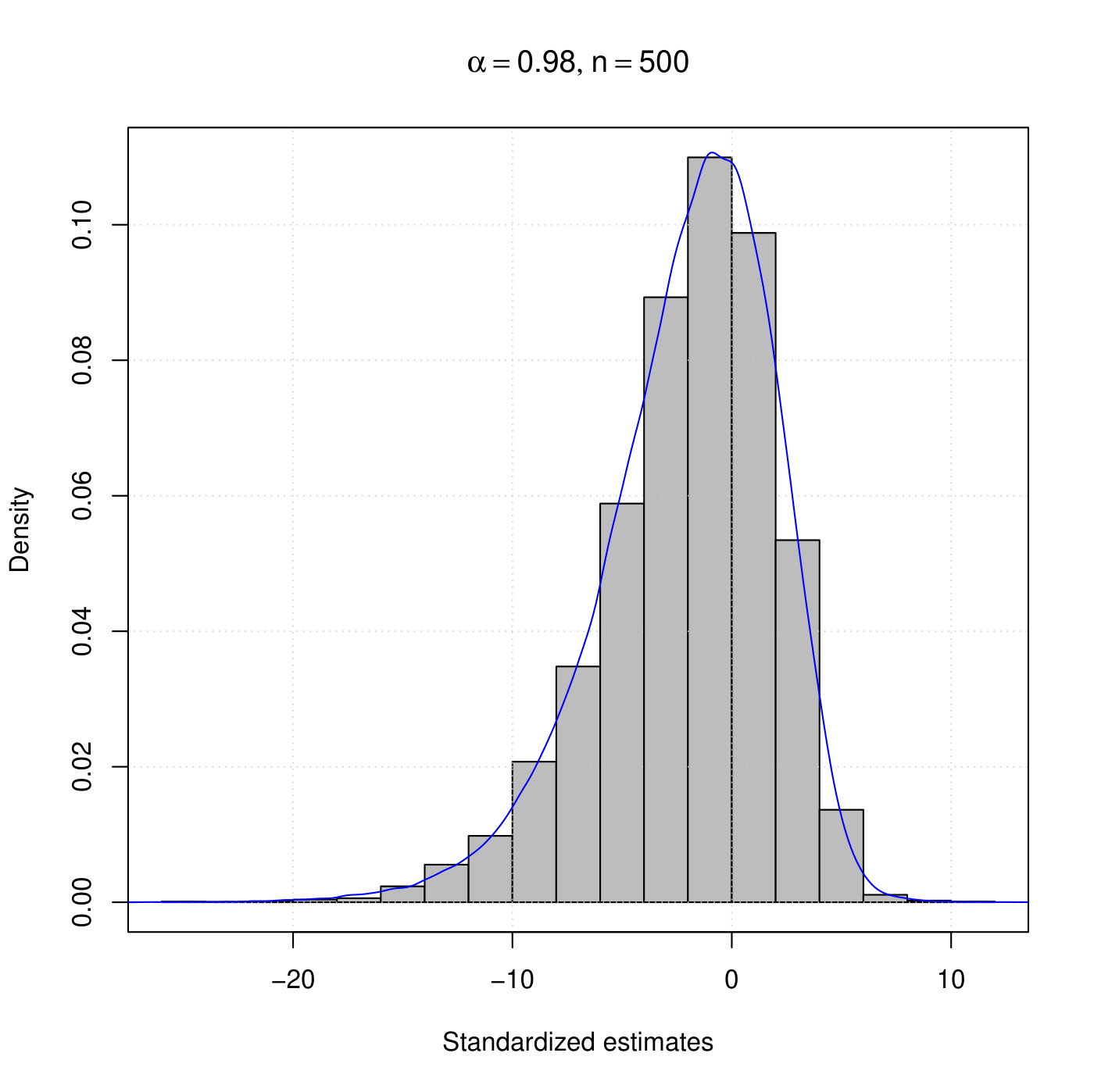}\includegraphics[width=5.4cm]{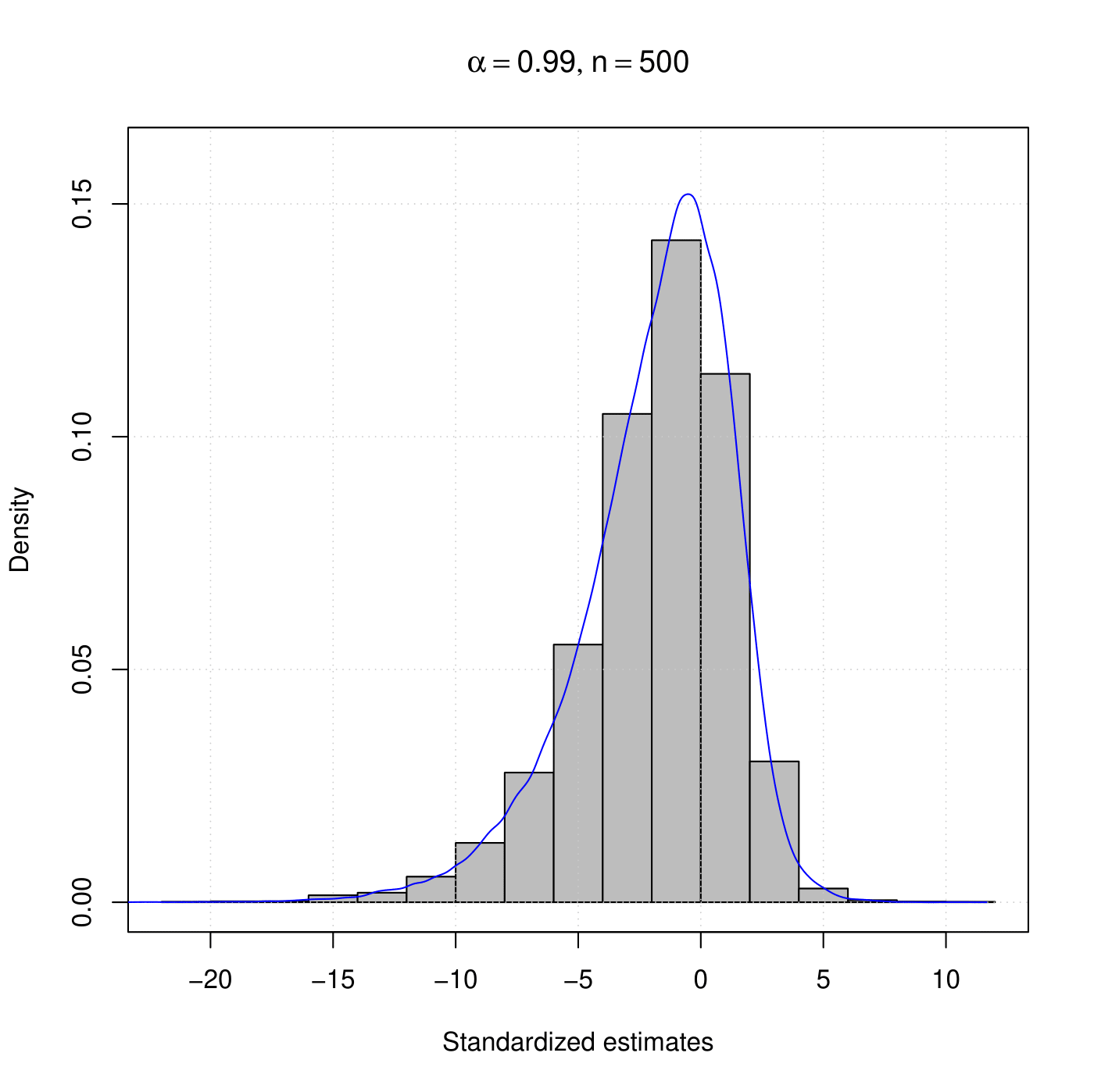}\includegraphics[width=5.4cm]{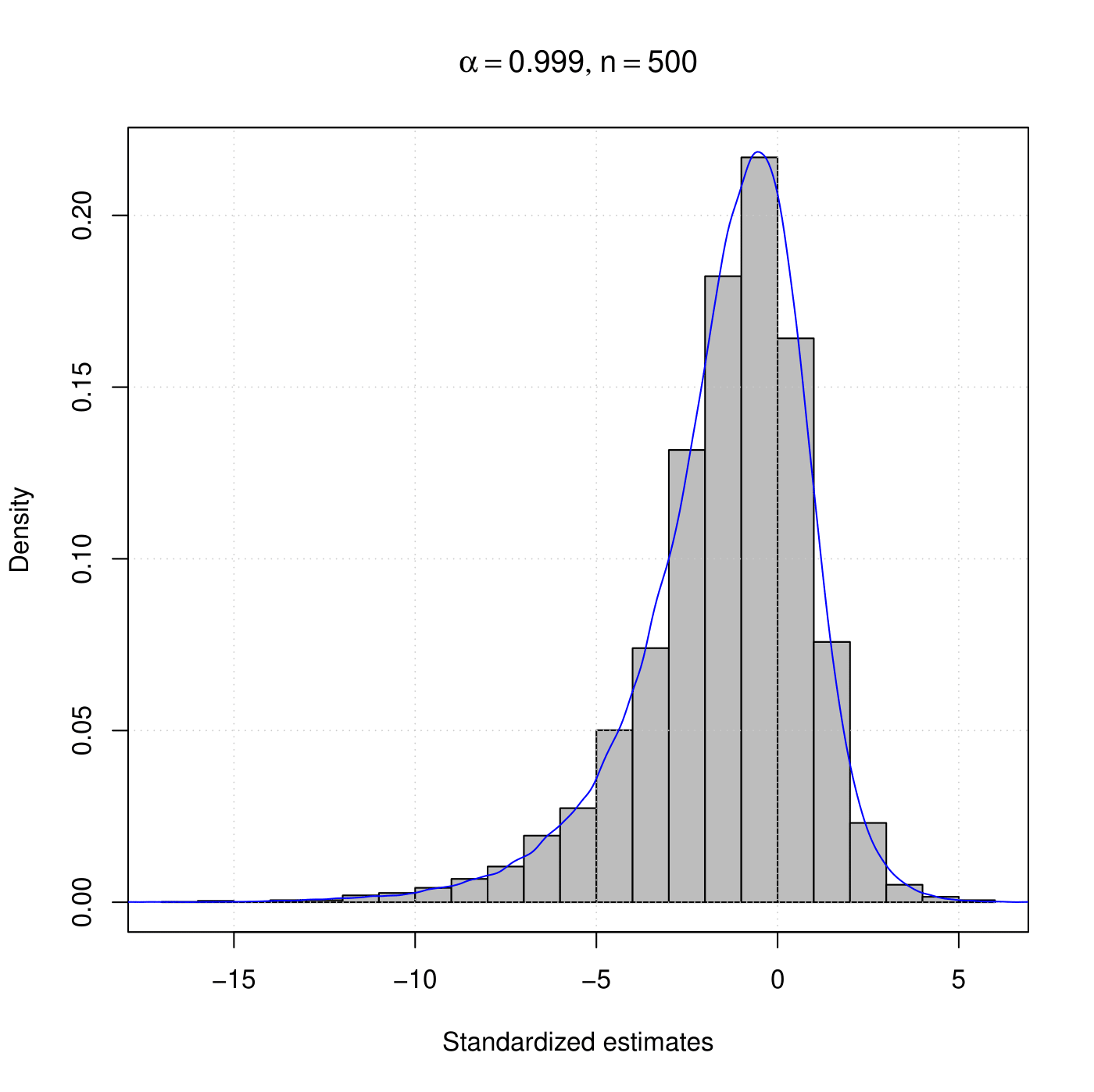}
\includegraphics[width=5.4cm]{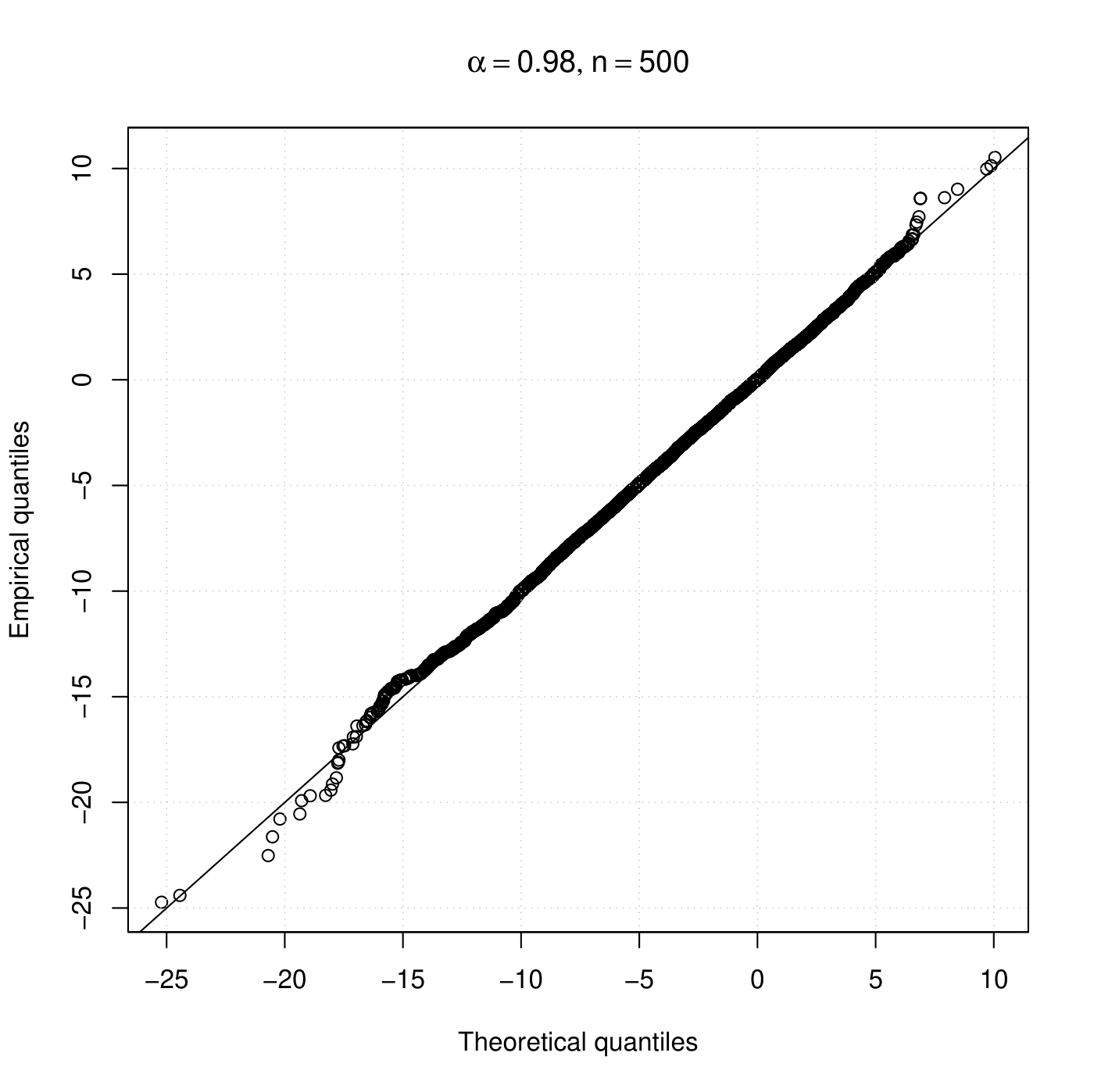}\includegraphics[width=5.4cm]{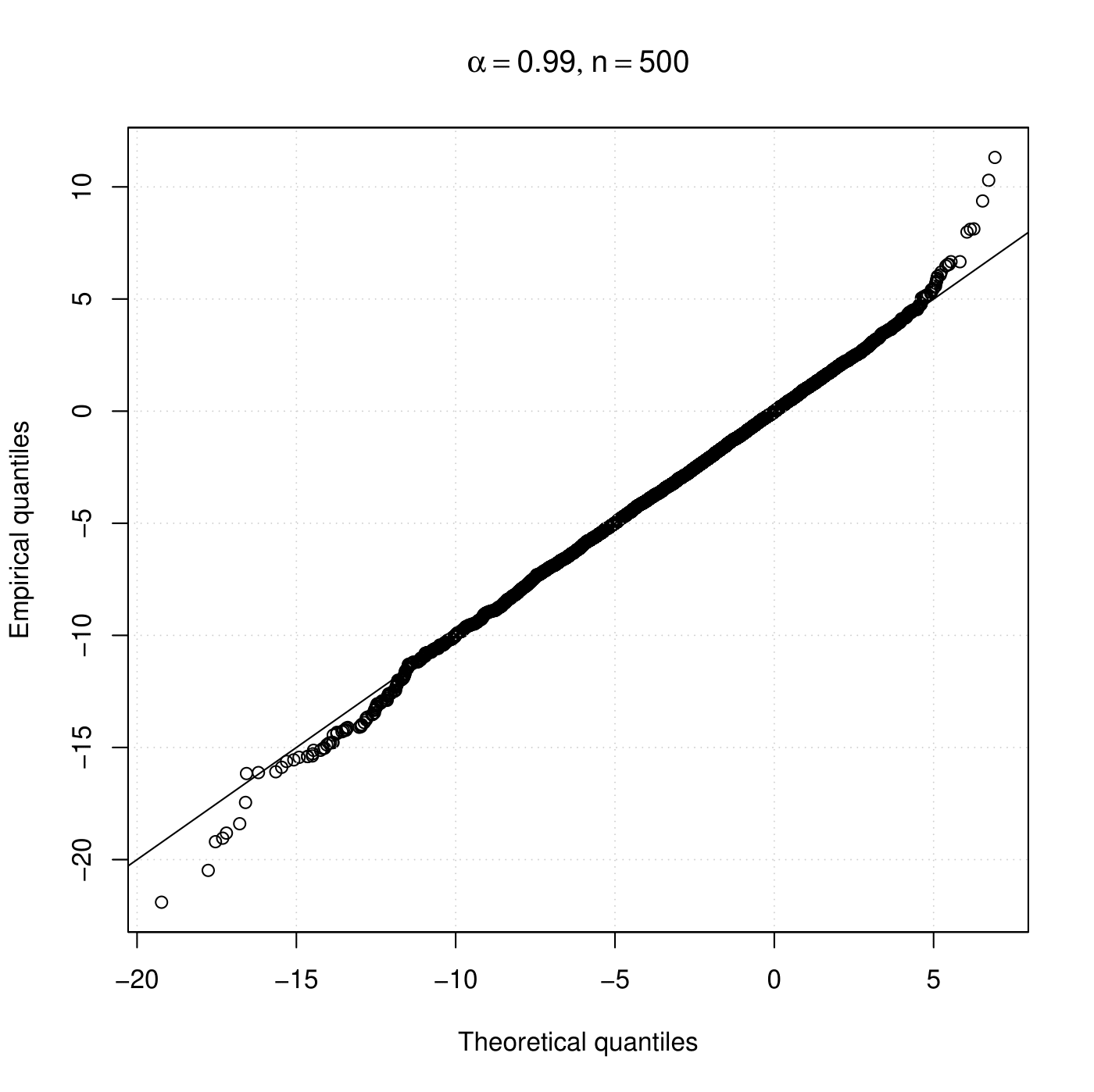}\includegraphics[width=5.4cm]{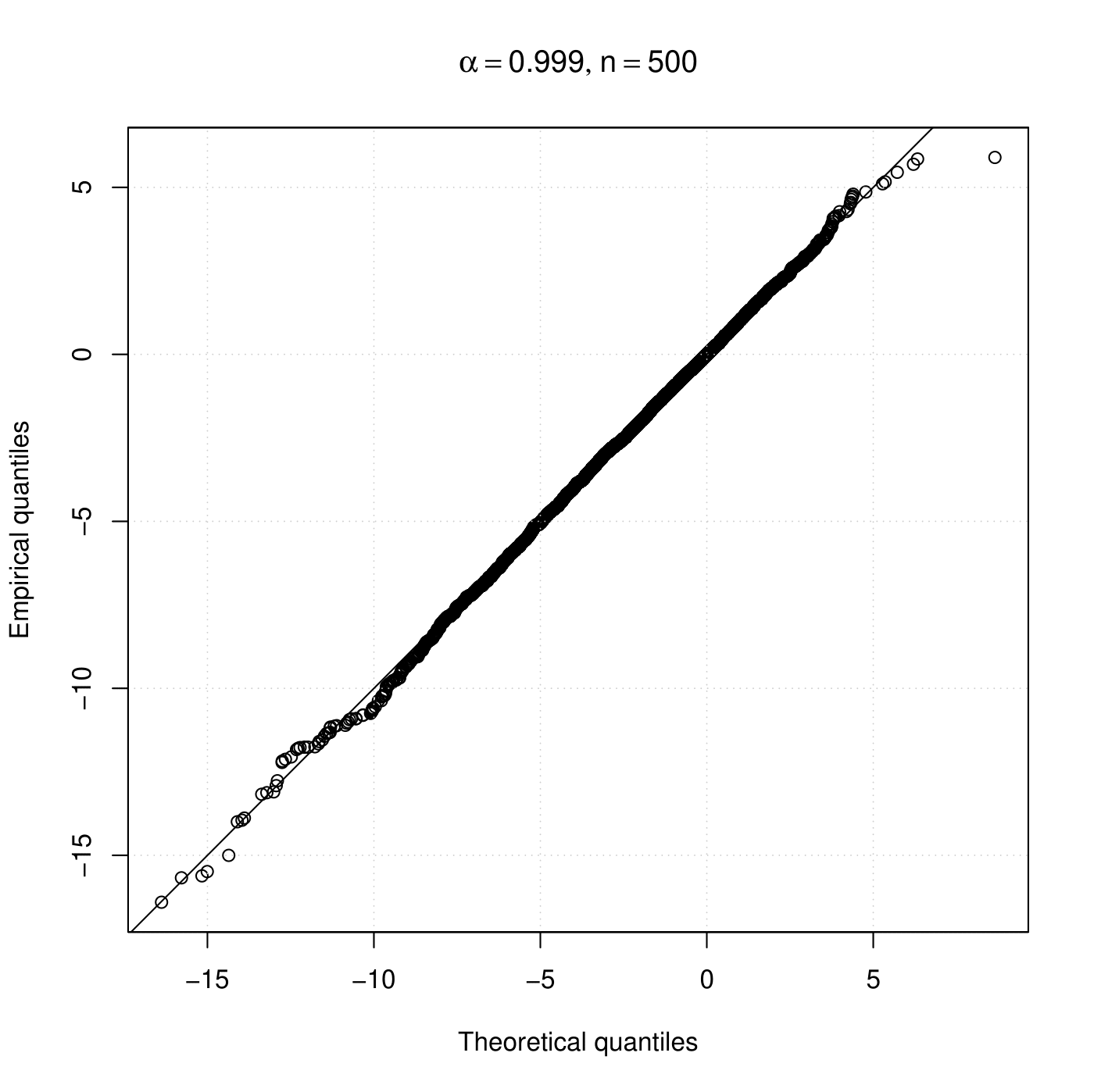}  	
\caption{Histograms and qq-plots of the standardized estimates $n(\widehat\alpha_n-\alpha)$ for $\alpha=0.98$, $\alpha=0.99$, and $\alpha=0.999$, along with their associated limiting density/quantiles given in Theorem \ref{weak_limit_nonstat} (under nearly non-stationarity). The sample size is $n=500$.}\label{fig:NONSTAT_alpha_n500}
	\end{center}
\end{figure}

A natural question is what happens when $\alpha$ is not close to 1. To address this point, we run additional simulations with $\alpha=0.7,0.8,0.9$, and the remaining settings as before. Figures \ref{fig:STAT_alpha_n500II} and \ref{fig:NONSTAT_alpha_n500II} exhibit histograms and qq-plots of the standardized CLS estimates of $\alpha$ obtained from a Monte Carlo simulation for the stationary and nearly non-stationary Poisson INARCH processes. From Figure \ref{fig:STAT_alpha_n500II}, we observe some deviation from the normality even for the case $\alpha=0.7$. This is well evidenced by the qq-plots. Surprisingly, the results based on the nearly unstable methodology work quite satisfactorily even for $\alpha=0.7$. These conclusions can be drawn again in Figure \ref{fig:NONSTAT_alpha_n500II}, where we note a good agreement between the empirical standardized CLS estimates and the theoretical asymptotic distribution derived in Theorem \ref{weak_limit_nonstat}.

All the configurations considered here are repeated again with a sample size $n=1000$.  Figures \ref{fig:STAT_alpha_n1000} and \ref{fig:NONSTAT_alpha_n1000} give us the histograms and qq-plots of the standardized CLS estimates under the stable and nearly unstable Poisson INARCH processes, respectively, under the settings $\alpha=0.98,0.99,0.999$. The plots regarding the settings $\alpha=0.7,0.8,0.9$ for the stable and nearly unstable cases are reported in Figures \ref{fig:STAT_alpha_n1000II} and \ref{fig:NONSTAT_alpha_n1000II}, respectively.

 The conclusions are quite similar to the case $n=500$ for the configurations nearly to non-stationarity $\alpha=0.98,0.99,0.999$. Regarding the configurations where $\alpha=0.7,0.8,0.9$, although there is an improvement in the results based on the stationary case (compared to $n=500$), deviations from the normality can still be observed. In contrast, the nearly unstable approach again works very well and provides the best outcomes. As a short conclusion, we recommend using the nearly unstable-based approach even when the fitted model may in practice not be too close to the non-stationarity region because the proposed methodology works well and perform better than the stationary-based approach.

\begin{figure}
	\begin{center}
		\includegraphics[width=5.4cm]{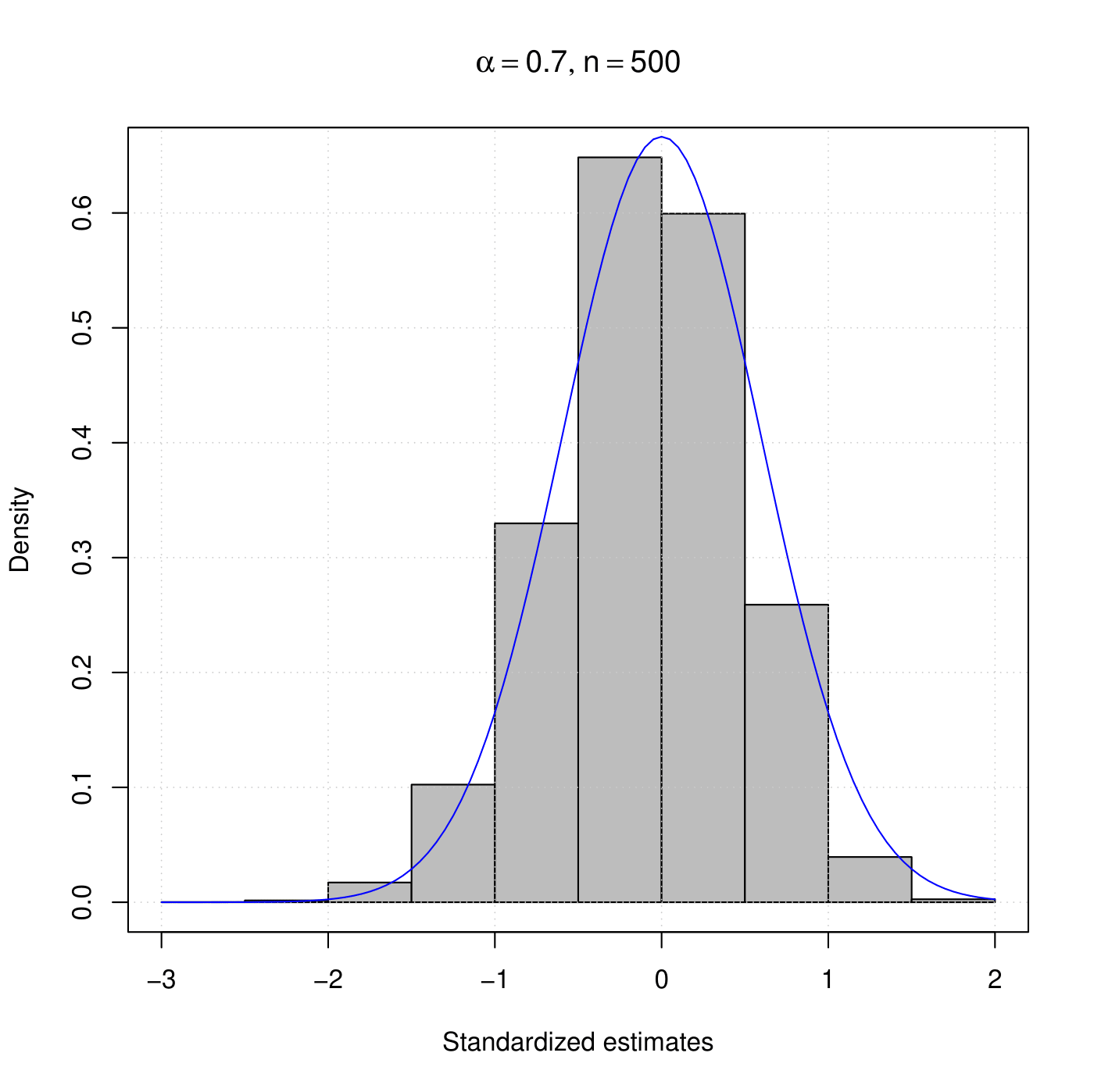}\includegraphics[width=5.4cm]{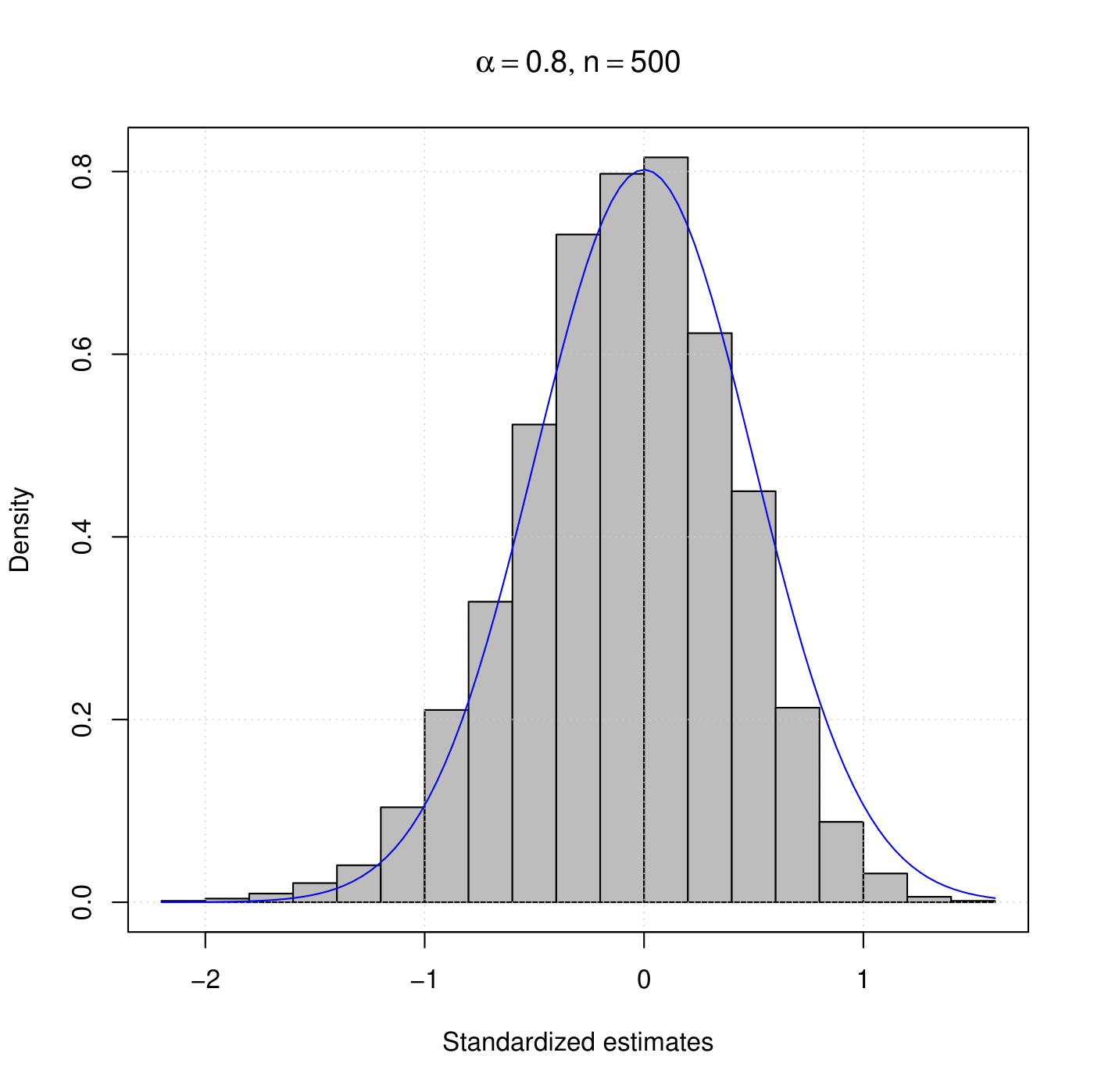}\includegraphics[width=5.4cm]{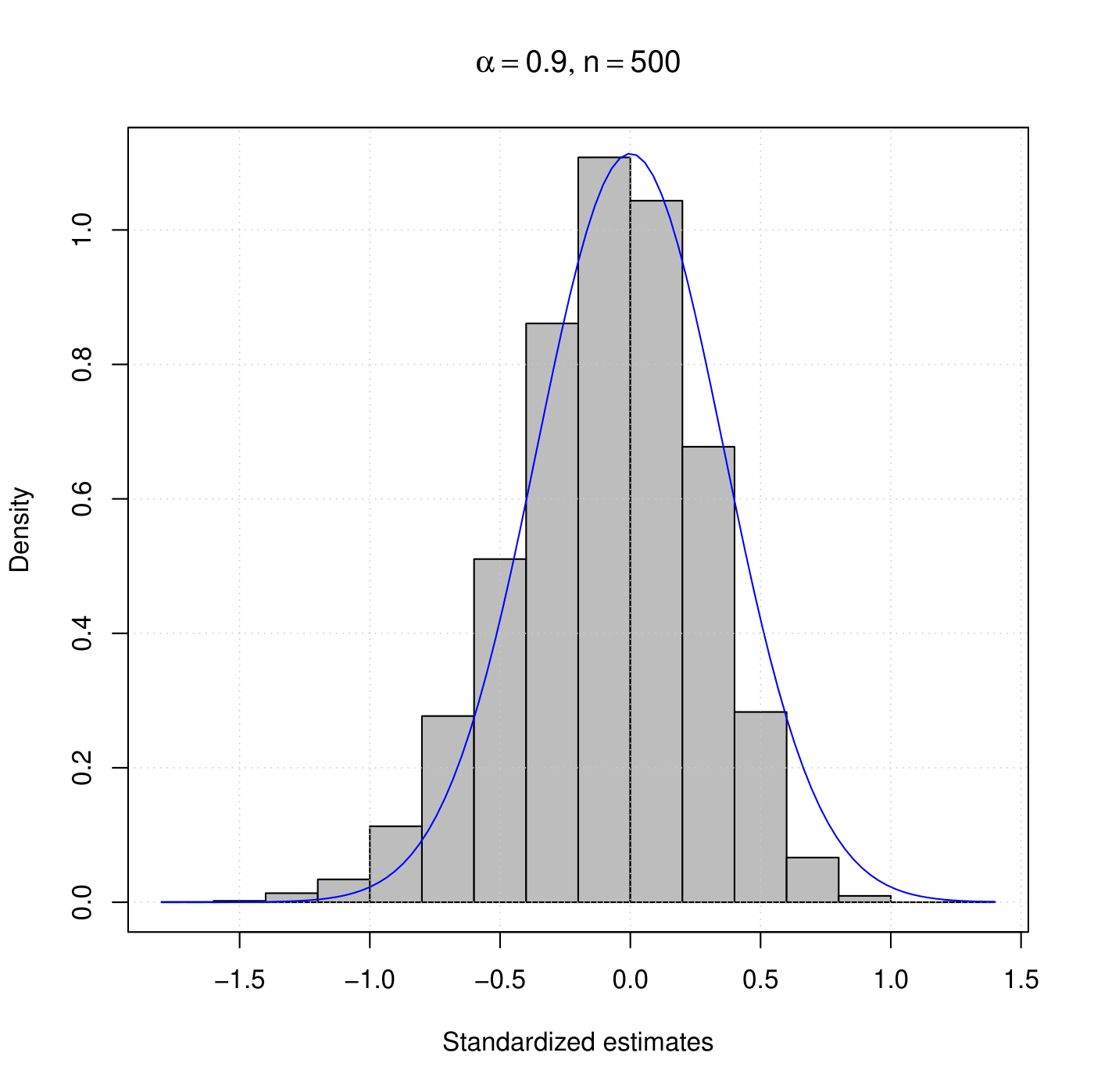}
		\includegraphics[width=5.4cm]{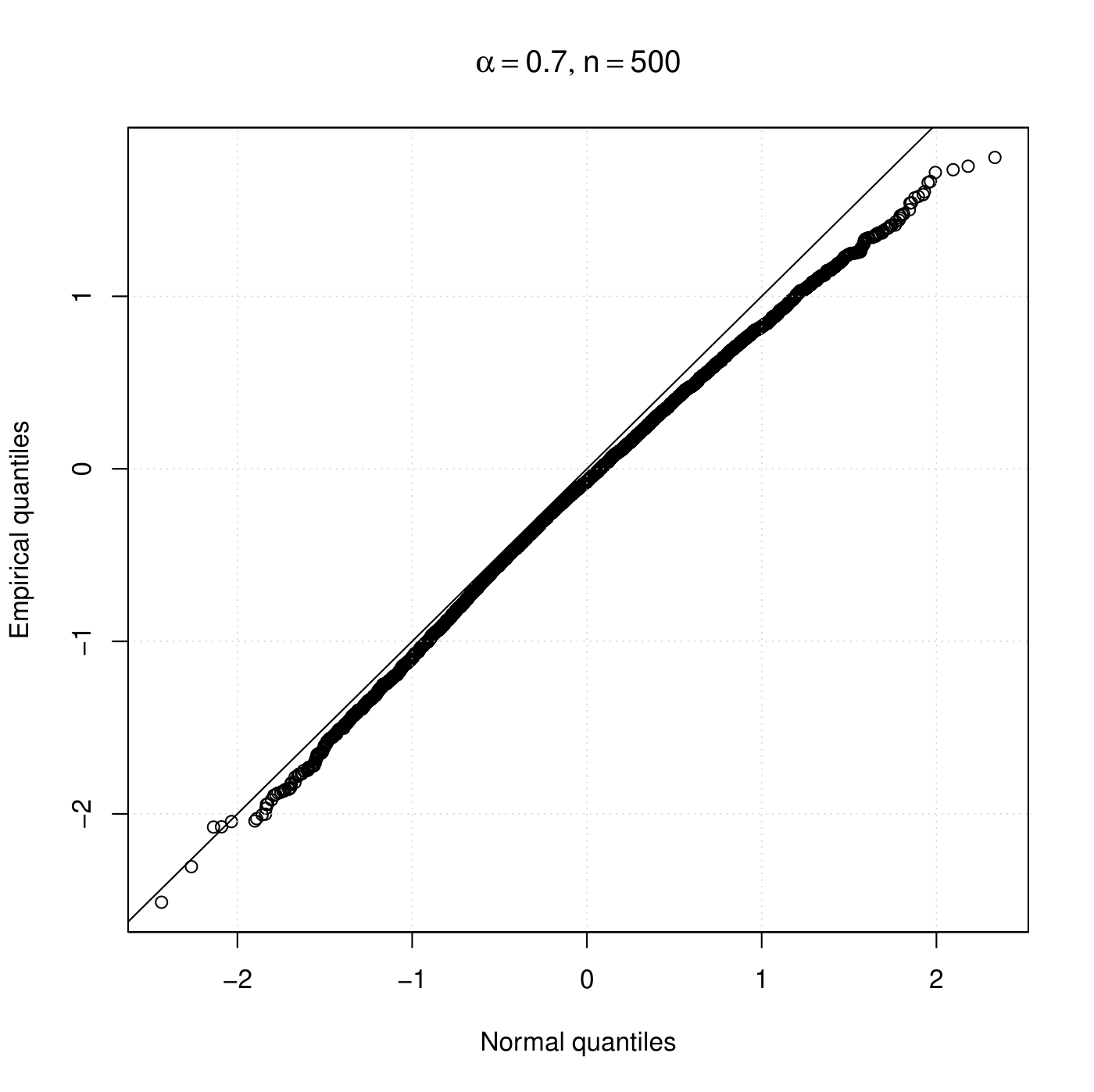}\includegraphics[width=5.4cm]{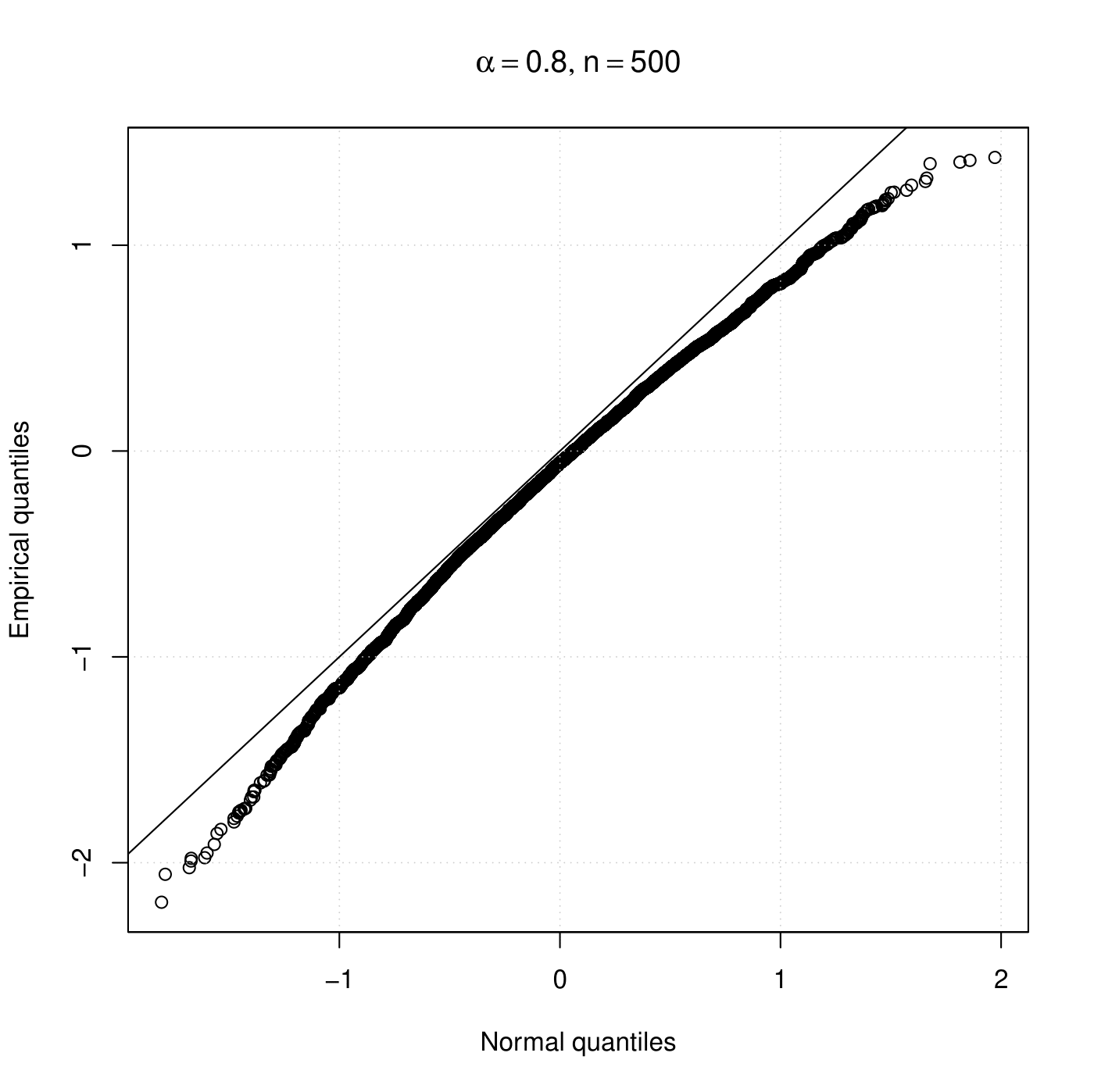}\includegraphics[width=5.4cm]{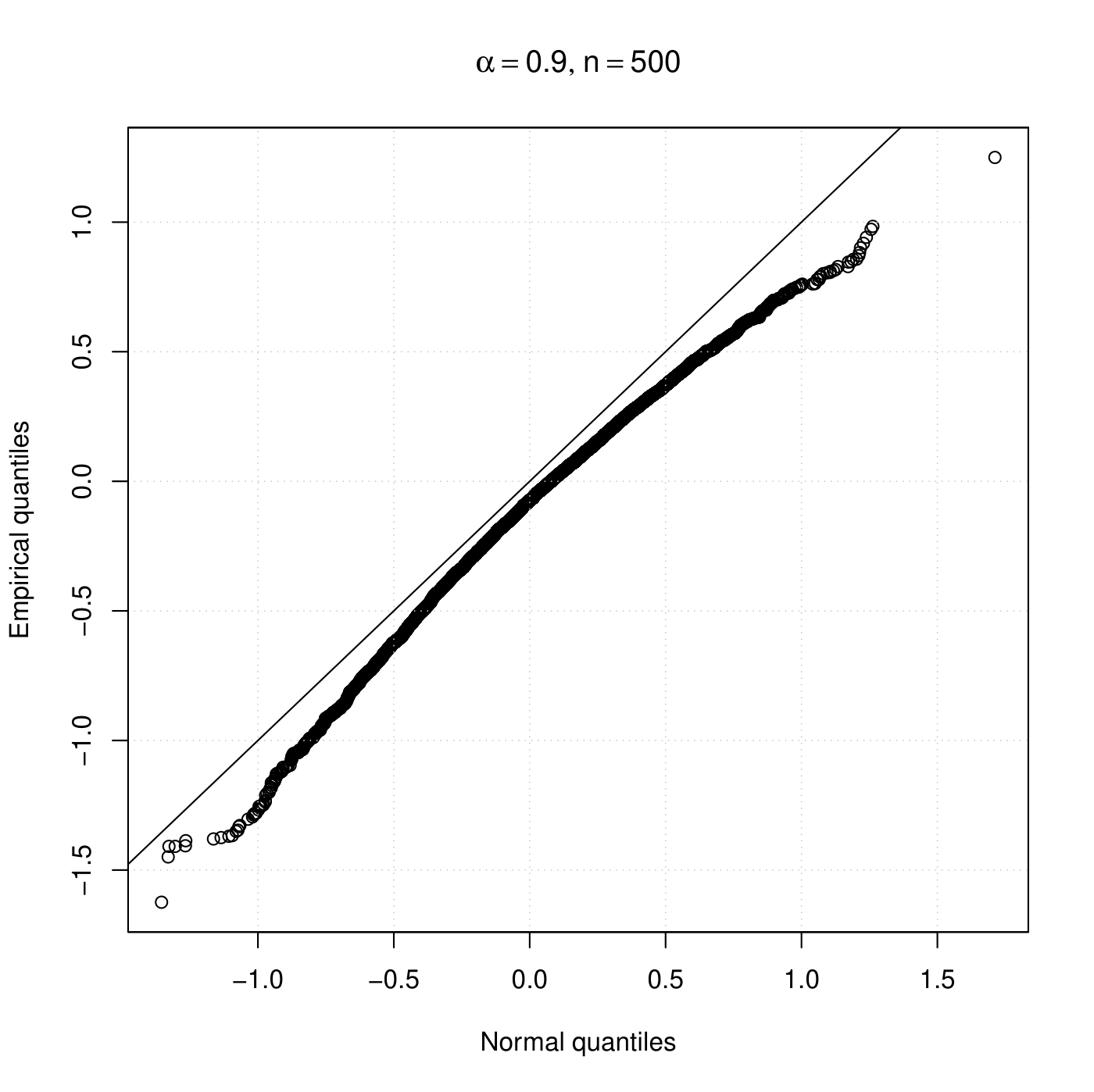}  	
		\caption{Histograms and qq-plots of the standardized estimates $\sqrt{n}(\widehat\alpha_n-\alpha)$ for $\alpha=0.7$, $\alpha=0.8$, and $\alpha=0.9$, along with their associated limiting normal density/quantiles given in Theorem \ref{weak_limit_stat} (under stationarity). The sample size is $n=500$.}\label{fig:STAT_alpha_n500II}
		\includegraphics[width=5.4cm]{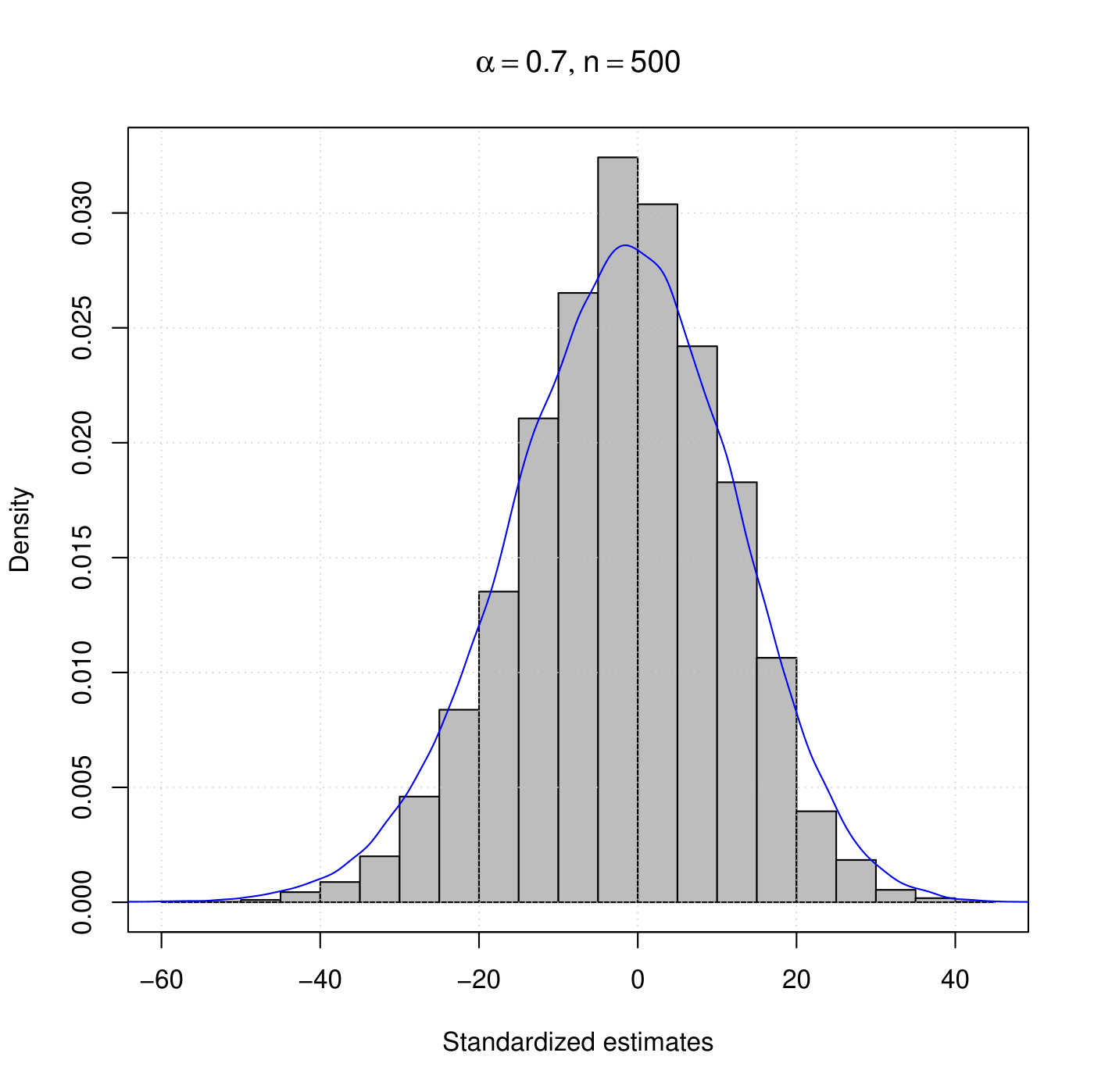}\includegraphics[width=5.4cm]{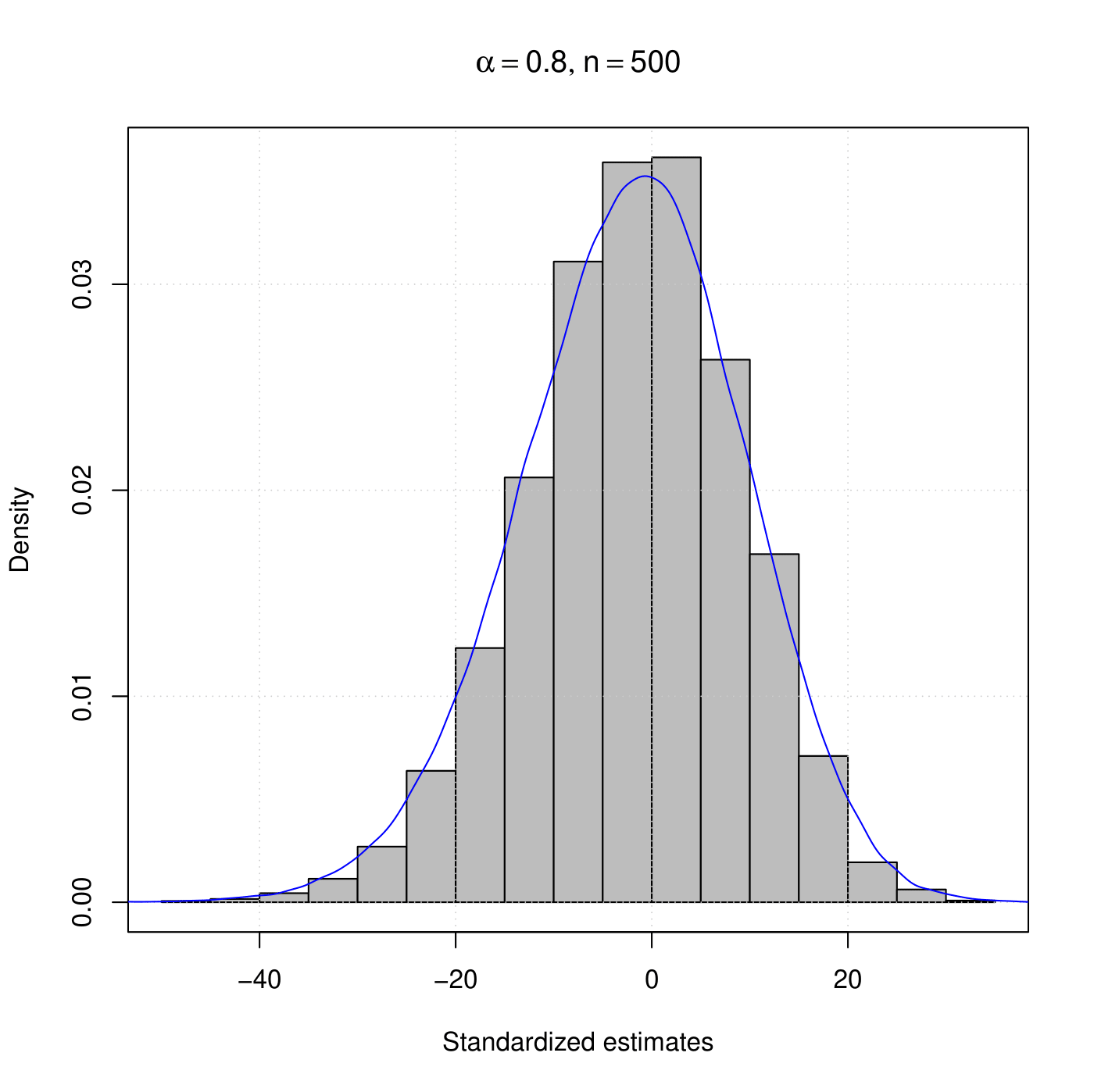}\includegraphics[width=5.4cm]{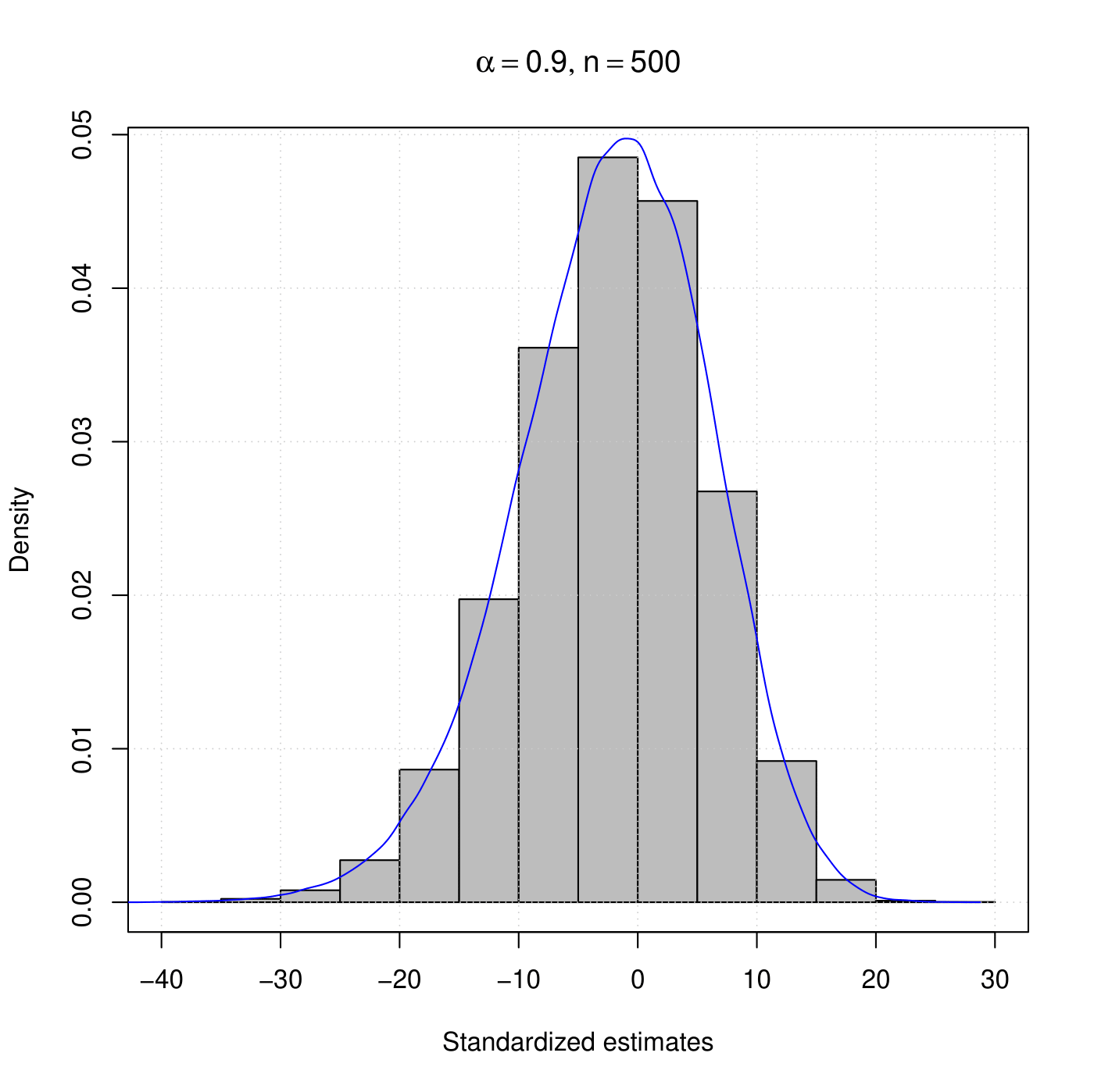}
		\includegraphics[width=5.4cm]{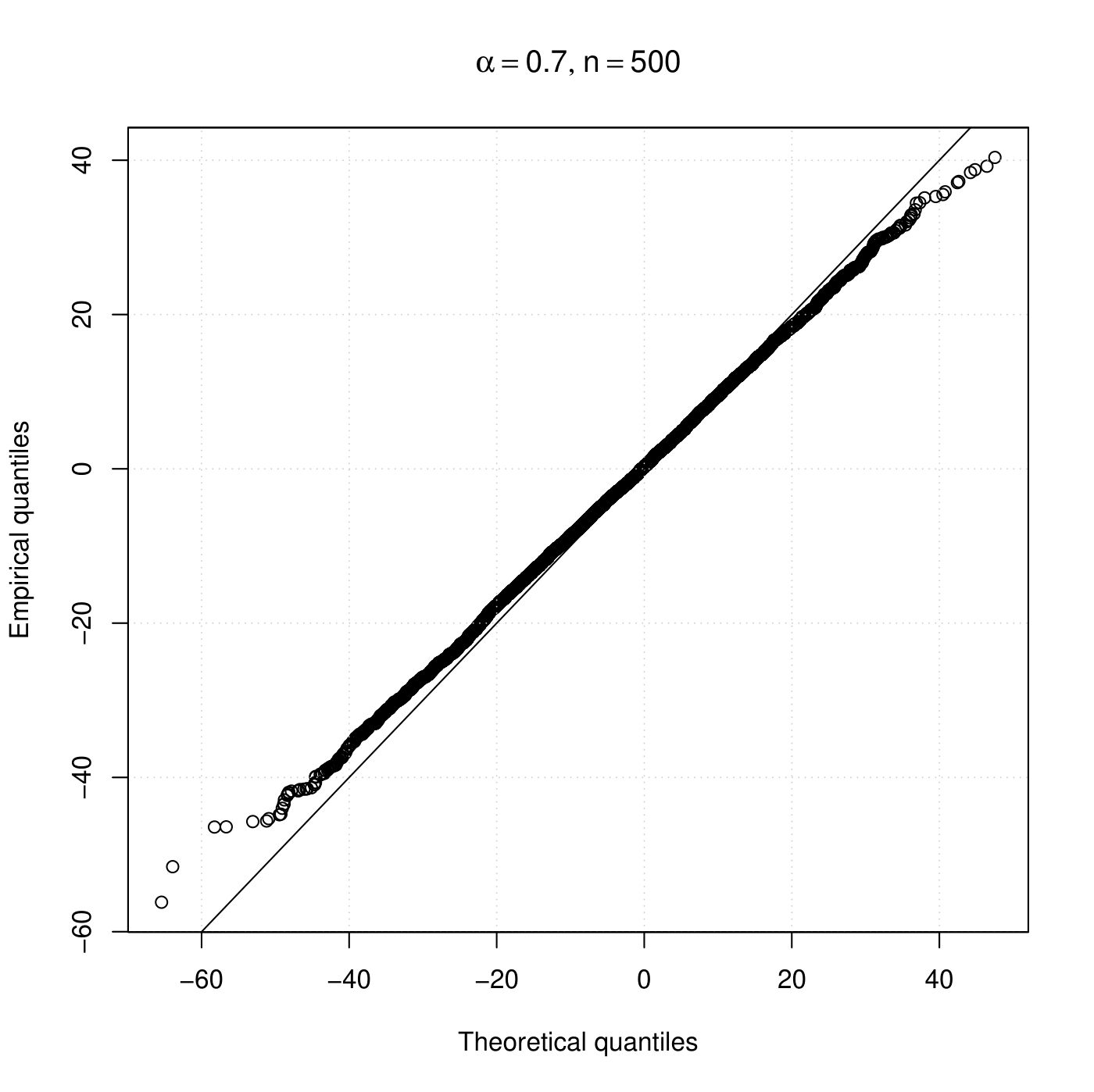}\includegraphics[width=5.4cm]{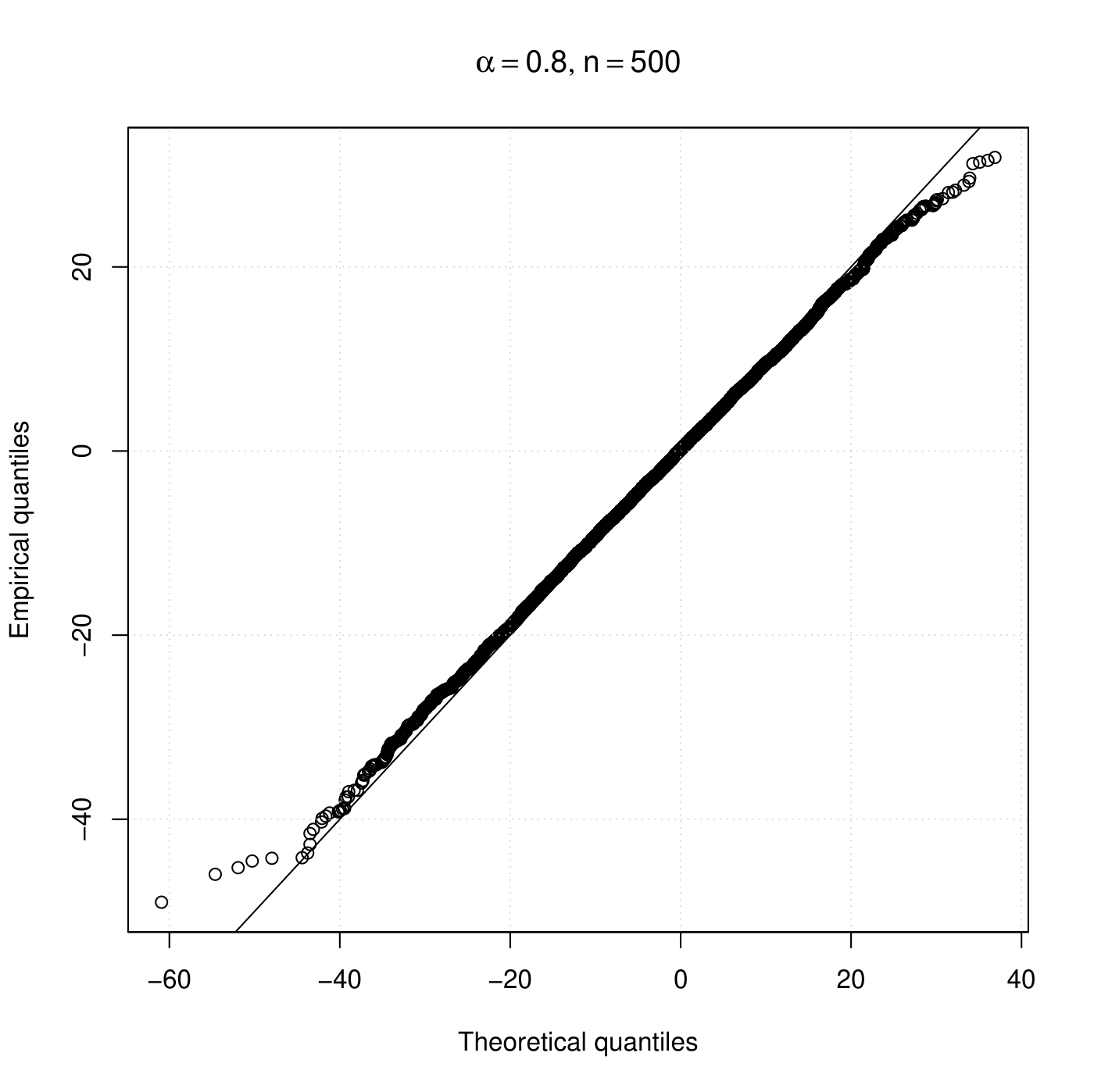}\includegraphics[width=5.4cm]{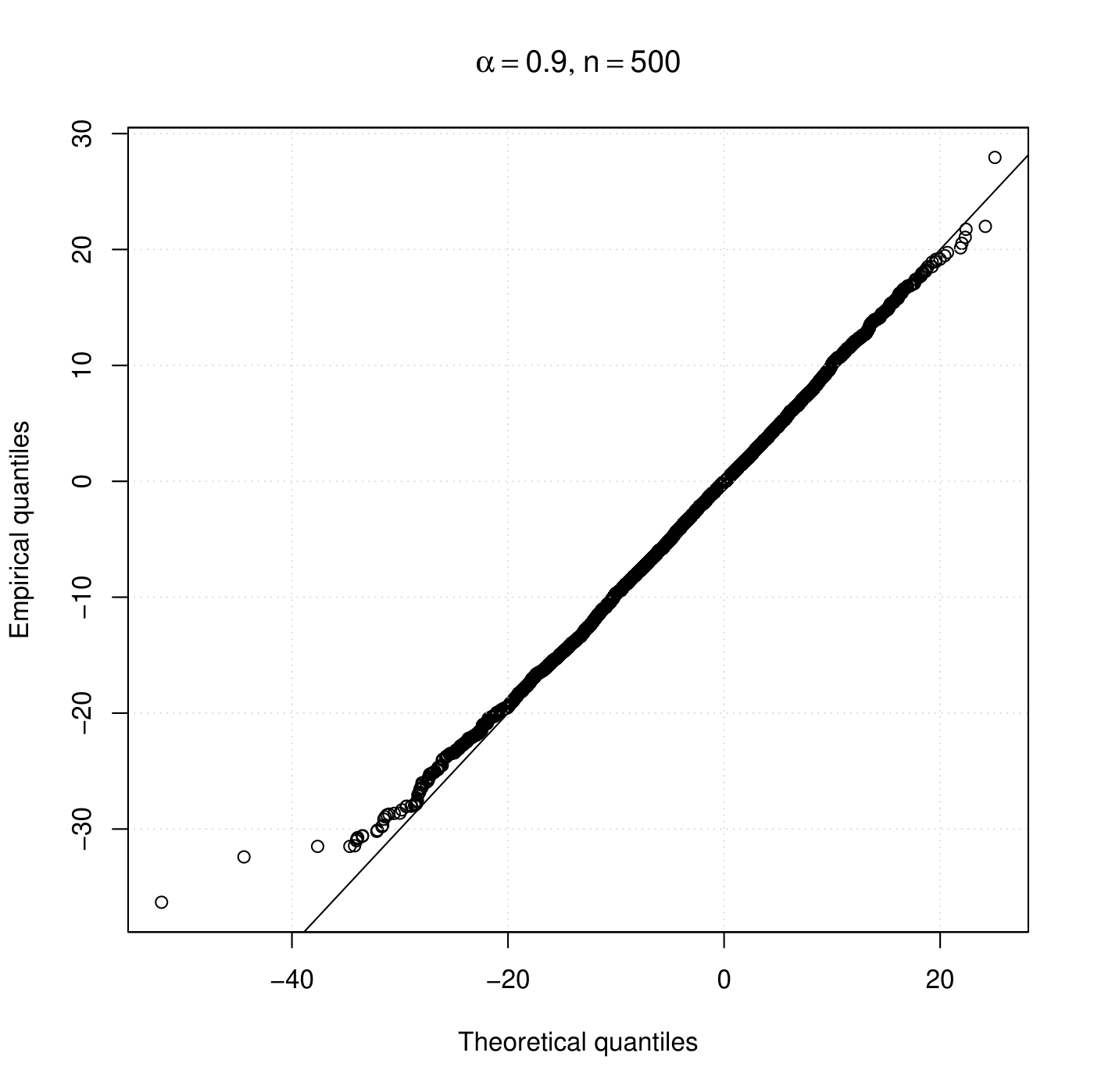}  	
		\caption{Histograms and qq-plots of the standardized estimates $n(\widehat\alpha_n-\alpha)$ for $\alpha=0.7$, $\alpha=0.8$, and $\alpha=0.9$, along with their associated limiting density/quantiles given in Theorem \ref{weak_limit_nonstat} (under nearly non-stationarity). The sample size is $n=500$.}\label{fig:NONSTAT_alpha_n500II}
	\end{center}
\end{figure}

\begin{figure}
	\begin{center}
		\includegraphics[width=5.4cm]{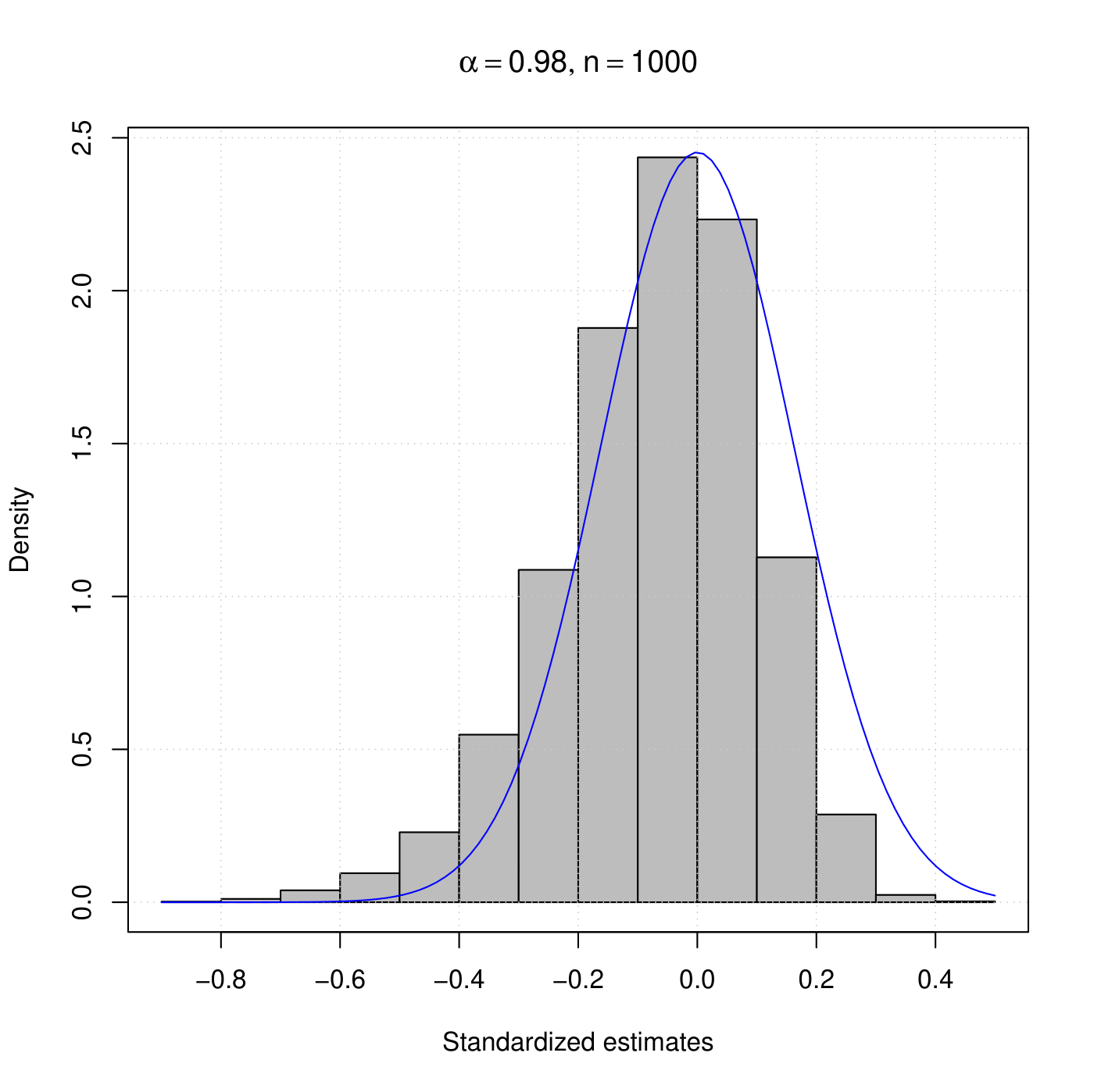}\includegraphics[width=5.4cm]{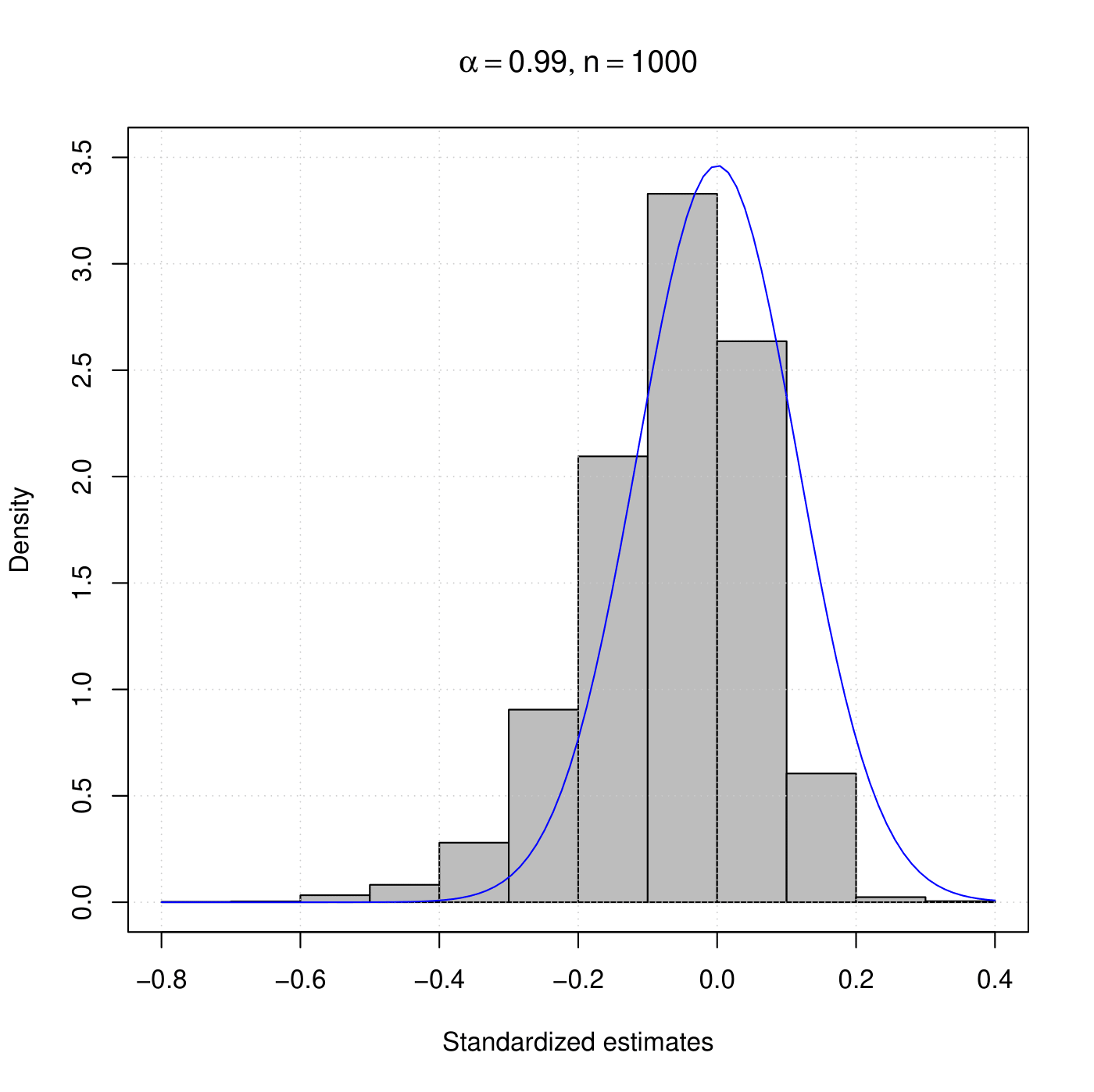}\includegraphics[width=5.4cm]{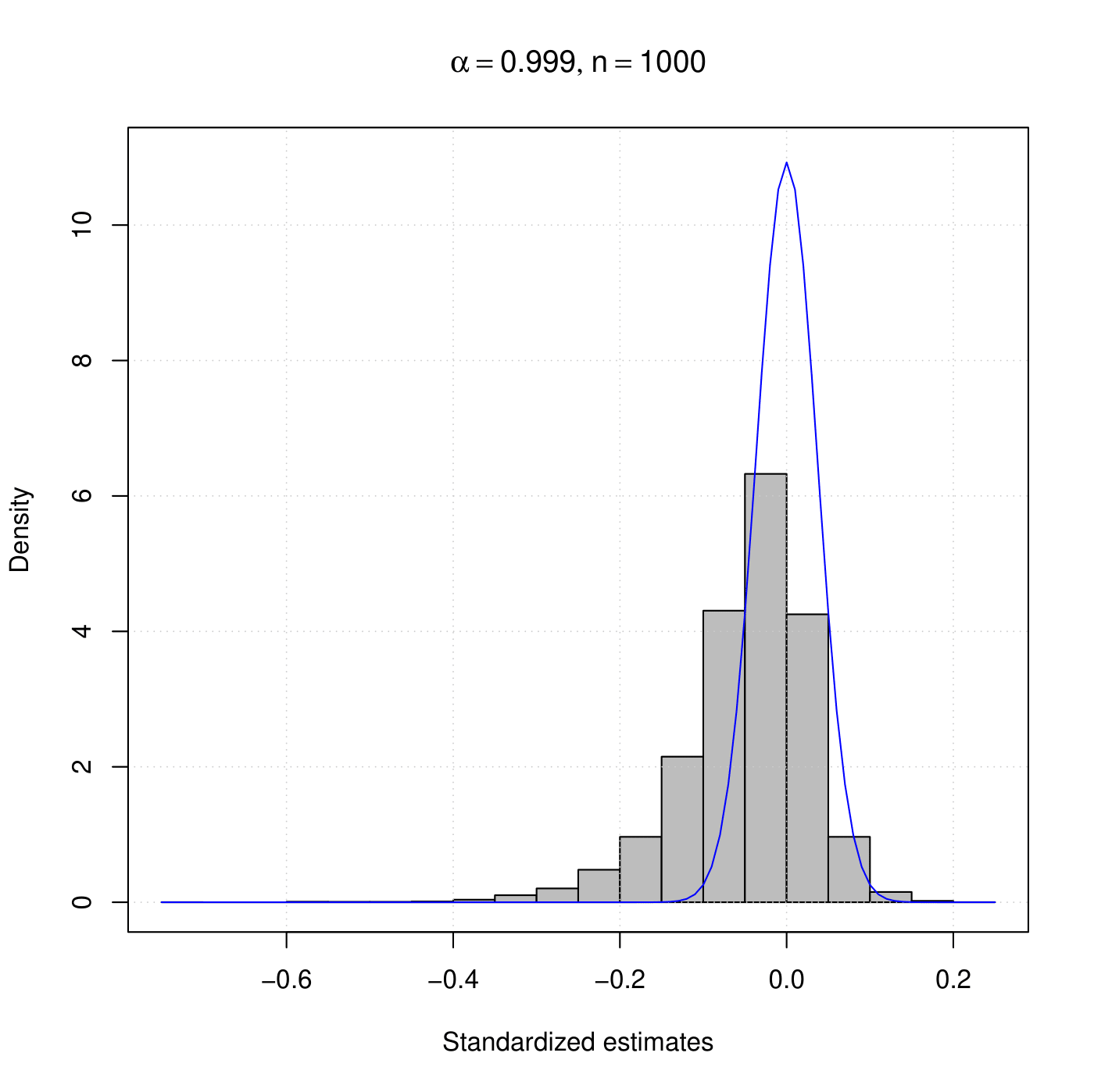}
		\includegraphics[width=5.4cm]{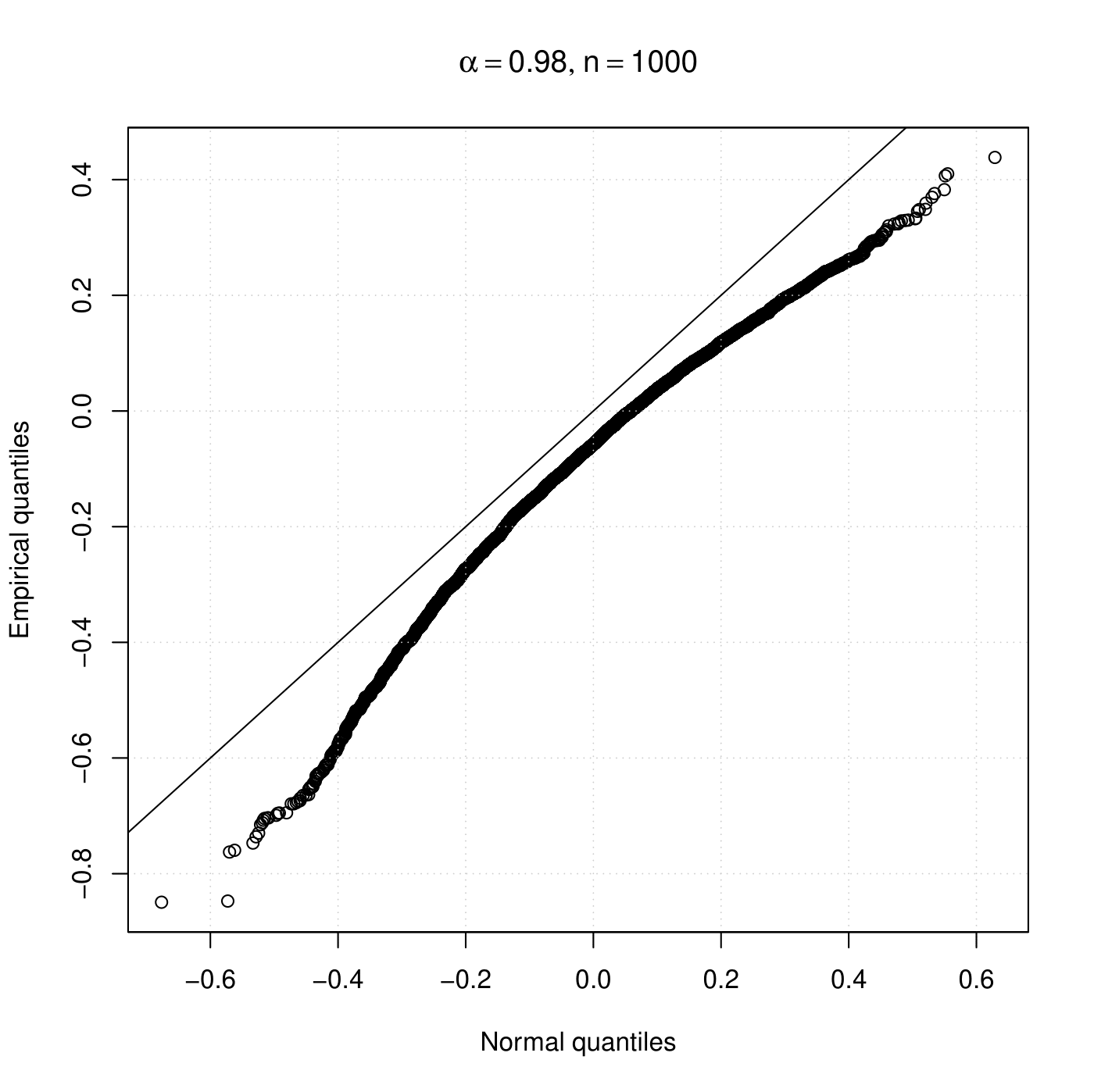}\includegraphics[width=5.4cm]{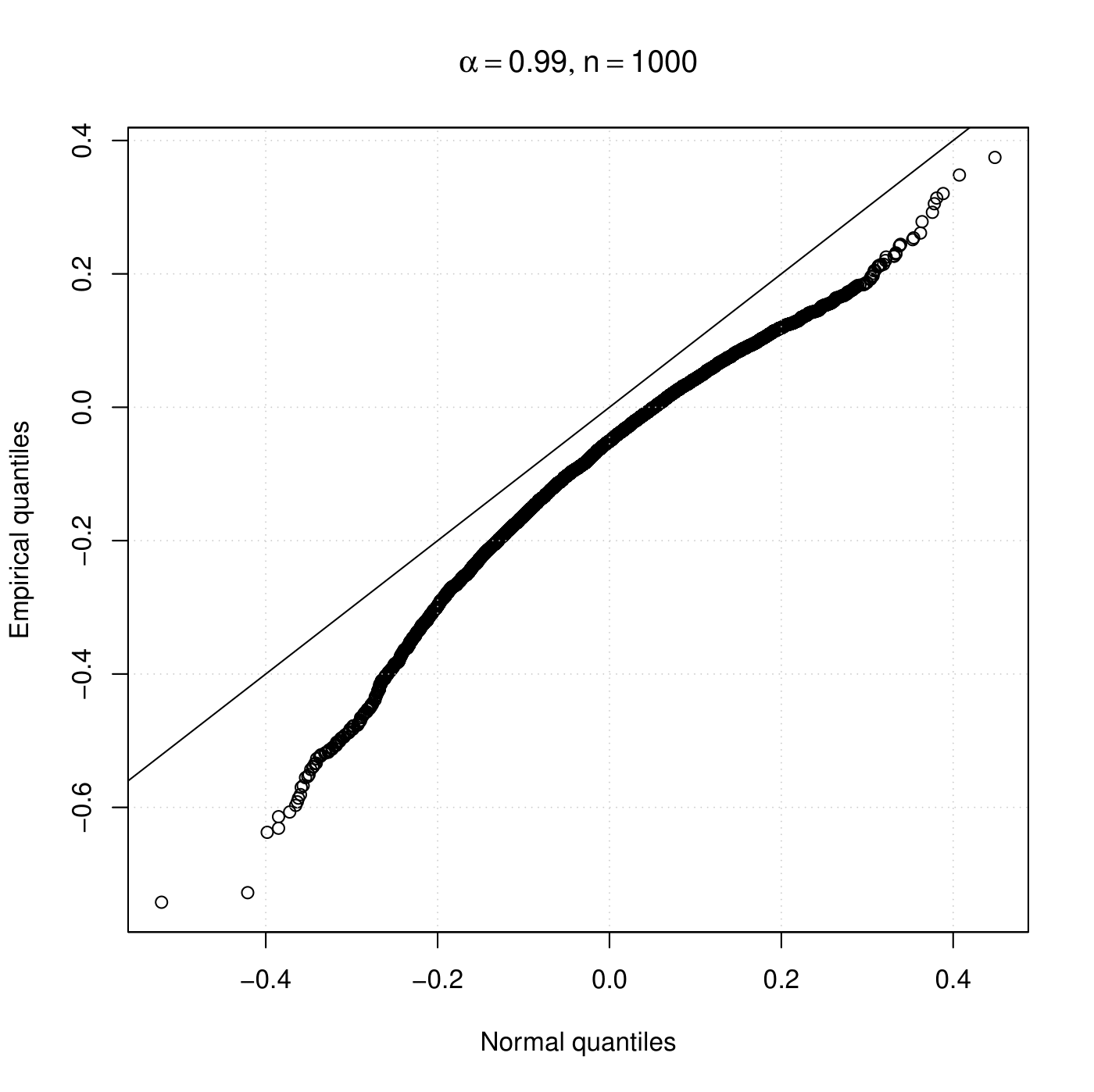}\includegraphics[width=5.4cm]{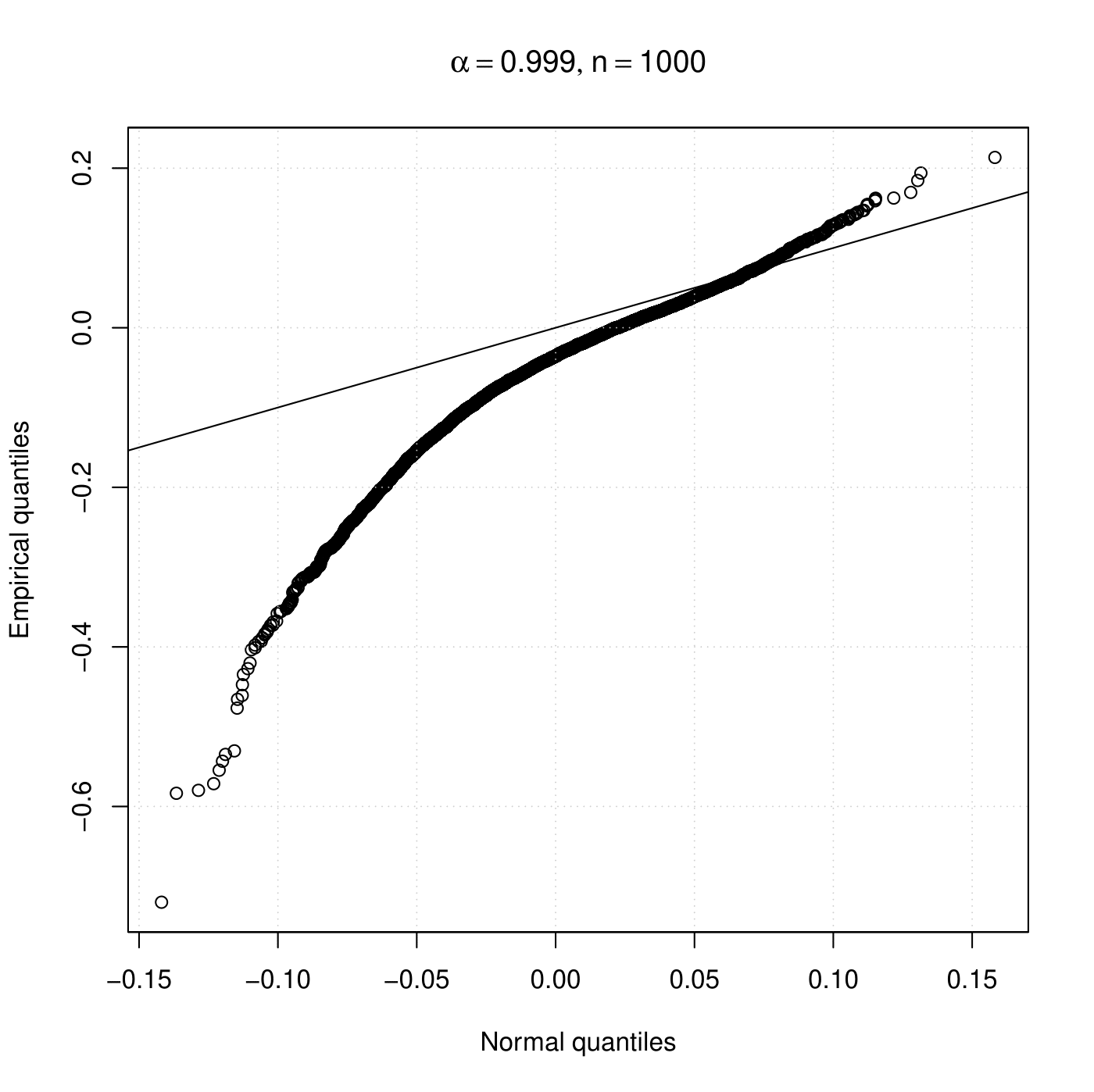}  	
		\caption{Histograms and qq-plots of the standardized estimates $\sqrt{n}(\widehat\alpha_n-\alpha)$ for $\alpha=0.98$, $\alpha=0.99$, and $\alpha=0.999$, along with their associated limiting normal density/quantiles given in Theorem \ref{weak_limit_stat} (under stationarity). The sample size is $n=1000$.}\label{fig:STAT_alpha_n1000}
		\includegraphics[width=5.4cm]{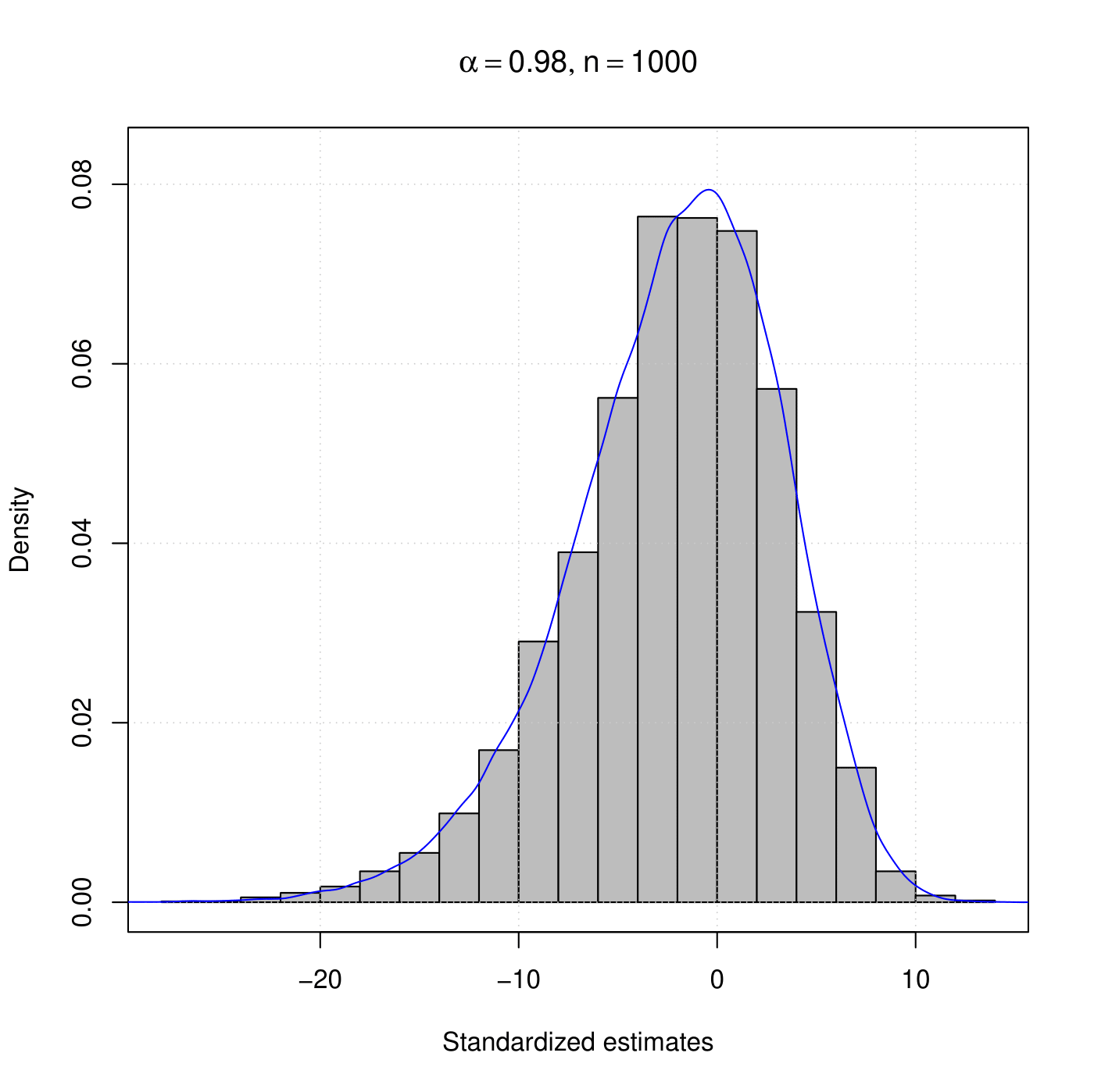}\includegraphics[width=5.4cm]{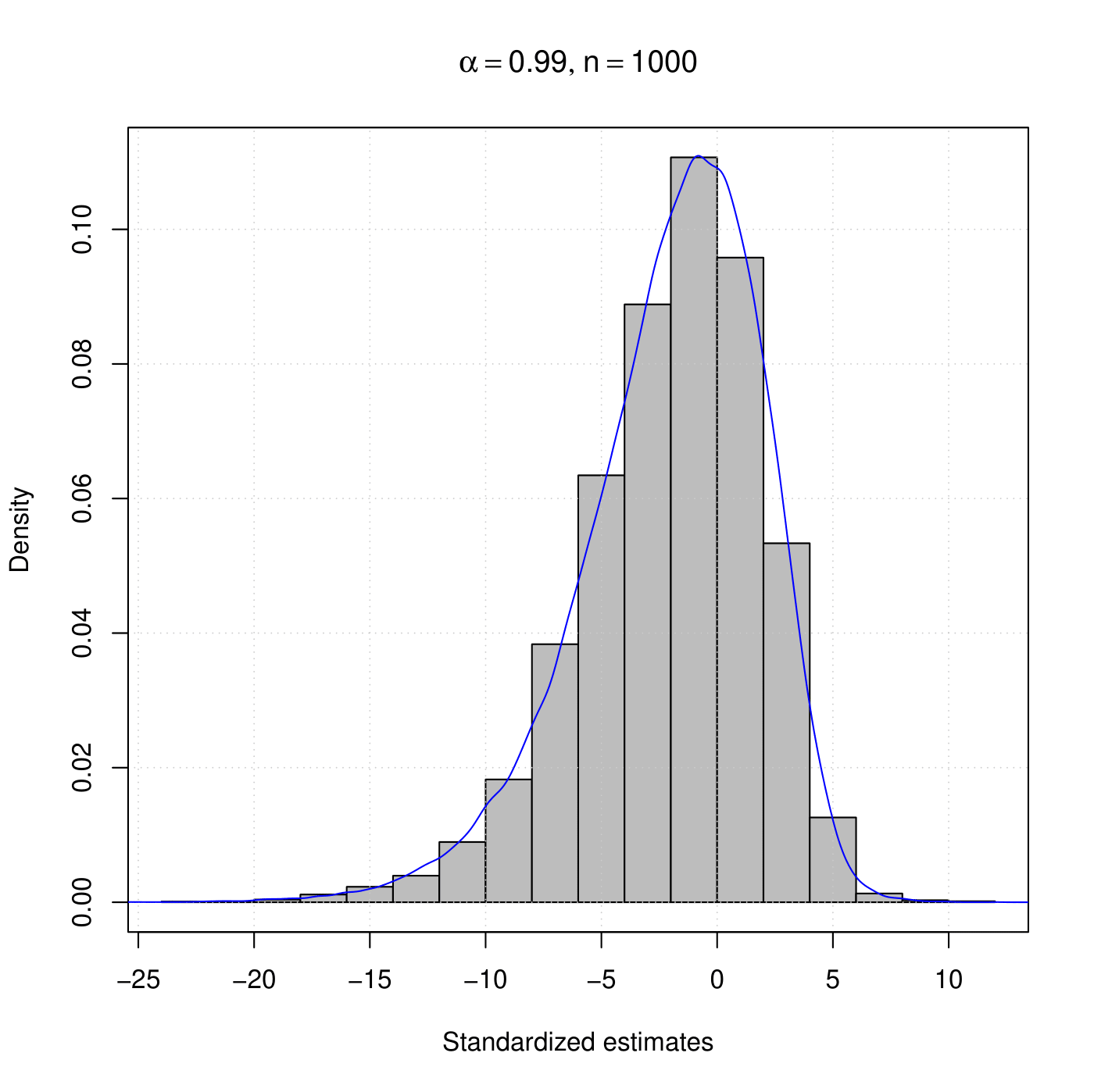}\includegraphics[width=5.4cm]{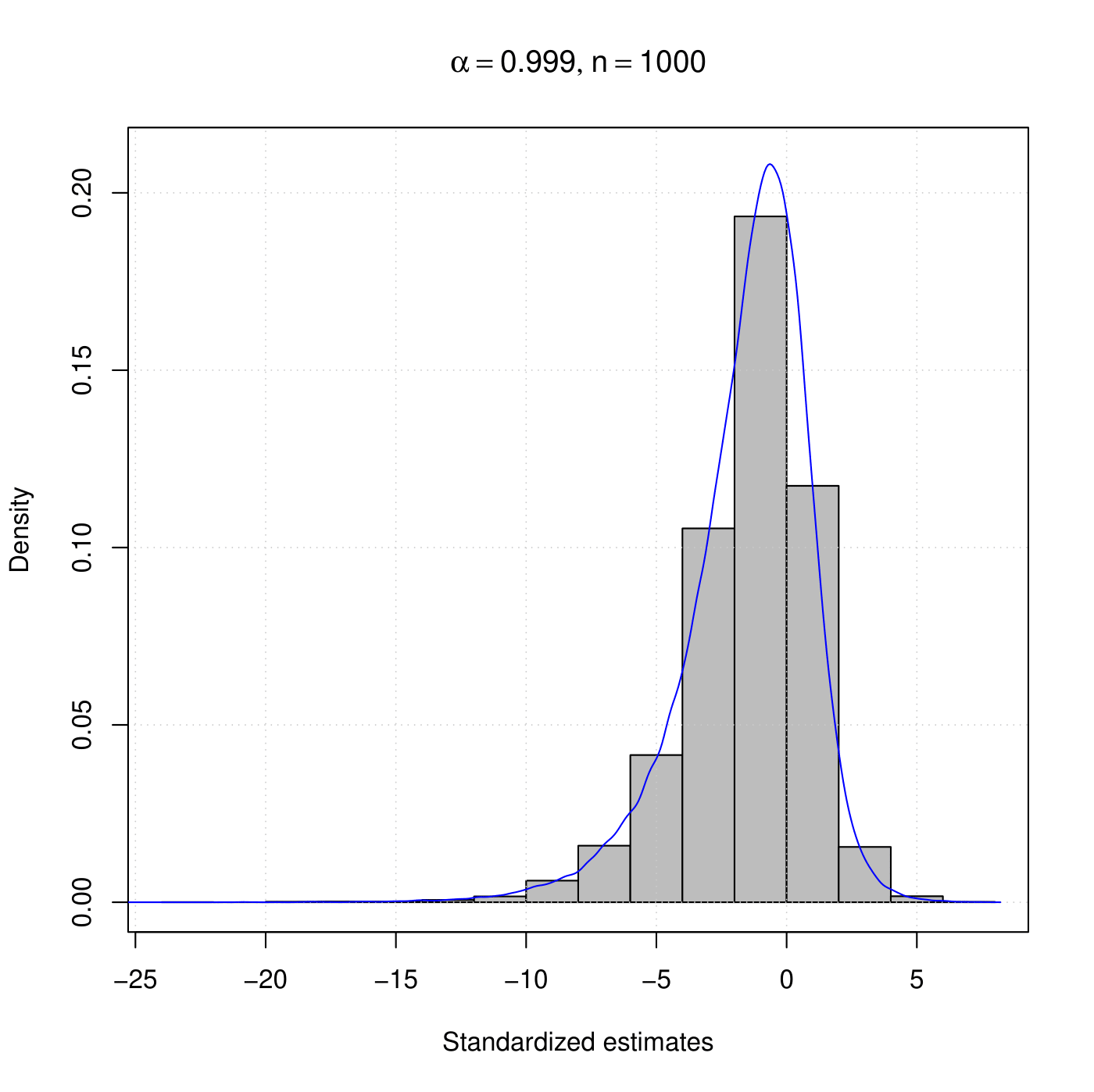}
		\includegraphics[width=5.4cm]{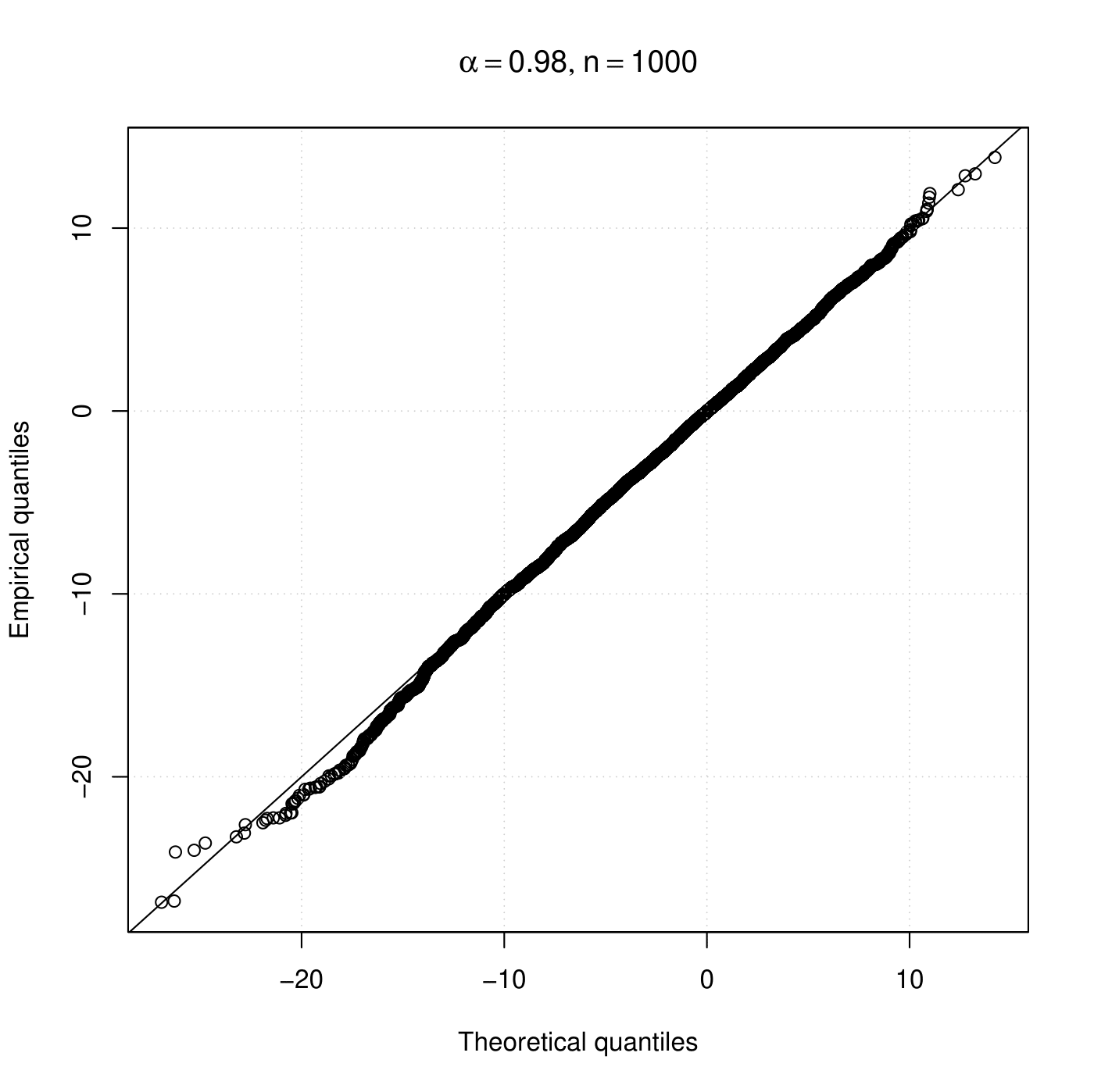}\includegraphics[width=5.4cm]{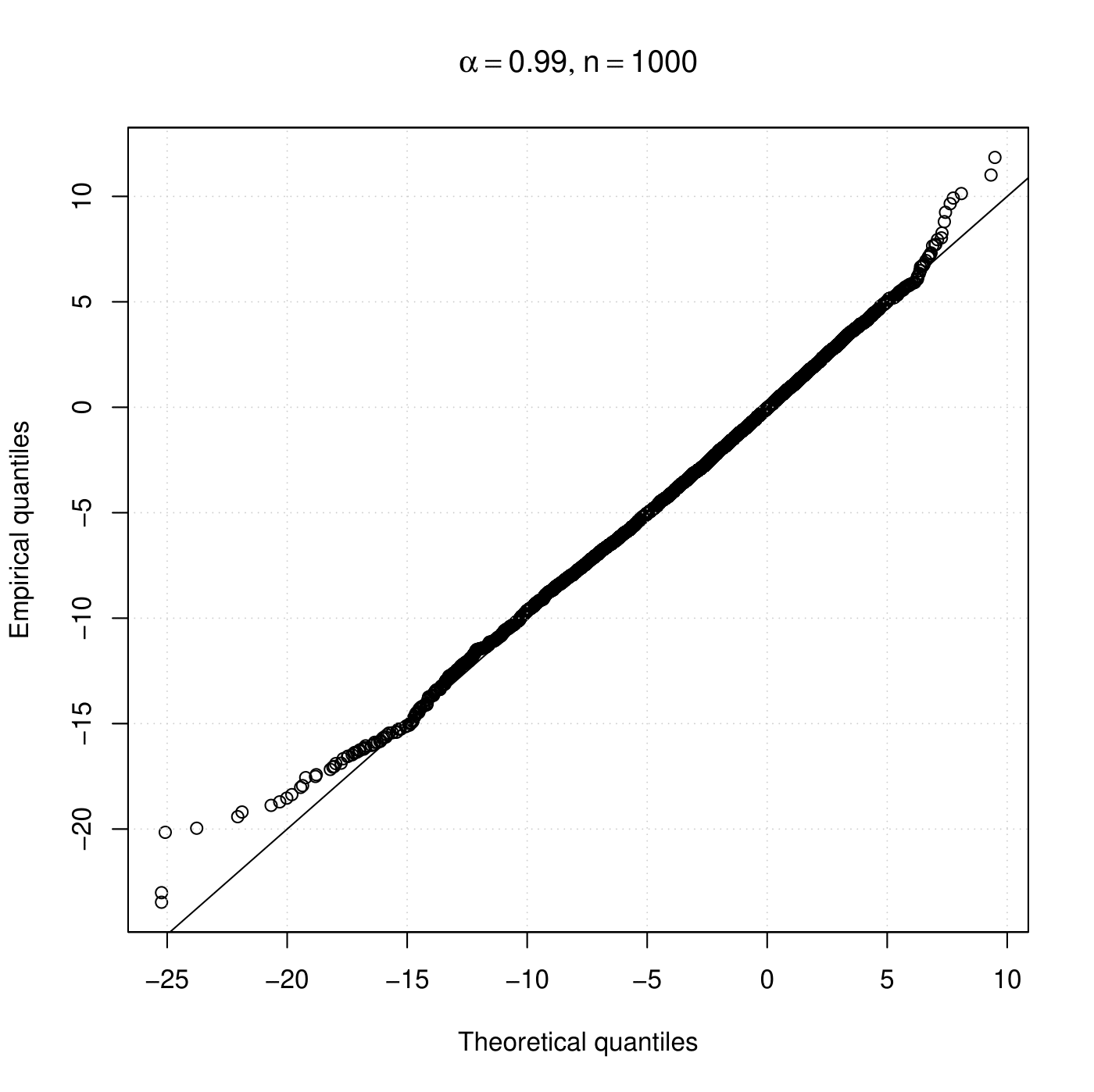}\includegraphics[width=5.4cm]{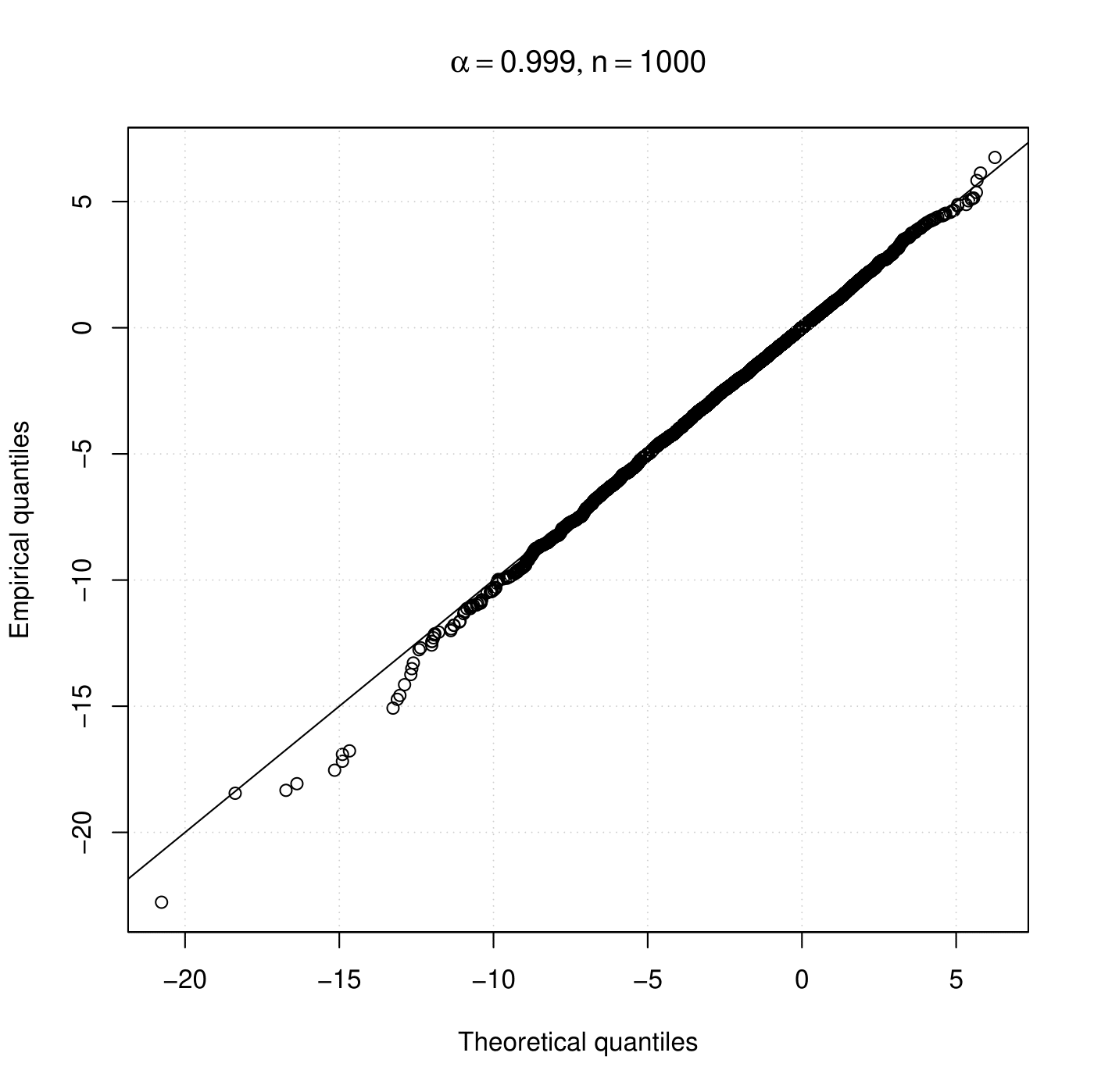}  	
		\caption{Histograms and qq-plots of the standardized estimates $n(\widehat\alpha_n-\alpha)$ for $\alpha=0.98$, $\alpha=0.99$, and $\alpha=0.999$, along with their associated limiting density/quantiles given in Theorem \ref{weak_limit_nonstat} (under nearly non-stationarity). The sample size is $n=1000$.}\label{fig:NONSTAT_alpha_n1000}
	\end{center}
\end{figure}

\begin{figure}
	\begin{center}
		\includegraphics[width=5.4cm]{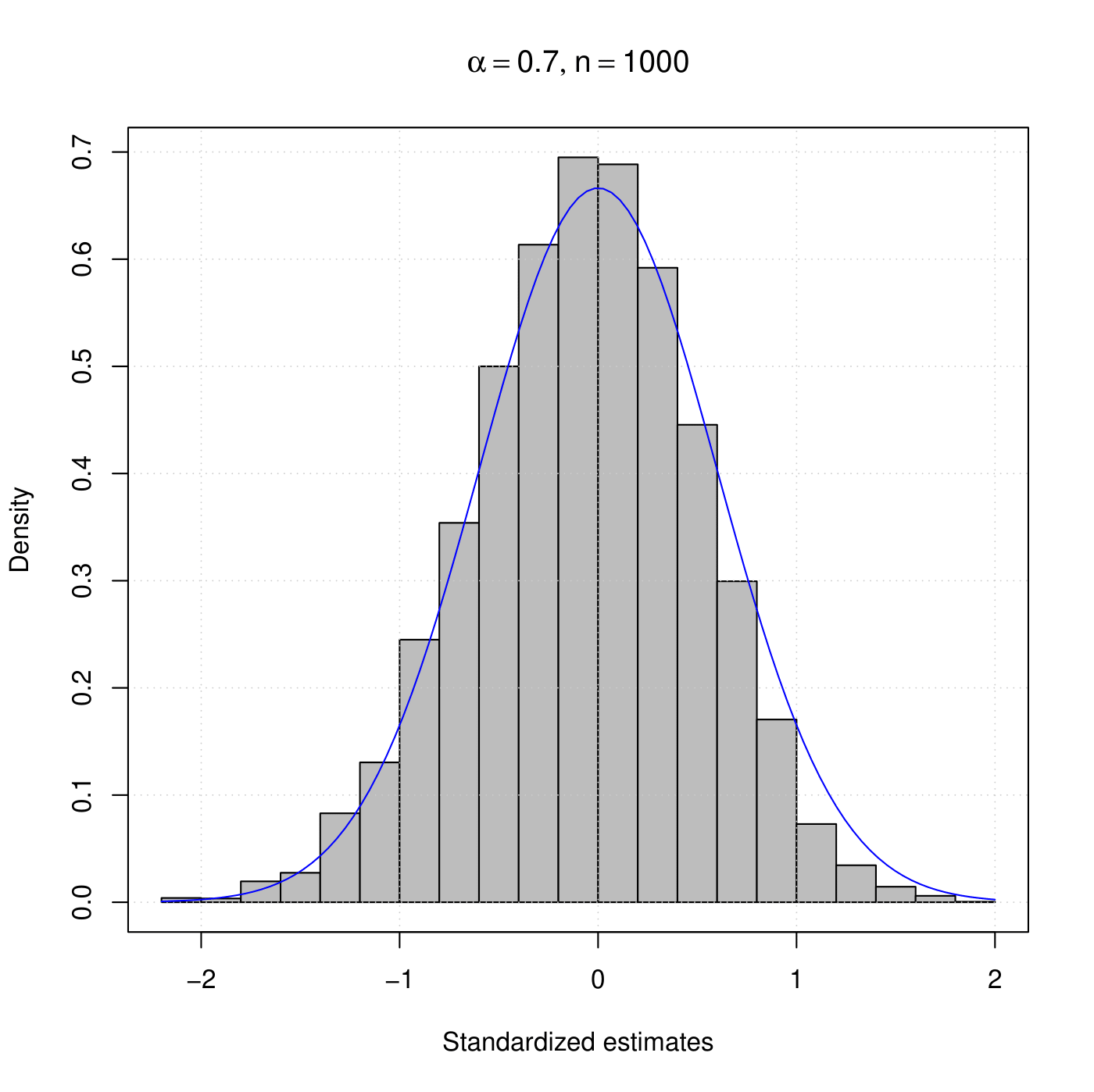}\includegraphics[width=5.4cm]{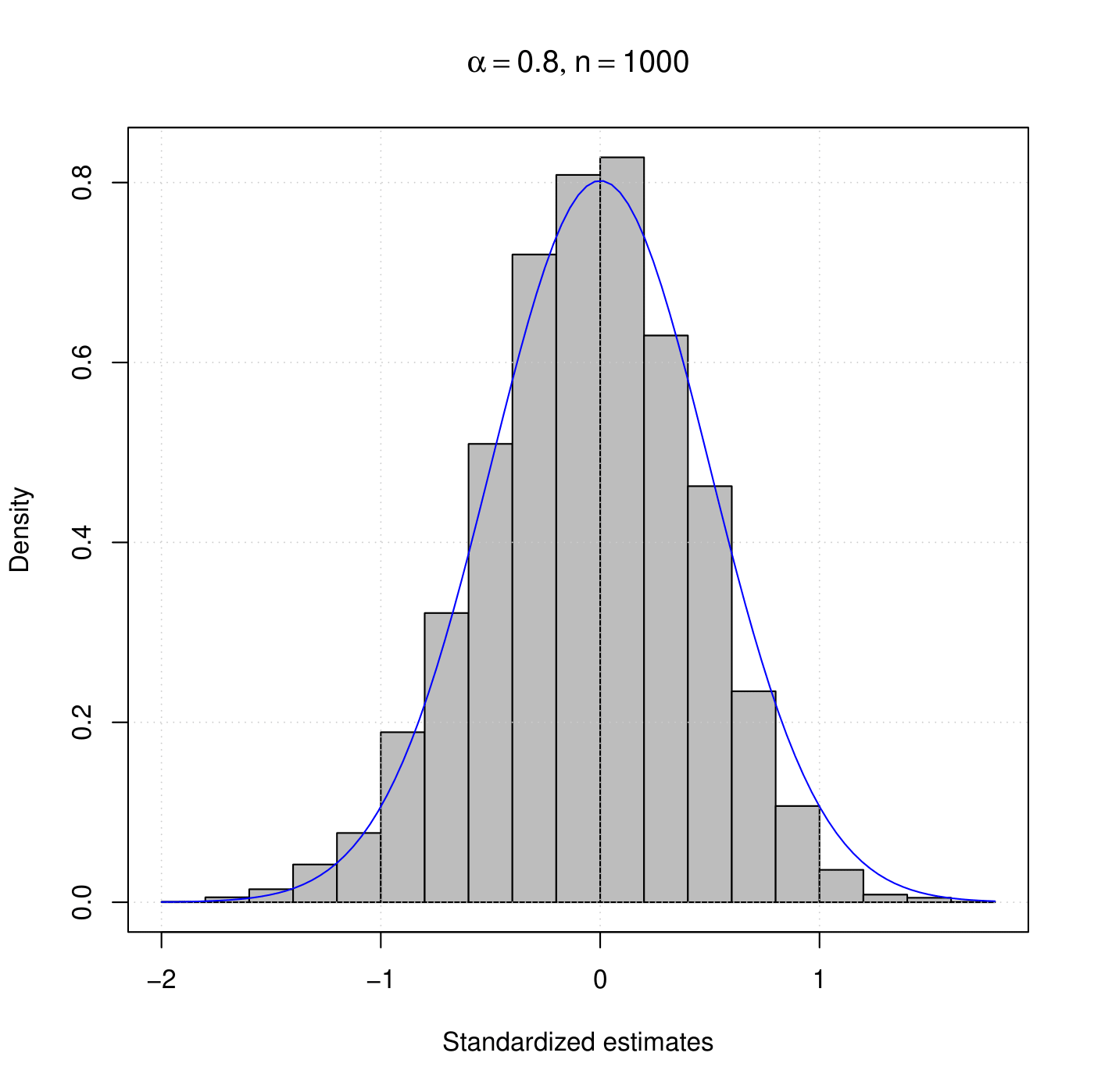}\includegraphics[width=5.4cm]{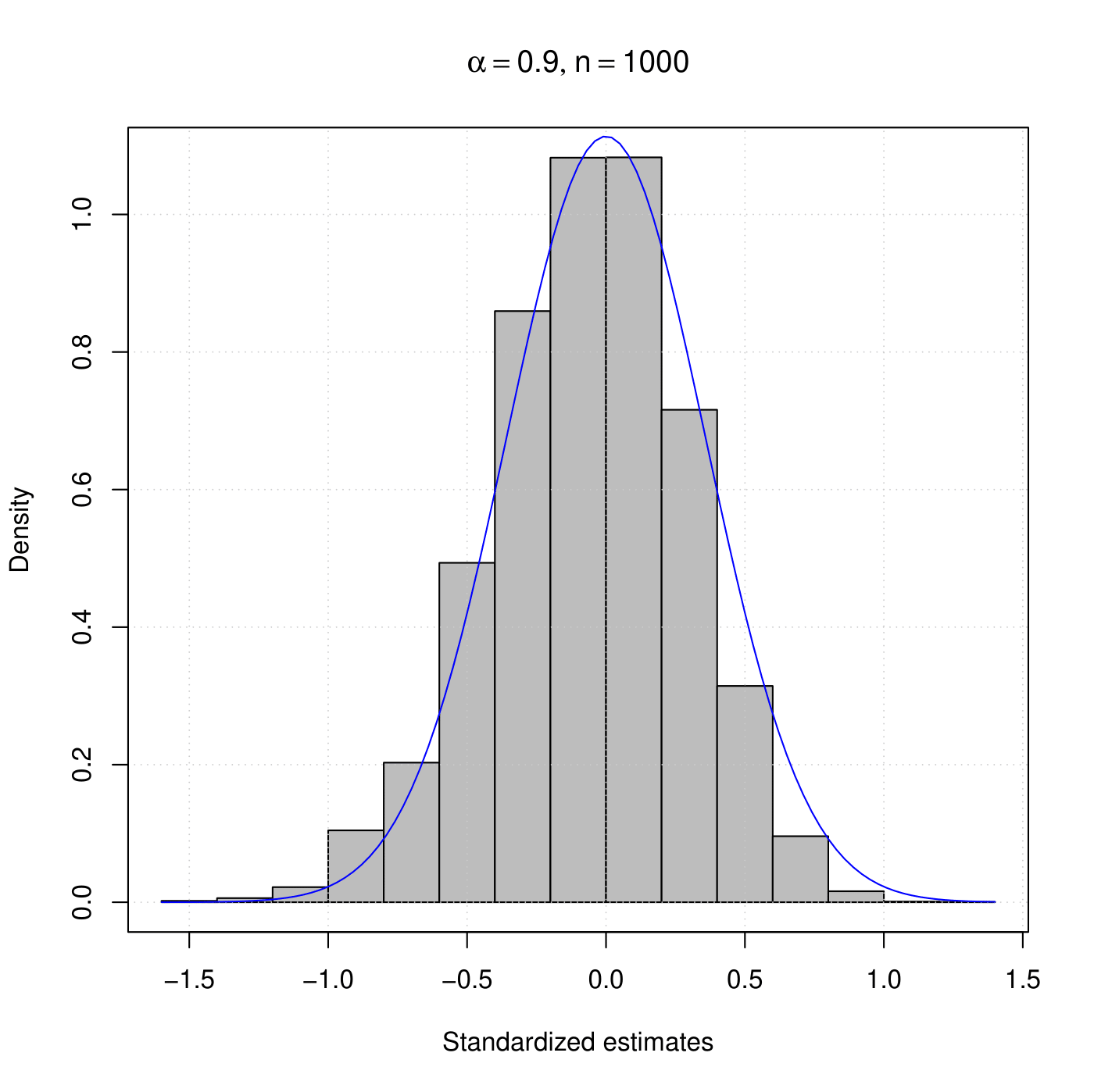}
		\includegraphics[width=5.4cm]{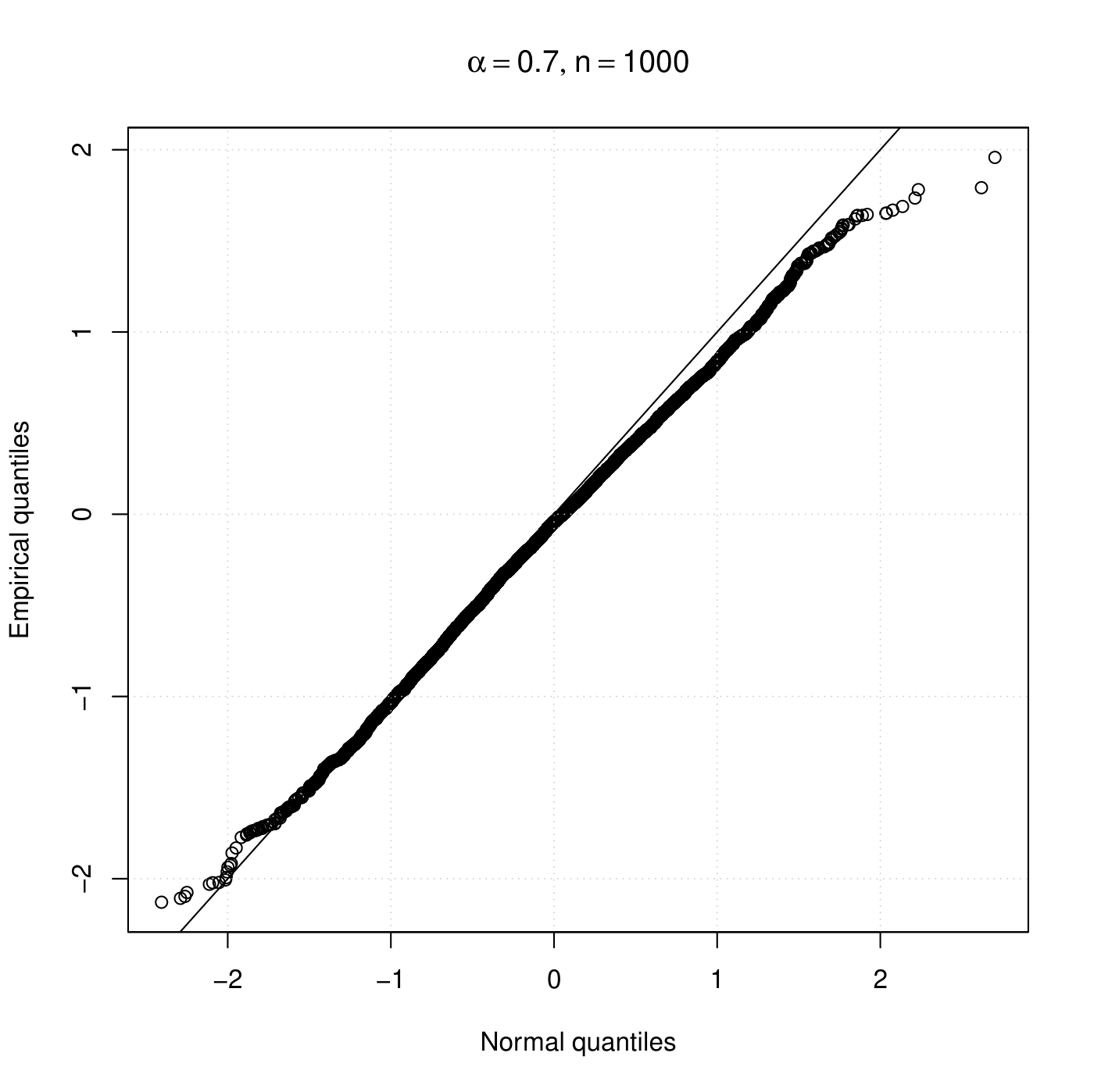}\includegraphics[width=5.4cm]{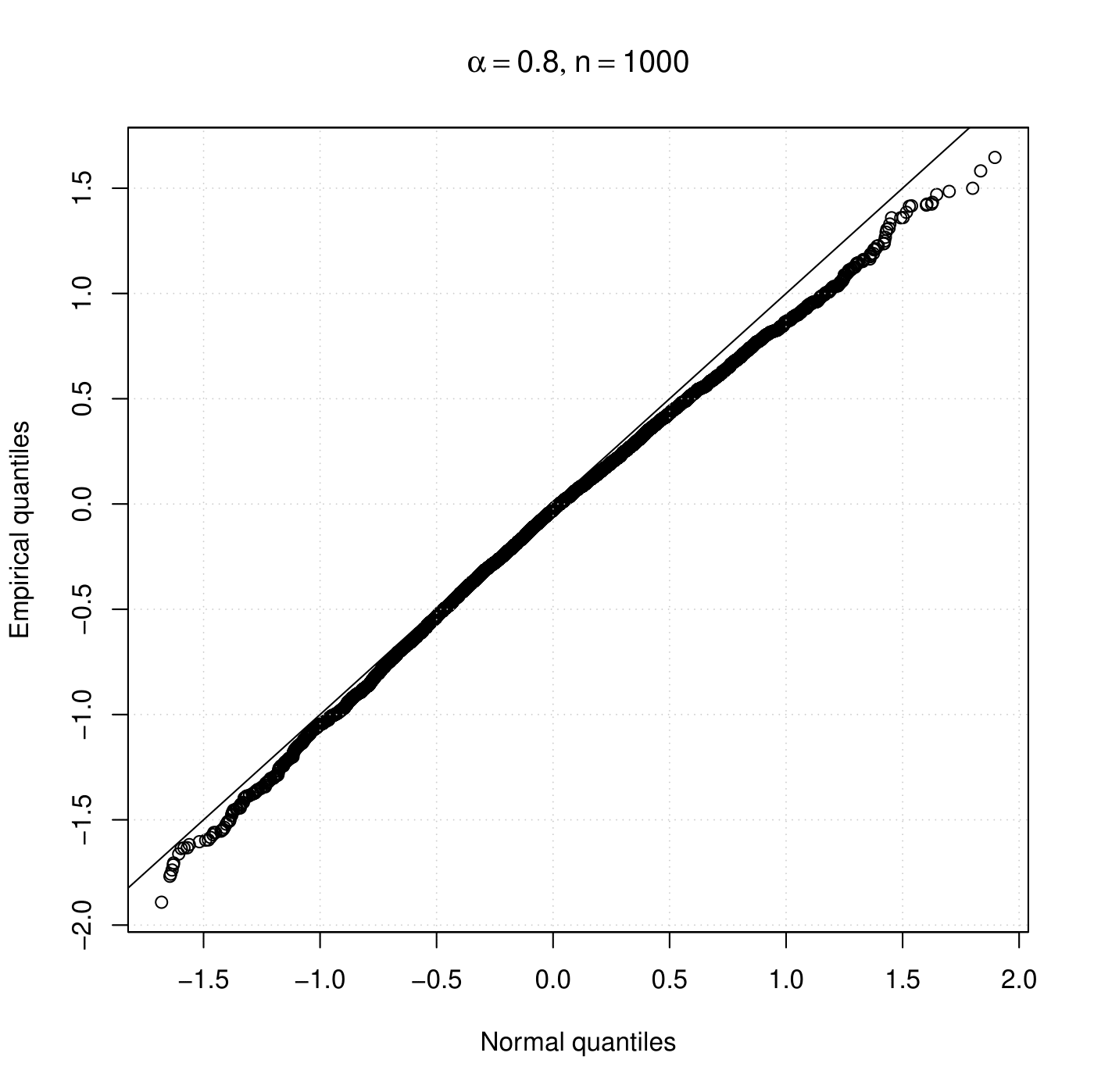}\includegraphics[width=5.4cm]{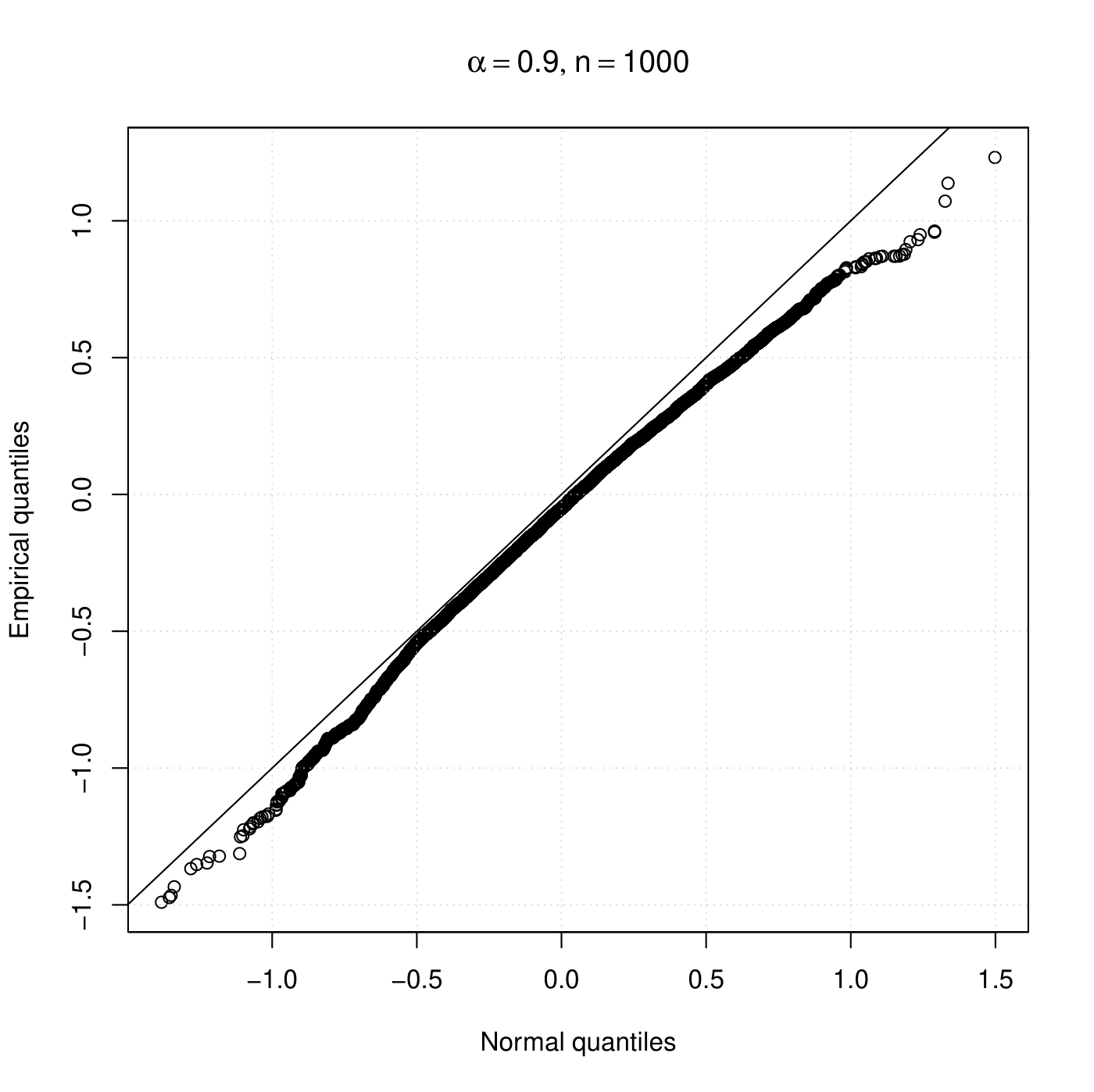}  	
		\caption{Histograms and qq-plots of the standardized estimates $\sqrt{n}(\widehat\alpha_n-\alpha)$ for $\alpha=0.7$, $\alpha=0.8$, and $\alpha=0.9$, along with their associated limiting normal density/quantiles given in Theorem \ref{weak_limit_stat} (under stationarity). The sample size is $n=1000$.}\label{fig:STAT_alpha_n1000II}
		\includegraphics[width=5.4cm]{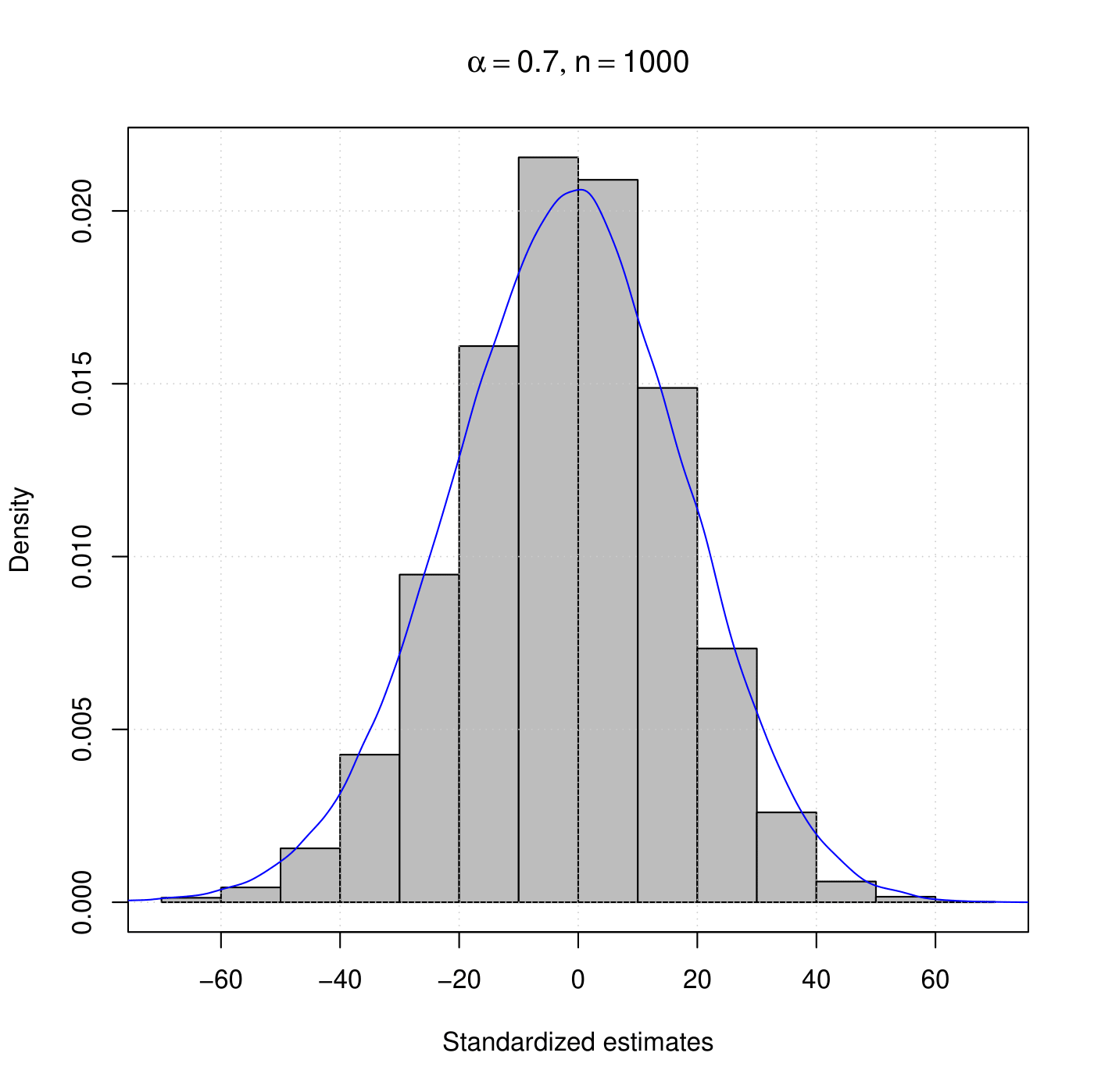}\includegraphics[width=5.4cm]{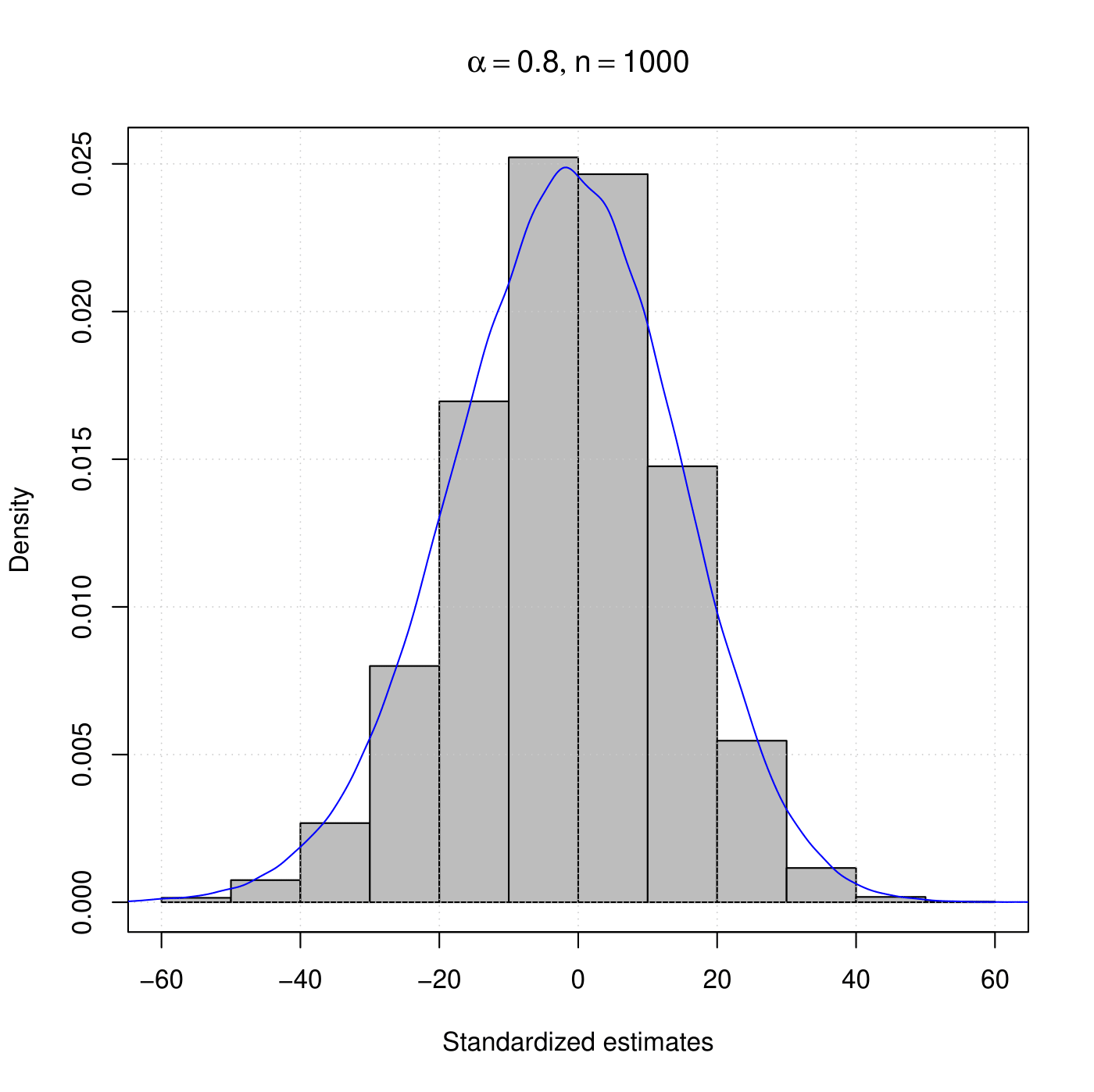}\includegraphics[width=5.4cm]{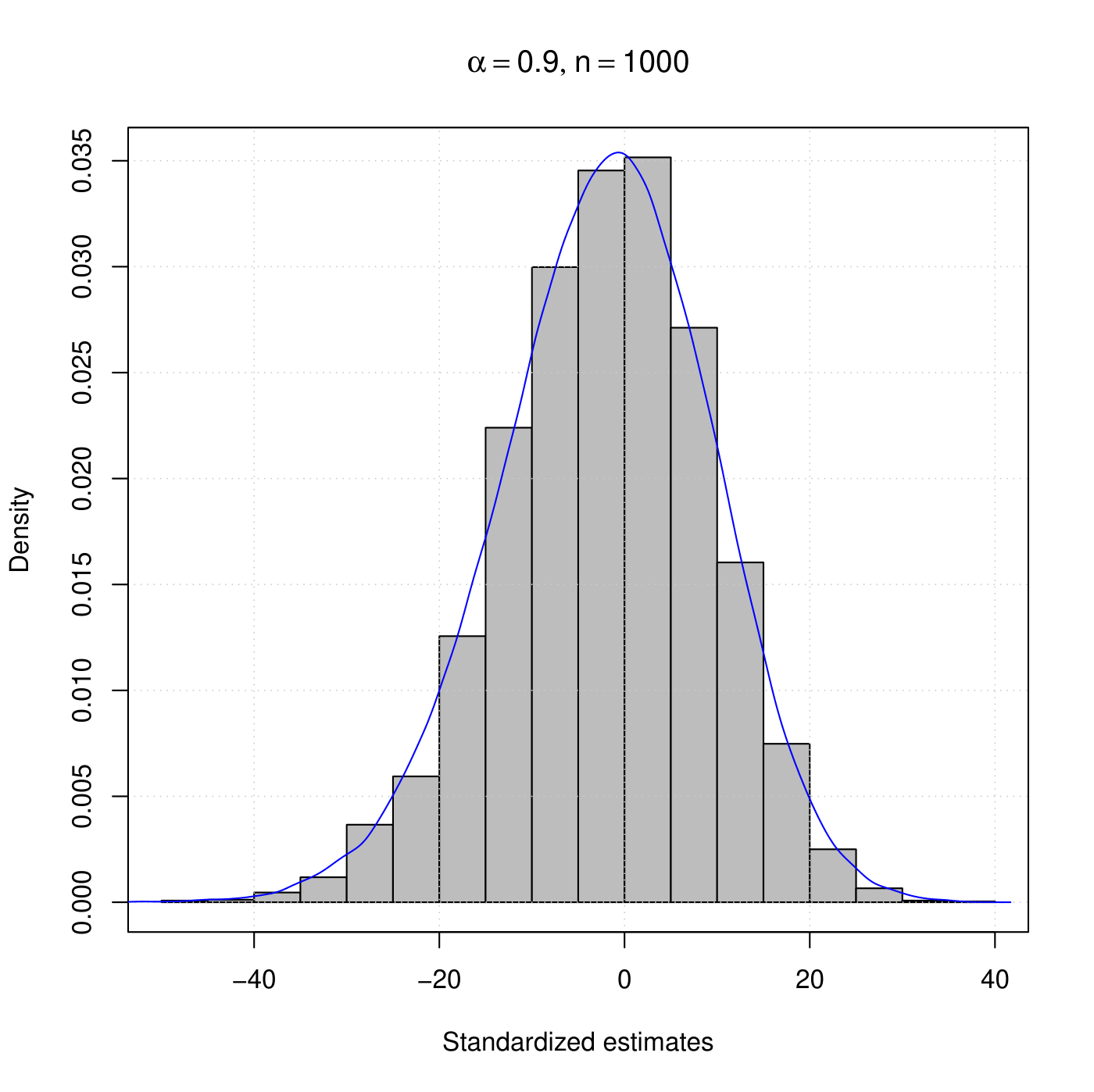}
		\includegraphics[width=5.4cm]{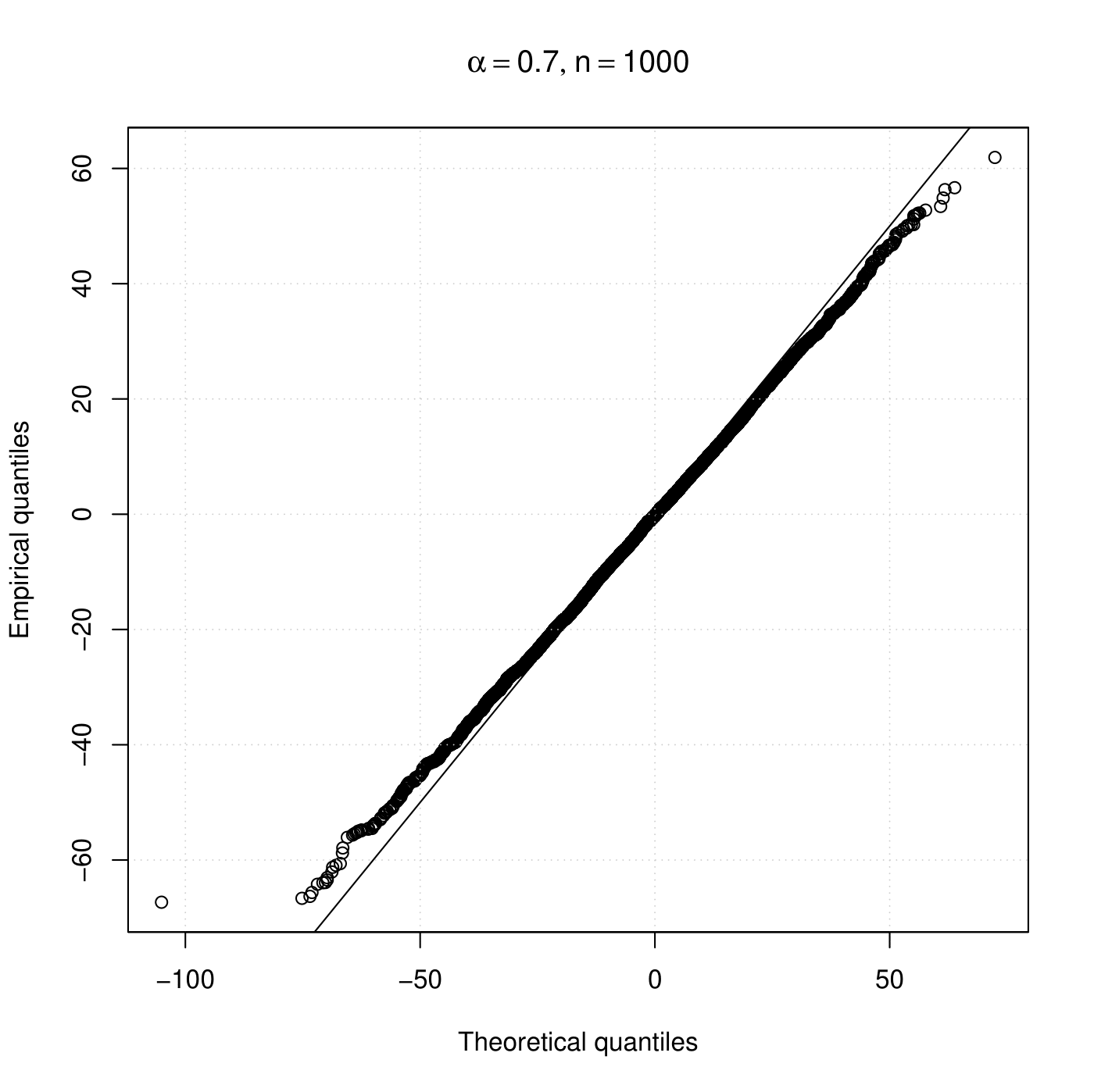}\includegraphics[width=5.4cm]{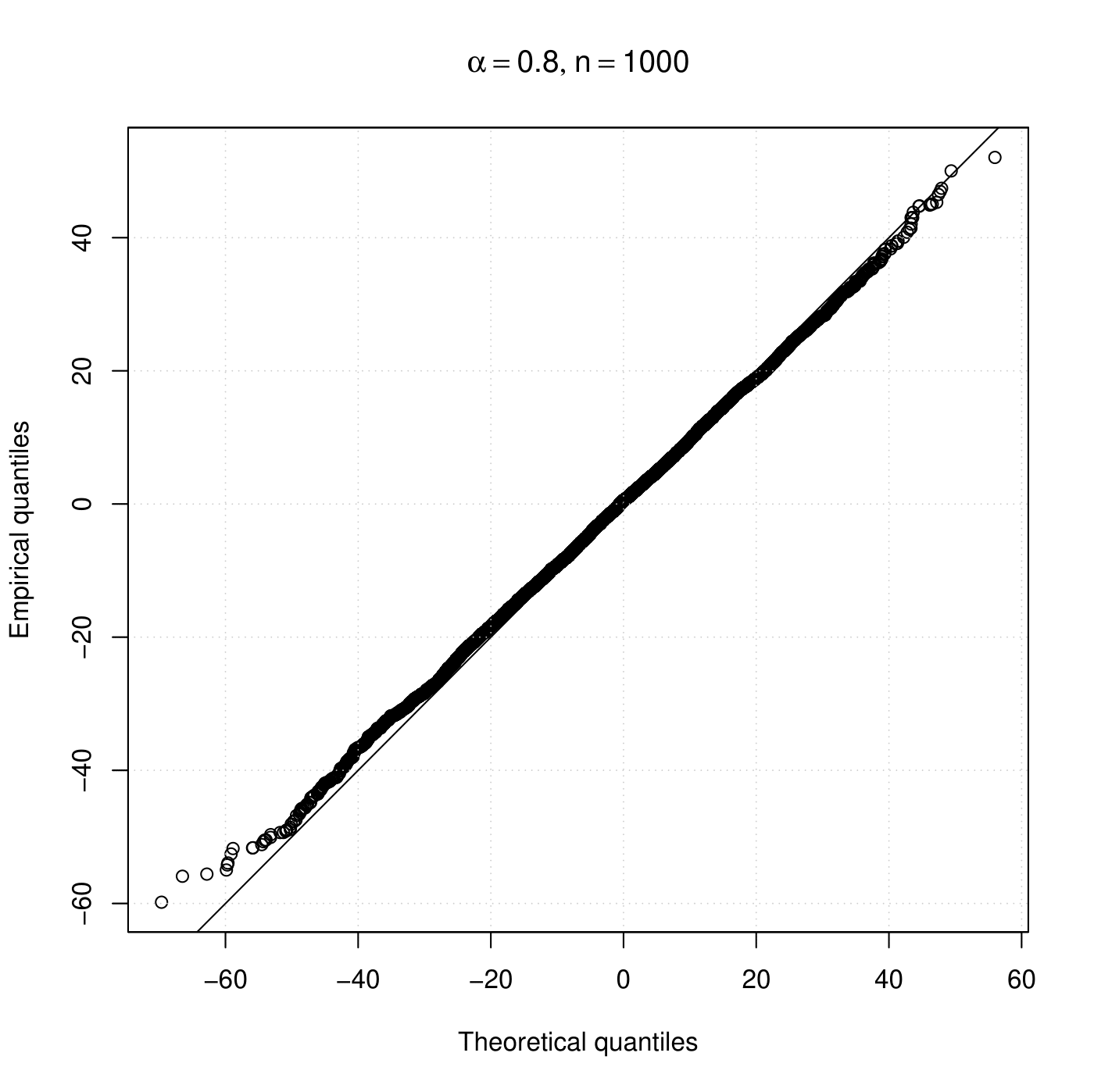}\includegraphics[width=5.4cm]{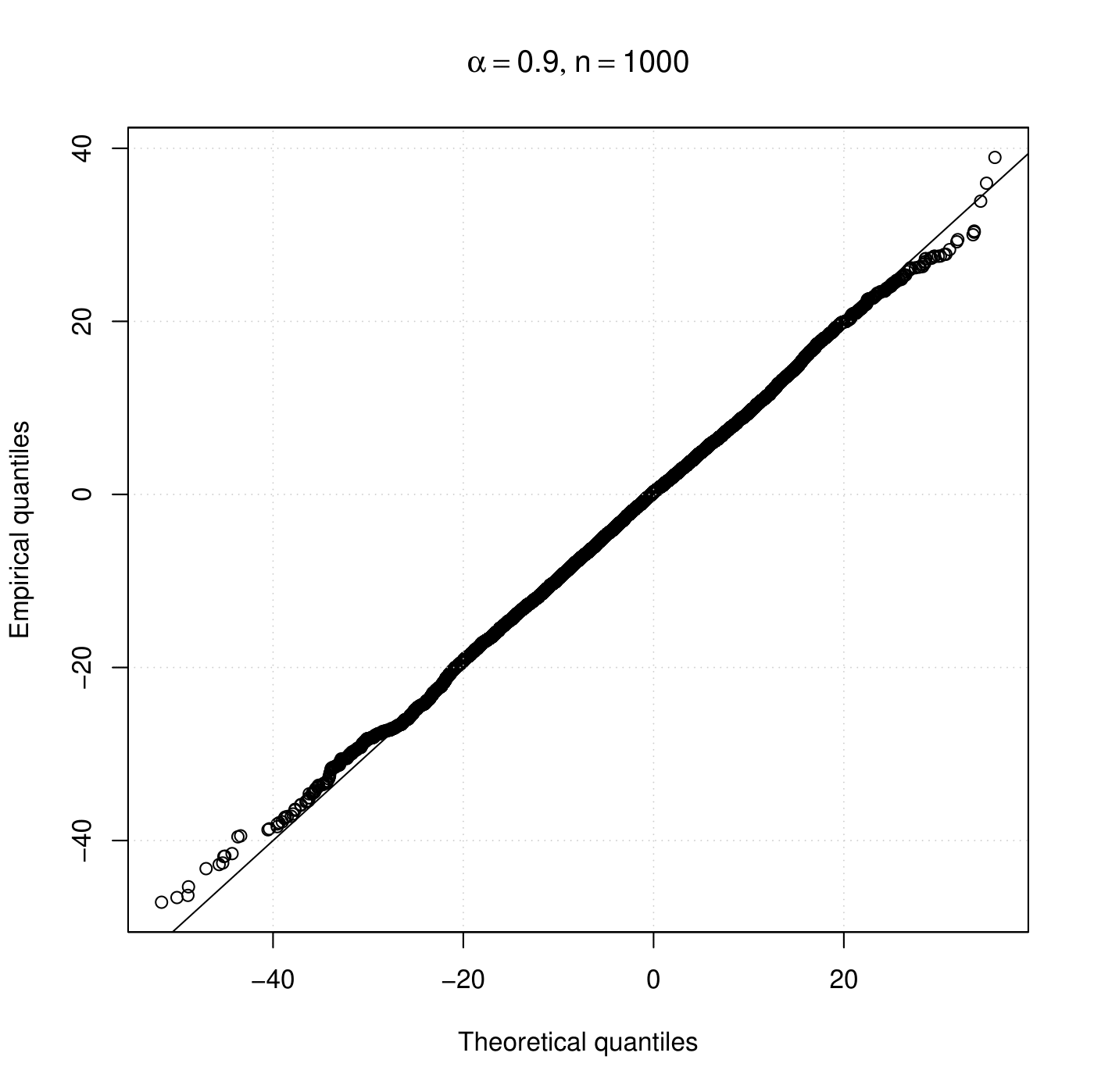}  	
		\caption{Histograms and qq-plots of the standardized estimates $n(\widehat\alpha_n-\alpha)$ for $\alpha=0.7$, $\alpha=0.8$, and $\alpha=0.9$, along with their associated limiting density/quantiles given in Theorem \ref{weak_limit_nonstat} (under nearly non-stationarity). The sample size is $n=1000$.}\label{fig:NONSTAT_alpha_n1000II}
	\end{center}
\end{figure}

Our interest now is to evaluate the coverages of the confidence intervals based on the asymptotic results under the nearly unstable and stable assumptions. In Table \ref{Table:CI}, we provide the empirical coverages of confidence intervals, from a Monte Carlo simulation with 10000 replications, for $\alpha$ with significance level at 10\%, 5\%, and 1\% based on Theorem \ref{weak_limit_nonstat} (under nearly non-stationarity). The sample size is  $n=500$ and we consider $\alpha=0.999,0.99,0.98,0.9,0.8,0.7$. These results show that inference on the correlation parameter using our methodology is satisfactory since the coverages are close to the nominal levels for all cases considered, even when $\alpha$ is not close to the non-stationarity region.

\begin{table}
	\centering
	\begin{tabular}{ccccccc}
		\hline
		$\alpha\rightarrow$ & 0.999 & 0.99 & 0.98 & 0.9 & 0.8 & 0.7 \\
		\hline
		90\% & 0.934 & 0.917 & 0.897 & 0.916 & 0.920 & 0.915 \\
		95\% & 0.967 & 0.952 & 0.939 & 0.960 & 0.968 & 0.966 \\
		99\% & 0.989 & 0.984 & 0.982 & 0.994 & 0.993 & 0.990\\
		\hline
	\end{tabular}
\caption{Empirical coverages of the 90\%, 95\%, and 99\% confidence intervals for $\alpha$ based on the nearly unstable approach. Sample size $n=500$.}\label{Table:CI} 
\end{table}

\section{Unit Root Test}\label{sec:urt}

In this section, we propose a statistical procedure for testing unit root in a Poisson INARCH(1) model with correlation parameter $\alpha$. The null and alternative hypotheses are respectively $\texttt{H}_0: \alpha=1$ and $\texttt{H}_1: \alpha<1$. To this end, we consider the nearly unstable approach and the statistic $n(\widehat\alpha_n-1)$, which is inspired by the traditional unit root test for the continuous AR(1) model of \cite{dicful1979}. Under the conditions of Theorem \ref{weak_limit_nonstat}, we have that
\begin{eqnarray}\label{AD_URT}
n(\widehat\alpha_n-1)=n(\widehat\alpha_n-\alpha_n)-\gamma_n\stackrel{d}{\longrightarrow}\mathcal D_\gamma-\gamma,
\end{eqnarray}
as $n\rightarrow\infty$, where $\mathcal D_\gamma$ is a random variable (depending on the parameter $\gamma$) following the asymptotic distribution given in the right-hand side of (\ref{cls_weak_conv}). We can approach the null hypothesis of interest through our methodology by taking $\gamma\rightarrow0$. In this case, the distribution of the right-hand side of (\ref{AD_URT}) approaches that of $\mathcal D_0$, which has the associated $\mathcal X$ process satisfying the stochastic differential equation $d\mathcal X(t)=\beta t+\sqrt{\mathcal X(t)}dB(t)$, $t\geq0$. Denote by $q_\zeta$ the $\zeta$-quantile of the distribution of $\mathcal D_0$, for $\zeta\in(0,1)$, that is $P(\mathcal D_0\leq q_\zeta)=\zeta$. These quantiles can be obtained from Monte Carlo simulation as done in Section \ref{sec:sim}.
Based on the above discussion, we propose the following decision rule for testing $\texttt{H}_0: \alpha=1$ against 
$\texttt{H}_1: \alpha<1$ with significance level at $\zeta\times100\%$:
\begin{itemize}
	\item Reject $\texttt{H}_0$ in favor of $\texttt{H}_1$ if $n(\widehat\alpha_n-1)<q_\zeta$.
\end{itemize}

To evaluate the finite-sample performance of the proposed unit root test (URT), we run a Monte Carlo simulation with 10000 replications. We set $\beta=1$ and sample sizes $n=50,80,100,200,300,400,\\500,1000,2000,5000$. In Table \ref{Table:SigLevels}, we provide the empirical significance levels with nominal levels at 10\%, 5\%, and 1\%. We observe that the URT is yielding the desired Type-I error even for small sample sizes (for instance, $n=50,80$). 

Aiming at the investigation of the test power, another Monte Carlo simulation is considered under the same setup as before and with a significance level at 5\%. We consider $\alpha=0.999,0.99,0.98,0.95,0.9,0.8,0.7$ and compute the proportion of rejections of the null hypothesis in each scenario. The results are presented in Table \ref{Table:Power}. As expected, the power increases when either we are going away from the null hypothesis or the sample size increases. The proportions of rejections given in that table indicate that the proposed URT is working satisfactorily and offers a promising device in dealing with count time series based on the INGARCH approach.

\begin{table}
	\centering
	\begin{tabular}{ccccccccccc}
		\hline
		 $n\rightarrow$ & 50 & 80 & 100 & 200 & 300 & 400 & 500 & 1000 & 2000 & 5000 \\
		 \hline
		10\% & 0.116 & 0.104 & 0.109 & 0.101 & 0.105 & 0.104 & 0.103 & 0.101 & 0.103 & 0.104 \\
		5\% & 0.061 & 0.056 & 0.057 & 0.054 & 0.055 & 0.054 & 0.049 & 0.054 & 0.054 & 0.049 \\
		1\% & 0.014 & 0.013 & 0.014 & 0.012 & 0.010 & 0.010 & 0.010 & 0.011 & 0.010 & 0.010\\
		\hline
	\end{tabular}\caption{Empirical significance levels obtained from a Monte Carlo study to evaluate the proposed unit root test under some sample sizes and nominal significance levels at 10\%, 5\%, and 1\%.}\label{Table:SigLevels}  
\end{table}

\begin{table}
		\centering
	\begin{tabular}{cccccccc}
		\hline
		$n\downarrow$\,\,\,$\alpha\rightarrow$ & 0.999 & 0.99 & 0.98 & 0.95 & 0.9 & 0.8 & 0.7 \\
		\hline
		50 & 0.057 & 0.085 & 0.115 & 0.258 & 0.623 & 0.983 & 1.000 \\
		80 & 0.114 & 0.173 & 0.264 & 0.618 & 0.973 & 1.000 & 1.000 \\
		100 & 0.166 & 0.270 & 0.413 & 0.865 & 1.000 & 1.000 & 1.000 \\
		200 & 0.216 & 0.425 & 0.695 & 0.998 & 1.000 & 1.000 & 1.000 \\
		300 & 0.265 & 0.614 & 0.930 & 1.000 & 1.000 & 1.000 & 1.000 \\
		400 & 0.314 & 0.798 & 0.995 & 1.000 & 1.000 & 1.000 & 1.000 \\
		500 & 0.367 & 0.927 & 1.000 & 1.000 & 1.000 & 1.000 & 1.000 \\
		1000 & 0.434 & 0.998 & 1.000 & 1.000 & 1.000 & 1.000 & 1.000 \\
		2000 & 0.551 & 1.000 & 1.000 & 1.000 & 1.000 & 1.000 & 1.000 \\
		5000 & 0.839 & 1.000 & 1.000 & 1.000 & 1.000 & 1.000 & 1.000\\
        \hline
       \end{tabular}\caption{Empirical power obtained from a Monte Carlo study to evaluate the proposed unit root test under some sample sizes and values of $\alpha$. Significance level at 5\%.}\label{Table:Power}
\end{table}

\section{Real Data Application}\label{sec:application}

We here apply the proposed methodology to the daily number of deaths due to COVID-19 in the United Kingdom from January 30, 2020, to June 4, 2021, so yielding $n=492$ observations. This dataset is publicly available at the site \url{https://coronavirus.data.gov.uk}. The plot of the daily number of deaths and its associated ACF are provided in Figure \ref{fig:data}, which reveals a nearly unstable/non-stationary behavior.

\begin{figure}
	\begin{center}
		\includegraphics[width=9cm]{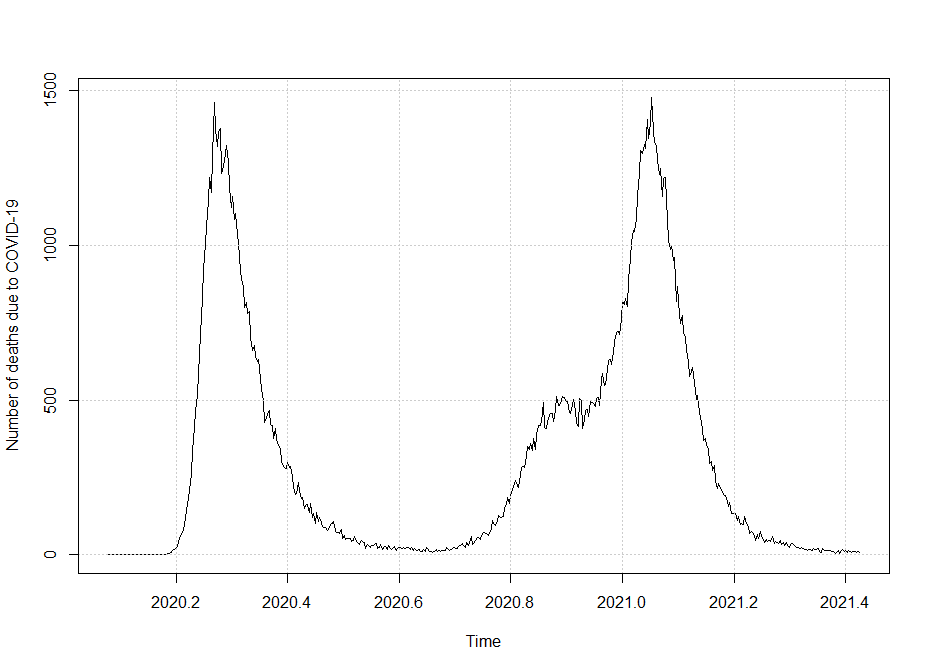}\includegraphics[width=9cm]{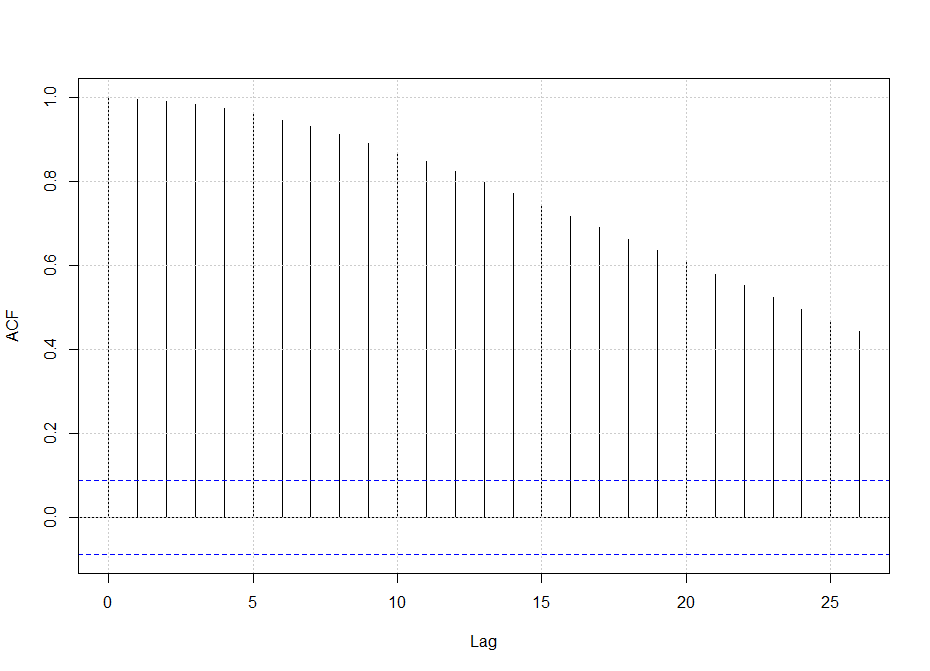}  	
		\caption{Plot of the daily number of deaths due to COVID-19 in UK and its associated ACF.}\label{fig:data}
	\end{center}
\end{figure}

We assume that the time series comes from an NU-INARCH(1) process. The aim of this application is to illustrate that the theoretical results found in this paper can reveal the unit root behavior for a real dataset. We first need to deal with $\beta$, which is unknown and can be seen as a nuisance parameter; our primary interest in this paper relies on the correlation parameter $\alpha_n$. One strategy is to estimate $\beta$ through the conditional maximum likelihood method, which consists in maximizing $\ell\propto\sum_{t=2}^n(y_t\log\lambda_t-\lambda_t)$, and then assume it known in what follows. This procedure gives $\beta=0.269$. 
At the end of this application, we will evaluate such an approach by performing a small Monte Carlo simulation study.

Using (\ref{cls}), we obtain the estimate for the correlation parameter equal to $\widehat\alpha_n=0.997$, which is very close to 1. We obtain the standard error of the $\alpha_n$ estimate ($\mbox{s.e.}(\widehat\alpha_n)$) using the asymptotic distribution stated in Theorem \ref{weak_limit_stat}, which gives the $\mbox{s.e.}(\widehat\alpha_n)\approx0.014$. We perform the URT proposed in Section \ref{sec:urt} for testing the hypothesis $\texttt{H}_0: \alpha=1$ against $\texttt{H}_1: \alpha<1$. We obtain $n(\widehat\alpha_n-1)= -1.257>-17.952=q_{0.05}$ and therefore we do not reject the null hypothesis on the unit root with significance level at 5\%. The density function of $\mathcal D_0$ based on Gaussian kernel and 100000 Monte Carlo replications is provided in Figure \ref{fig:URT} along with vertical lines denoting the statistic test and the $0.05$-quantile (of the $\mathcal D_0$ distribution). The associated $p$-value is $0.704$, which shows that we obtain the same indication by using any usual significance level.

\begin{figure}
	\begin{center}
		\includegraphics[width=10cm]{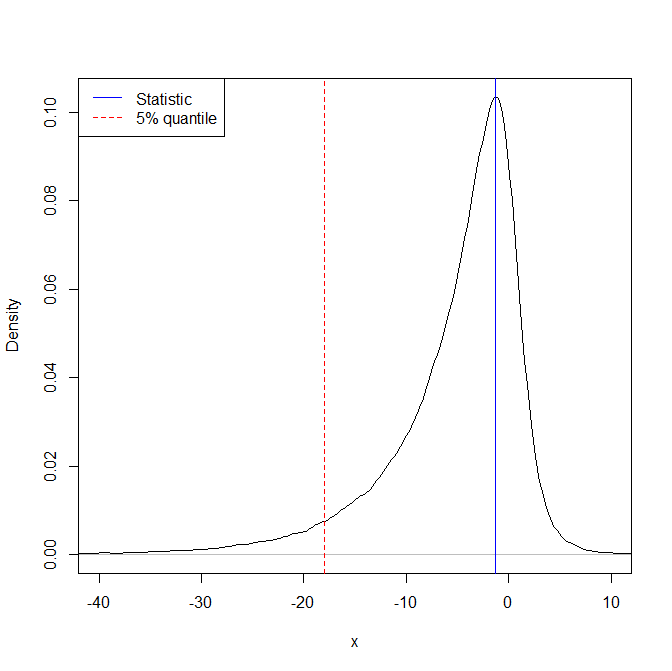}	
		\caption{Density (based on Gaussian kernel) of $\mathcal D_0$ using 100000 Monte Carlo replications. Vertical solid and dashed lines represent the statistic and the $0.05$-quantile (of the $\mathcal D_0$ distribution), respectively.}\label{fig:URT}
	\end{center}
\end{figure}

In Figure \ref{fig:pred}, we present the count time series data and the predicted means based on the fitted NU INARCH model, which reveals a good agreement between the observed time series and the model.

\begin{figure}
	\begin{center}
		\includegraphics[width=9cm]{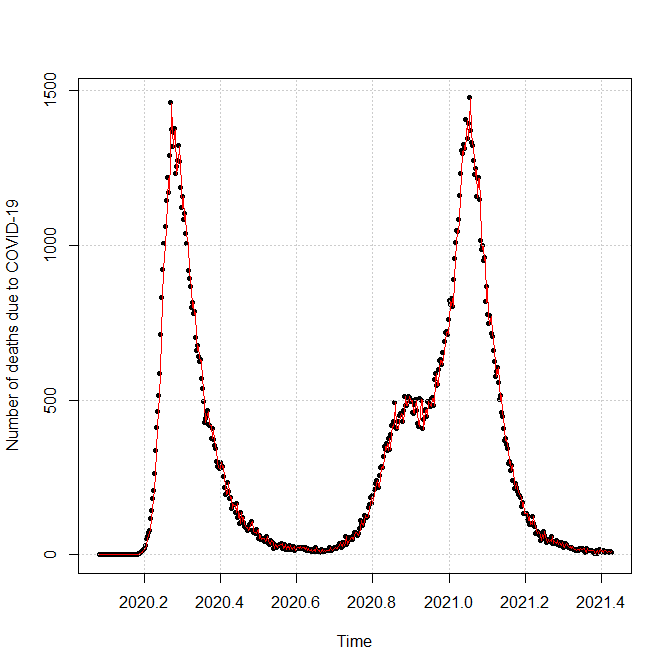}\includegraphics[width=9cm]{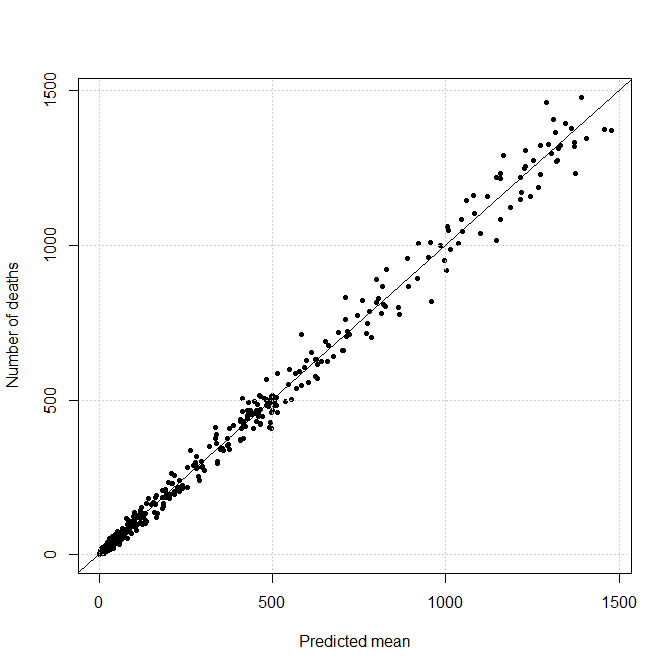}  	
		\caption{Left: Number of deaths (points) and the predicted mean $\widehat{E}(Y_t|Y_{t-1})=\beta+\widehat{\alpha}_nY_{t-1}$ (solid line). Right: Number of deaths against predicted means.}\label{fig:pred}
	\end{center}
\end{figure}

We conclude this application by evaluating our strategy by estimating $\beta$ and assuming known. To do this, we run a small Monte Carlo simulation with 1000 replications. In each loop, we generate an NU-INARCH model with $\beta=0.269$, $\alpha=0.997$, and $n=492$ (specifications of the application), construct confidence intervals for $\alpha$ based on both approaches with fixed and non-fixed (estimated as done in this section and then assumed known) $\beta$, and check if they contain the "true" value. The empirical coverages of the 90\%, 95\%, and 99\% confidence intervals under both approaches are reported in Table \ref{Table:coverage_beta_unknown}. As can be seen from this table, the proposed solution given here in the application provides the expected nominal coverages and works even better than the fixed $\beta$ case for the 90\% and 95\% coverages; the 99\% coverages are very close to each other.

\begin{table}
	\centering
	\begin{tabular}{ccc}
		\hline
		& Fixed $\beta$ & Non-fixed $\beta$ \\
		\hline
		90\% & 0.927 & 0.911 \\
		95\% & 0.960 & 0.948 \\
		99\% & 0.986 & 0.979 \\
		\hline
	\end{tabular}\caption{Empirical coverages of 90\%, 95\%, and 99\% confidence intervals for the correlation parameter $\alpha$ based on a Monte Carlo simulation under the NU-INGARCH model with the settings $\beta=0.269$, $\alpha=0.997$, and $n=492$. Both approaches with fixed and non-fixed (estimated) $\beta$ are reported.}\label{Table:coverage_beta_unknown}
\end{table}

\section{Discussion and Future Research}\label{conclusions}

A nearly unstable INARCH(1) process was introduced and weak convergence of a normalized version was established. The asymptotic distribution of the CLS estimator of the correlation parameter was derived under both nearly unstable and stable cases, which have been explored via Monte Carlo simulations. We also proposed a unit root test and checked its performance in terms of yielding the desired Type-I error and power through simulation. The nearly unstable INARCH approach was applied to the daily number of deaths due to the COVID-19 in the UK, which exhibits a non-stationary behavior. the proposed URT has provided evidence for the existence of a unit root in agreement with the descriptive analysis. 

We have assumed that the conditional distribution in (\ref{eq1}) is Poisson, but the methods presented in this paper can be easily adapted for other distributional assumptions such as negative binomial or more generally mixed Poisson distributions, among others. More specifically, the very same strategy given in Proposition \ref{prop:moments} and Lemma \ref{lem:cov_M} can be employed to find the proper normalizations for the processes $\{X^{(n)}_{\floor{nt}},\,\,t\geq0\}$ and $\{W^{(n)}(t),\,\,t\geq0\}$ in these other cases. After obtaining these results, the asymptotic distributions of the normalized count process and CLS estimator are established following the same steps as those given in Theorems \ref{mainthm} and \ref{weak_limit_nonstat}, respectively. We also believe that extending the results for higher-order INGARCH models deserves future investigation.

\section*{Acknowledgments}  \noindent Research supported in part by grants from the KAUST and NIH 1R01EB028753-01 (W. Barreto-Souza), HKSAR-RGC-GRF No. 14325216 and the Theme-based Research Scheme of HKSAR-RGC-TBS T32-101/15-R (N.H. Chan).

\end{document}